\def\showtableofcontents{1}
\newif\ifnotes
\notesfalse

\documentclass[11pt,letterpaper]{article}

\ifnotes
  \usepackage[draft,notes=true,later=true]{dtrt}
\else
  \usepackage[notes=false,later=false]{dtrt}
\fi
\usepackage[T1]{fontenc} %
\usepackage{lmodern} %
\usepackage[utf8]{inputenc} %
\usepackage{cmap} %
\usepackage{multirow}
\usepackage{booktabs}
\usepackage{float}
\usepackage{amsmath,amsthm,amssymb,algpseudocode,algorithm,cryptocode}
\usepackage{bm} %
\usepackage{graphicx}
\usepackage{fullpage}
\usepackage{caption}
\usepackage{tikz}
\usetikzlibrary{angles,quotes,quantikz}
\usepackage{physics}
\usepackage{mathtools}
\usepackage[capitalize,nameinlink]{cleveref}
  \crefname{step}{Step}{Steps}
  \crefname{claim}{Claim}{Claims}
  \crefname{definition}{definition}{definitions}
  \Crefname{definition}{Definition}{Definitions}
\hypersetup{colorlinks={true},linkcolor={blue},urlcolor={red},citecolor={blue}}
\usepackage{xspace} %
\usepackage{paralist}
\usepackage{enumitem}
\usepackage{mdframed}
\usepackage{bbm}
\usepackage{xcolor}
\allowdisplaybreaks

\newtheorem{theorem}{Theorem}[section]
\newtheorem{proposition}[theorem]{Proposition}
\newtheorem{claim}[theorem]{Claim}
\newtheorem{lemma}[theorem]{Lemma}

\newtheorem{corollary}[theorem]{Corollary}

\newtheorem{remark}[theorem]{Remark}

\theoremstyle{definition}
\newtheorem{definition}[theorem]{Definition}

\theoremstyle{remark}

\def\poly{{\rm poly}}
\def\negl{{\rm negl}}
\newcommand{\Win}{\mathrm{Win}}

\let\oldket\ket
\renewcommand{\ket}[1]{\oldket*{#1}}
\let\oldketbra\ketbra
\renewcommand{\ketbra}[1]{\oldketbra*{#1}}

\newcommand{\cA}{\mathcal{A}}

\newcommand{\cK}{\mathcal{K}}

\newcommand{\wt}{\mathrm{wt}}

\newcommand{\Id}{\mathbf{I}}

\renewcommand{\epsilon}{\varepsilon}

\newcommand{\cO}{\mathsf{cO}}

\crefname{equation}{}{}

\usepackage{titling}

\title{Lower Bounds for Preprocessing Attacks \\ on Quantum Cryptography}
\date{October 3, 2026}
\author{Fangqi Dong \\ Princeton \and Alex Lombardi \\ Princeton}

\begin{document}

\maketitle
\thispagestyle{empty} 
\begin{abstract}
    
We prove near-optimal lower bounds for preprocessing attacks on quantum cryptography in the random oracle model. Specifically, we prove that a $T$-query adversary with $S$ qubits of non-uniform advice can recover a random key $k$ from the $n$-qubit binary phase state $\ket{\psi_k} \propto \sum_{x} R(k,x) \ket{x}$ with probability at most $O(\frac{T^2 + \sqrt{ST}}{N})$ for $N=2^n$. In contrast, the best known bound for post-quantum one-way functions is $O(\frac{T^2 + ST}{N})$, and indeed there is a trivial attack with $S = N$. This demonstrates a new advantage of quantum cryptography over classical cryptography: $n$ qubits of communication suffice for security against preprocessing attacks with space up to $N^2$ rather than $N$. 

We achieve this result by making use of a simple methodology: express the optimal preprocessing attack as the operator norm of a random matrix, and bound this value in expectation over the random oracle via the trace-moment method. Bounding these trace moments turns out to have a natural interpretation using the compressed oracle methodology [Zhandry, Crypto 2019], which we then analyze. This can be viewed as a simplification and generalization of the approach of Liu [Eurocrypt 2023] for proving the security of post-quantum cryptography against preprocessing attacks. 

We use this methodology to prove several additional results.
\begin{itemize}
    \item We tighten Liu's analysis of post-quantum PRGs in the QROM, achieving a distinguishing advantage bound of $O(\frac{T^2}N + \sqrt{\frac{ST}N})$.
    \item We apply this methodology in the context of \emph{unitary synthesis}, extending the one-query lower bound of Lombardi-Ma-Wright [STOC 2024] to hold against adversaries that can make one arbitrary function query \emph{along with} polynomially many (adaptive) queries to the random oracle, either before or after the function query. This also gives a way to interpret the original LMW24 result in terms of compressed oracles. 
    \item Finally, we prove a tight $O(\frac{\sqrt{S}}N)$ bound for the \emph{pseudorandomness} of random binary phase states against space $S$ distinguishers.
\end{itemize}
     
\end{abstract}

\ifnum\showtableofcontents=1
{
\newpage
\setcounter{tocdepth}{2}
\tableofcontents
\thispagestyle{empty}
}
\newpage
\fi

\setcounter{page}{1}
\section{Introduction}
\label{sec:introduction}

We study the security of quantum cryptography against preprocessing attacks with time-space complexity $(T,S)$ in the quantum random oracle model~\cite{AC:BDFLSZ11}. In a preprocessing attack \cite{Hellman80,STOC:FiaNao91,C:Unruh07}, the adversary may perform an unbounded amount of work before receiving its challenge, storing information in the form of an $S$-bit (or $S$-qubit) advice string (or quantum state). It then receives the challenge and
makes $T$ queries to the random oracle as usual. The question is to understand the resulting time-space security tradeoffs: given $S$ qubits of advice, how many queries are necessary to break cryptosystems of interest? 

The study of preprocessing attacks in classical cryptography dates back to the 1980s \cite{Hellman80,STOC:FiaNao91}, where Hellman showed how to break the one-wayness of an $n$-bit random oracle with success probability $O\Big((\frac{S^2 T}{N^2})^{1/3}\Big)$, i.e., a complete break with $N^{2/3}$ time and advice length. Famously, there is no matching lower bound: while a sequence of recent works \cite{C:Unruh07,EC:DodGuoKat17,EC:CDGS18} (building on a seminal work of Yao \cite{STOC:Yao90}) establish that algorithms with time-space complexity $(T,S)$ achieve inversion probability at most $O(\frac{ST}N)$, it is known that sharpening the tradeoff $ST = \Omega(N)$ would imply new circuit lower bounds \cite{TCC:CorKog19}.

Preprocessing attacks have also received significant attention in the \emph{post-quantum setting} \cite{AC:HhaXagYam19,ITC:ChuLiaQia20,FOCS:CGLQ20,TCC:GLLZ21,EC:Liu23,C:DonLiuWu24}, where we study the security of classical cryptosystems against quantum attacks. In this setting, quantum attackers have two notable sources of additional power: an adversary can query the oracle in superposition, and it can store a \emph{quantum state} during the preprocessing phase.

For our two most foundational symmetric-key cryptographic primitives, one-way functions and pseudorandom generators, the state-of-the-art security bounds in the aux-input QROM are due to Liu \cite{EC:Liu23}, who showed that a random oracle can be inverted (as a one-way function) with probability at most
\[
    \epsilon_{\mathrm{OWF}}\leq O\!\left(\frac{T^2+ST}{N}\right),
\]
while a random oracle can be broken as a PRG with distinguishing advantage at most 
\[
    \epsilon_{\mathrm{PRG}}\leq O\!\left(\frac{T}{\sqrt{N}} + \Big(\frac{ST}N\Big)^{1/3} \right). 
\]

These results concern the security of \emph{classical} primitives against quantum preprocessing attacks. In this work, we study the analogous problem in the context of \emph{fully quantum cryptography}, where the cryptosystem itself involves quantum computation and communication. In particular, we focus on \emph{single-copy} one-way states \cite{TQC:MY24} and pseudorandom states \cite{C:JiLiuSon18}, natural lightweight forms of quantum cryptography that can be used to build quantum bit commitments \cite{C:MorYam22} and therefore other quantum cryptographic primitives such as oblivious transfer, zero-knowledge proofs, and secure computation \cite{C:BCKM21b,EC:GLSV21}.

We therefore ask:

\begin{center}
    \emph{What is the achievable quantitative security of these quantum cryptographic primitives? }
\end{center}

To study this question in the random oracle model, we consider the standard instantiation using binary phase states. For a random oracle $R: [K] \times [N] \rightarrow \{\pm 1\}$, we define the family of pure states
\[ \ket{\psi_{R, k}} = \frac 1 {\sqrt N} \sum_{x\in [N]} R(k,x) \ket{x}.
\]

For one-wayness, we ask that it is hard to recover a random key $k$\footnote{In this paper, we consider the regime $K\leq N$; while this inversion problem is not information-theoretically solvable with probability $1$, it is solvable with high probability when $K \leq cN$ for a sufficiently small constant $c$, while it is solvable with constant probability when $K=N$.} from the state $\ket{\psi_{R,k}}$, while for pseudorandomness (with $K < N/2$) we ask that the mixed state $\mathbb E_k \ketbra{\psi_{R,k}}$ is computationally indistinguishable from maximally mixed. 

Interestingly, there is also a fascinating connection between these cryptographic primitives and the \emph{unitary synthesis problem} \cite{AK07,STOC:LomMaWri24,DLM26}; indeed, breaking the one-wayness or pseudorandomness of random binary phase states implies a unitary synthesis algorithm for an associated family of unitaries in the following hybrid query model:

\begin{itemize}
    \item The algorithm can make \emph{one} query to an arbitrary boolean function $f_R$ depending on $R$, and
    \item The algorithm can make up to $T$ (sequential) queries to the random oracle $R$ itself.
\end{itemize}
This is a hybrid model lying in between one-query unitary synthesis and full-fledged polynomial-query unitary synthesis. Note that just as one-query algorithms are more powerful than ``quantum programs'' ($0$-query algorithms with quantum advice), this hybrid model is more powerful than $(T,S)$-preprocessing algorithms. 

\subsection{Our Results}

For our first result, we prove a quantitatively strong security bound for OWS in the aux-input QROM for the binary phase state family $\ket{\psi_{R,k}}$ defined above. 

\begin{theorem}[Informal; see Theorem~\ref{thm:1ows-ts-tradeoff}]\label{thm:main-ows-intro}
    Any adversary with $S$ qubits of advice that makes $T$ queries to $R$ recovers the random key $k$ from $\ket{\psi_{R,k}}$ with probability at most
    \[
        O\!\left(\frac{1+T^2+\sqrt{S(T+1)}}{K}\right).
    \]
\end{theorem}
In particular, this bound simplifies to $O\!\left(\frac{1+T^2+\sqrt{S(T+1)}}{N}\right)$ when $K=N$. This bound is tight in two important settings:
\begin{itemize}
    \item $S=0$: in this case, the bound $O(T^2/N)$ matches a Grover-style search algorithm \cite{Rosenthal22}
    \item $T=0$: in this case, the bound $O(\sqrt{S}/N)$ matches (up to logarithmic factors) a folklore attack \cite{STOC:LomMaWri24} in which the preprocessing algorithm stores $O(\sqrt{S})$ copies of $\sqrt{S}$ of the $\ket{\psi_k}$ and applies a sequence of symmetric subspace tests to recover the key whenever the challenge key lies in the memorized set (see \cref{sec:symmetric-subspace-attack}). 
\end{itemize}
Whether the mixed term $\frac{\sqrt{ST}}{N}$ is tight remains an open problem; in particular, it would be interesting to understand whether there are barriers (like in the classical setting \cite{TCC:CorKog19}) to identifying this term.

We obtain a stronger bound when the advice is classical. Specifically, Theorem~\ref{thm:1ows-ts-tradeoff-classical} bounds the success probability of an adversary with $S$ bits of advice by
\[
    O\!\left(\frac{1}{K}+\frac{T^2}{K}+\frac{S(T+1)}{KN}\right).
\]
For $T=o(\sqrt K)$, this requires $ST=\Omega(KN)$ for constant success, compared with $ST=\Omega(K^2)$ for quantum advice. These thresholds coincide when $K=N$, but the classical advice bound is stronger for $K \ll N$. 

\paragraph{Advantage over classical cryptography.}
We particularly note that \cref{thm:main-ows-intro} demonstrates a fascinating advantage of quantum cryptography over classical cryptography: a cryptosystem with $n$ qubits of communication\footnote{For example, one can convert this OWS into a quantum bit commitment scheme with similar security with $n$-qubit commitments. See \cref{sec:commitments} for more details. } can have security against preprocessing attacks with space up to $S = N^2$ rather than $S=N$ in the classical setting! In general, we obtain a time-space security tradeoff beating the classical (post-quantum) setting for all $T \ll \sqrt{N}$. This identifies a new advantage of quantum cryptography over classical cryptography: the achievable security level for a fixed (e.g.) commitment length. 

\paragraph{Methodology and additional results.}

As we explain in the technical overview, we prove \cref{thm:main-ows-intro} using a simple methodology inspired by a combination of the prior work \cite{EC:Liu23} and recent results on unitary synthesis \cite{STOC:LomMaWri24,DLM26}:

\begin{itemize}
    \item express the success probability of the optimal preprocessing algorithm as the maximum eigenvalue of a random matrix defined in terms of $R$,
    \item bound this eigenvalue using the trace moment method
    \item bound the trace moments via a natural interpretation using the compressed oracle methodology \cite{C:Zhandry19}.
\end{itemize}  
This paradigm turns out to be quite powerful, leading to several additional results in this paper. 

\paragraph{Post-quantum PRGs.}
As a byproduct of this methodology, we obtain an improvement on the state-of-the-art security of PRGs in the aux-input QROM \cite{EC:Liu23}. 

\begin{theorem}[Informal; see \cref{thm:prg-ts-tradeoff}]\label{thm:prg-intro}
    Any $T$-query adversary with $S$ qubits of advice has distinguishing advantage at most
    \[
        O\!\left(\frac{T^2}{N}+\sqrt{\frac{ST}{N}}\right)
    \]
    for breaking a PRG with seed length $n = \log N$. 
\end{theorem}

As mentioned earlier, the prior known bound was $O(T/\sqrt N+(ST/N)^{1/3})$. In a nutshell, while \cite{EC:Liu23} consider a sequential repeated variant of PRG distinguishing, the operator norm method enables considering a ``centered'' random matrix governing the distinguishing advantage rather than the win probability; via the trace moment method, this enables reducing the bound to a variant of uniform PRG security.

\paragraph{Single-copy pseudorandom states.}
We also prove a tight bound for 1-copy PRS in the special case $T=0$. Specifically, we show (see \cref{thm:program-1prs}) that an adversary with $S$ qubits of advice has
\[
    \varepsilon_{\mathrm{PRS}}
    \leq O\!\left(\frac{\sqrt S}{K}\right)
    \qquad (T=0).
\]
Thus, in this setting, we obtain the same quantitative
security for pseudorandomness as for one-wayness. The symmetric subspace attack described above also gives a matching distinguishing attack up to
logarithmic factors.

In particular, \cref{thm:program-1prs} gives a formal hardness-of-approximation result for \emph{quantum programs} \cite{Yuen}, which are unitary synthesis algorithms that consist only of an advice state (and are otherwise independent of the unitary being implemented). Yang, Renner, and Chiribella~\cite{YRC20} proved that $2^{-\mathsf{poly}(n)}$-accurate quantum programs require advice states of length $\Omega(N^2)$; our result extends the tradeoff to rule out algorithms with low (or even just constant) fidelity (see \cite{STOC:LomMaWri24,DLM26} for formal definitions).

\cref{tab:preprocessing-comparison} summarizes the known time-space-probability tradeoffs for breaking OWFs, PRGs, 1OWS, and 1PRS, respectively. 
\begin{table}[H]
    \centering
    \small
    \setlength{\tabcolsep}{10pt}
    \renewcommand{\arraystretch}{1.5}
    \begin{tabular}{@{}lcl@{}}
        \toprule
        \textbf{Primitive} & \textbf{Security} & \textbf{Reference} \\
        \midrule
        Post-quantum OWFs
        & $\displaystyle \varepsilon\leq O\!\left(\frac{T^2+S(T+1)}{N}\right)$
        & \cite{EC:Liu23} \\[4pt]
        Post-quantum PRGs
        & $\displaystyle \varepsilon\leq O\!\left(\frac{T^2}{N}+\sqrt{\frac{S(T+1)}{N}}\right)$
        & \textbf{This work} \\[4pt]
        One-way states (1OWS)
        & $\displaystyle \varepsilon\leq O\!\left(\frac{T^2+\sqrt{S(T+1)}}{N}\right)$
        & \textbf{This work} \\[4pt]
        Pseudorandom states (1PRS)
        & $\displaystyle \varepsilon\leq O\!\left(\frac{\sqrt S}{N}\right)\quad(T=0)$
        & \textbf{This work} \\
        \bottomrule
    \end{tabular}
    \caption{Security bounds against $(T,S)$-preprocessing attacks in the aux-input QROM. We set $K = \Theta(N)$ in this comparison.}
    \label{tab:preprocessing-comparison}
    \label{tab:preprocessing-lower-bounds}
\end{table}

For the classical primitives, the OWF and PRG bounds have the same online query term
$T^2/N$, but the preprocessing term increases from $S(T+1)/N$ to
$\sqrt{S(T+1)/N}$. This is
inherent, as for $T=0$, a simple distinguisher\footnote{For a fixed oracle $G$, memorize the optimal statistical distinguisher between the distribution $G(k)|_{[1:\log S]}$) and the uniform distribution. This requires $S$ bits of advice and achieves expected distinguishing advantage $\Omega(\sqrt{S/N})$ for a random $G$. This is closely related to more sophisticated distinguishers studied in \cite{C:DeTreTul10}.} achieves advantage
$\Omega(\sqrt{S/N})$.

In contrast, at $T=0$, both our (tight) 1OWS and 1PRS
bounds are $O(\sqrt S/N)$; pseudorandomness does not incur an analogous square root loss. 

\paragraph{Unitary synthesis with random oracle queries}
Finally, we analyze the hybrid query unitary synthesis model described above: a 1OWS or 1PRS attacker that can make polynomially many queries to the random oracle $R$ as well as one \emph{arbitrary} function query $f_R$ depending on $R$. Without any
extra queries to $R$, the results of~\cite{STOC:LomMaWri24,DLM26} bound the
search success probability and distinguishing advantage by $
    O\!\left(\frac{\log(2M)}{K}\right)$ and $
    O\!\left(\sqrt{\frac{\log(2M)}{K}}\right)$, respectively, where $M$ is the workspace dimension. Using our methodology, we extend these results to the setting of additional adaptive queries to the random oracle.

\begin{theorem}[Informal; see \cref{thm:binary-phase-search-post,thm:binary-phase-distinguish-post}]
    Consider an adversary with workspace dimension $M$ that makes one
    query to an arbitrary Boolean function $f_R$, followed by $T$ queries
    to $R$. Its search success probability and distinguishing advantage are
    at most
    \[
        O\!\left(\frac{T^2+(T+1)\log(2M)}{K}\right)
        \qquad\text{and}\qquad
        O\!\left(\sqrt{\frac{T^2+(T+1)\log(2M)}{K}}\right),
    \]
    respectively.\footnote{We note that the distinguishing bound is likely not tight and can be improved to $O\big(\frac{T^2}K\big) + O\!\left(\sqrt{\frac{(T+1)\log(2M)}{K}}\right)$ by a similar trick used in the proof of \cref{thm:prg-intro}, but we have not worked out the details. }
\end{theorem}

We also obtain lower bounds when the adversary makes queries to $R$ both
before and after its query to $f_R$; see \cref{thm:1-query-pre-post-queries}. The quantitative bounds
are weaker in this case, but still give negligible success probability/advantage
when $K$ is exponential and the algorithm uses polynomially many queries
and workspace qubits.

Our proof also gives a new interpretation of the original
results of \cite{STOC:LomMaWri24,DLM26}. After applying their spectral
relaxation to remove the arbitrary function query, we can analyze the
resulting operator using the trace moment method (rather than appealing to a generic matrix concentration inequality). The key
point is that the rescaling matrices introduced by their spectral relaxation are
\emph{linear in the coefficients of $R$}. This is why a compressed oracle
appears even though the original algorithm queries an arbitrary function
of $R$, rather than $R$ itself. The same interpretation then accommodates
actual queries to $R$.

\subsection{Acknowledgements and AI use}
A.L. thanks Fermi Ma, John Wright, and John Bostanci for helpful conversations in the early stages of this project.

The research process leading to our main results included several interactive sessions with GPT-5.5 Pro and GPT-5.6 Sol during April-August 2026. The main ideas behind several of the results were provided by the authors, but ChatGPT was used to help work out proof details. \cref{thm:1ows-ts-tradeoff-classical} was one-shot proved by ChatGPT given our previously obtained results, and ChatGPT is primarily responsible for the proof of \cref{lem:program-local} (the main lemma used to prove \cref{thm:program-1prs}). All results and proofs have been post-processed by humans and the authors take full responsibility for the contents of this paper. 

GPT-5.6 Sol and GPT-6 Astra were used to help with the writing and preparation of this manuscript.

\section{Technical Overview}
\label{sec:technical-overview}

We now describe the techniques behind our results, focusing on the special case of \cref{thm:main-ows-intro}. We then describe some key modifications and innovations required to apply the technique to our other settings of interest. 

\subsection{One-way states with no online queries}
\label{subsec:overview-zero-query}

To show how we prove \cref{thm:main-ows-intro}, we primarily consider the special case of $T=0$; here, an adversary is given one copy of the state
\[ \ket{\psi_k} = \frac 1 {\sqrt N} \sum_{x\in [N]} R(k,x) \ket{x}
\]
as well as an $S$-qubit advice state $\ket{\phi_R}$ depending on $R$, and wants to recover the string $k$. We wish to prove that any such adversary succeeds with probability at most 
\[ O\Big(\frac{\sqrt{S}}K\Big).
\]
We note that in this $T=0$ special case, \cite{DLM26} already prove the suboptimal bound $O\big( \frac{S}K\big)$, but we will nevertheless use this case to demonstrate our new approach (and obtain the tight bound).

\paragraph{Operator norm captures advice.}
Fix the adversary's oracle-independent algorithm. Roughly following \cite{STOC:LomMaWri24,DLM26}, we can describe its computation as a fixed projective measurement $\{\Pi_k\}_{k\in[K]}$ on the
advice register $\mathsf{A}$, the challenge register $\mathsf{X}$, and an additional workspace register $\mathsf{W}$ initialized to $\ket{0}$. Then, using an arbitrary advice state
$\ket{\varphi}$, the adversary's success probability is given by
\[
    \mathbb E_k\,
    \bra{\varphi}_{\mathsf{A}}\bra{0}_{\mathsf{W}}\bra{\psi_{R,k}}_{\mathsf{X}} \Pi_k
    \ket{\varphi}_{\mathsf{A}}\ket{0}_{\mathsf{W}}\ket{\psi_{R,k}}_{\mathsf{X}}
    =\bra{\varphi}Y_R\ket{\varphi},
\]
where we have defined the matrix
\[  Y_R:=\mathbb E_k\,
    (I_{\mathsf{A}}\otimes \bra{0}_{\mathsf{W}} \otimes \bra{\psi_{R,k}}_{\mathsf{X}})\cdot \Pi_k \cdot
    (I_{\mathsf{A}}\otimes \ket{0}_{\mathsf{W}} \otimes \ket{\psi_{R,k}}_{\mathsf{X}}).
\]
By construction, $Y_R$ is positive semidefinite and acts on the $S$-qubit
advice register. Since the $\ket{0}_{\mathsf{W}}$ dilation term is irrelevant for our arguments, we will drop it for the rest of this overview.

With this setup, we can express the adversary's win probability in the aux-input random oracle model as the expected operator norm 
\begin{equation}
    \underset{R}{\mathbb E}\max_{\ket{\varphi}}
        \bra{\varphi}Y_R\ket{\varphi}
    =\underset{R}{\mathbb E}\|Y_R\|_{\mathrm{op}}.
    \label{eq:overview-advice-norm}
\end{equation}
Thus, we now ask: how do we prove tight bounds on this expected operator norm? \cite{DLM26} expressed $\|Y_R\|_{\mathrm{op}}$ as the squared operator norm of a rectangular square root of $Y_R$ and applied a concentration inequality for matrix Rademacher series, but this yields a suboptimal bound and does not generalize to $T>0$. 

\paragraph{The trace moment method.}
Rather than making black-box use of existing concentration inequalities, we apply one of the most common general-purpose methodologies for bounding the operator norm of a random matrix: the trace moment method. In particular, we consider the $2S$th trace moments and bound
\begin{equation}
    \mathbb E_R\|Y_R\|_{\mathrm{op}} \leq \Big(\mathbb E_R\|Y_R\|^{2S}_{\mathrm{op}}\Big)^{1/2S}
    \leq
    \left(\mathbb E_R\operatorname{Tr}(Y_R^{2S})\right)^{1/2S}.
    \label{eq:overview-trace-bound}
\end{equation}
We choose a moment of order $S$ because the trace sums over $2^S$
basis vectors, so the second inequality only incurs a constant factor loss. We thus focus on bounding
\[
    \mathbb E_R\operatorname{Tr}(Y_R^{2S})
    =\sum_{z\in\{0,1\}^S}\mathbb E_R\bra zY_R^{2S}\ket z,
\]
so every term now starts from an oracle-independent state $\ket z$.

\paragraph{Enter compressed oracles.}
The entries of the matrix $Y_R^{2S}$ are degree $4S$ polynomials in the coefficients of $R$; in such situations, in order to understand $\mathbb E_R Y_R^{2S}$ it is often fruitful to purify the random $R$ and study the resulting matrix using the compressed oracle framework \cite{C:Zhandry19}. 

In more detail, in the uncompressed
purification, the preparation of $\ket{\psi_{R,K}}$ is the isometry
\[
    \ket R
    \longmapsto
    \ket{\psi_{R,k}}_{\mathsf X}\ket R
    = \frac{1}{\sqrt N}\sum_{x\in[N]}
      \ket x_{\mathsf X}\otimes R(k,x)\ket R.
\]
In the compressed oracle view, multiplication by $R(k,x)$ is replaced
by the operator $T_{k,x}$ that toggles the corresponding database entry:
\[
    T_{k,x}\ket D
    := \ket{D\oplus\{(k,x)\}}.
\]
Thus, the purified challenge preparation map given intial database $\ket{D}_{\mathsf{db}}$ is given by
\[
    \widehat C_k:
    \ket D_{\mathrm{db}}
    \longmapsto
    \frac{1}{\sqrt N}\sum_{x\in[N]}
    \ket x_{\mathsf X}
    \otimes T_{k,x}\ket D_{\mathrm{db}}.
\]
The compressed version of $Y_R$ is therefore
\[
    \widehat Y
    := \mathbb E_k\,
       \widehat C_k^\dagger
       (\Pi_k\otimes I_{\mathrm{db}})
       \widehat C_k.
\]
This operator acts on the advice register $\mathsf A$ and the database register $\mathrm{db}$. As initially a uniformly random $R$ is sampled, the initial compressed oracle database is $\ket{\emptyset}$ and the expression we are trying to bound is equal to 

\begin{equation}
\label{eq:overview-compressed-moment}
    \mathbb E_R\operatorname{Tr}(Y_R^{2S})
    = \sum_{z\in\{0,1\}^S}
      \bra{z,\emptyset}\widehat Y^{2S}\ket{z,\emptyset}
    = \sum_{z\in\{0,1\}^S}
      \left\|\widehat Y^S\ket{z,\emptyset}\right\|^2.
\end{equation}

In a nutshell, there is hope for bounding this expression because (fixing a choice of $z$) $\ket{\phi_S} = \widehat Y^{S} \ket{z, \emptyset}$ is a state supported on compressed oracle databases of size at most $2S$. Indeed, by considering the sequence of states $\ket{\phi_0}, \ket{\phi_1}, \hdots, \ket{\phi_S}$ and applying a simple induction, our desired bound $\| \ket{\phi_S}\|^2 \leq \Big(O\big(\frac{\sqrt S}K\big)\Big)^{2S}$ follows from the following main technical claim:

\begin{claim}
    Let $\ket{\phi}$ on registers $\mathsf A, \mathrm{db}$ be supported on compressed oracle databases of size at most $\ell$. Then, 
    \[ \bra{\phi} \widehat Y \ket{\phi} \leq O\Big( \frac{1+\sqrt \ell}K\Big).
    \]
\end{claim}
Informally, we view this type of argument as an abstract generalization of the alternating measurement game in \cite{EC:Liu23} where a (potentially non-physical) process is repeated, maintaining a joint advice-database state supported on sparse databases. The expression $\bra{\phi} \widehat Y \ket{\phi}$ exactly corresponds to playing the OWS inversion game with a partially filled database register (entangled with the challenge state). 

\paragraph{A simple small database bound.}
To complete the proof, it remains to prove the claim, bounding $\bra{\phi}\widehat Y\ket{\phi}$ over databases of size at most $\ell$. We emphasize that in order to get the tight final result, we must currently prove a bound of $\sqrt{\ell}/K$ (rather than a simpler $\ell/K$).

Take any normalized state
\[
    \ket{\phi}
    = \sum_D \ket{\phi_D}_{\mathsf A}\ket D_{\mathrm{db}}
\]
supported on databases of size at most $\ell$.
For a fixed key $k$, the state after challenge preparation is
$\widehat C_k\ket{\phi}$, so the quantity we want to bound is
\[
    \bra{\phi}\widehat Y\ket{\phi}
    = \frac1K\sum_k
      \left\|
        (\Pi_k\otimes I_{\mathrm{db}})
        \widehat C_k\ket{\phi}
      \right\|^2.
\]
To make the expression more convenient, we now also purify the random key $k$. The joint state after coherent challenge preparation is
\[
    \ket{\Psi}
    := \frac1{\sqrt K}\sum_k
       \bigl(\widehat C_k\ket{\phi}\bigr)\ket k_{\mathsf K}
    = \frac1{\sqrt{KN}}\sum_{k,D,x}
      \ket{\phi_D}_{\mathsf A}\ket x_{\mathsf X}
      \ket{D\oplus\{(k,x)\}}_{\mathrm{db}}\ket k_{\mathsf K}.
\]
The event that the adversary's output equals the purified
hidden key is represented by
\[
    \Pi_{\mathrm{win}}
    := \sum_k
       \Pi_k\otimes I_{\mathrm{db}}
       \otimes\ketbra k_{\mathsf K},
\]
so that
\[
    \bra{\phi}\widehat Y\ket{\phi}
    = \left\|\Pi_{\mathrm{win}}\ket{\Psi}\right\|^2.
\]
We bound this squared norm by splitting $\Pi_{\mathrm{win}}\ket{\Psi}$ into three parts. Write $D_k=\{x:(k,x)\in D\}$ and choose a threshold $h\geq 1$. Our three parts are then given by:

\begin{itemize}
    \item The triples $(k, D, x)$ with $|D_k|>h$ (``heavy rows''). This component $\ket{\Psi_1}$ of $\ket{\Psi}$ has squared norm at most $\ell/(Kh)$, because a database with $\ell$ entries has at most $\ell/h$ such rows, and $k$ is uniform.
    \item Among the remaining $k$, consider the part satisfying $(k,x)\in D$ (where the challenge preparation actually \emph{removes} an entry from $D$). This component $\ket{\Psi_2}$ of $\ket{\Psi}$ has squared norm at most $\ell/(KN)$, since a uniform pair $(k,x)$ hits the database with at most this probability.
\end{itemize}

In the remaining \emph{new entry} (and main) part $\ket{\Psi_3}$ of $\ket{\Psi}$, challenge preparation adds $(k,x)$ to a row of size at most $h$. We group its amplitudes by the final
database $F=D\cup\{(k,x)\}$. For a fixed final database \(F\) and key \(k\), the only possible histories are obtained by declaring one of the entries \((k,x)\in F\) to be the newly added challenge entry; hence there are at most \(h+1\) coherently interfering histories. Cauchy--Schwarz therefore bounds its
success contribution by
\[
    \|\Pi_{\mathrm{Win}} \ket{\Psi_3}\|^2 \leq \frac{h+1}{KN}\sum_{D,k,x}
       \bra{\phi_D}_{\mathsf{A}}\bra{x}_{\mathsf{X}}\Pi_k
       \ket{\phi_D}_{\mathsf{A}}\ket{x}_{\mathsf{X}}
    =\frac{h+1}{K}.
\]
Combining the three parts using the triangle inequality gives
\[
    \bra\phi\widehat Y\ket\phi
    \leq 3\left(\frac{\ell}{Kh}
                    +\frac{\ell}{KN}
                    +\frac{h+1}{K}\right).
\]
Setting $h = \sqrt{\ell}$ yields the desired bound. 

\subsection{Adding online queries}
\label{subsec:overview-online-queries}
Fortunately, the $T=0$ case captures most of our full proof of \cref{thm:main-ows-intro}. To handle the case of $T>0$, we let $A_R$ denote the adversary's $T$-query algorithm. The matrix of interest then becomes
\[
    Y_R=\mathbb E_k\,
       \bra{\psi_{R,k}}A_R^\dagger\Pi_kA_R\ket{\psi_{R,k}},
\]
with fixed ancillas suppressed. Indeed, the matrix $Y_R$ still has the property that $\bra{\varphi} Y_R \ket{\varphi}$ is equal to the probability that the adversary successfully inverts on a random challenge with advice state $\ket{\varphi}$. Crucially, $Y_R$ still has the property that its entries have low degree in $R$, namely degree $2(T+1)$. As a result, when we purify $Y_R$ to $\widehat Y$, the resulting matrix can increase a compressed oracle database size by at most $2(T+1)$. 

With this setup, our proof from earlier goes through unchanged up until the key claim, which now becomes:

\begin{claim}
    Let $\ket{\phi}$ on registers $\mathsf A, \mathrm{db}$ be supported on compressed oracle databases of size at most $\ell$. Then, 
    \[ \bra{\phi} \widehat Y \ket{\phi} \leq O\Big( \frac{1+T^2+\sqrt \ell}K\Big).
    \]
\end{claim}
In other words, the $T$ queries to (the purified) $R$ help with inversion only so much as Grover's algorithm moves a state towards the desired solution. 

We prove this claim (\cref{lem:1ows-l-database-t-query}), by generalizing our $T=0$ analysis. The main new difficulty concerns contributions for which the entry $(k,x)$ associated with the original $\ket{x}$ summand of the challenge is
absent from the final database. In the $T=0$ case, this can happen only if $(k,x)$ was already present in the initial database, so challenge preparation deleted it. With online queries, however, an entry added by challenge preparation can also be removed by a subsequent online query to the same position. We bound these contributions inductively: their norm starts at most $\sqrt{\ell/(KN)}$ and increases by at most $O(1/\sqrt K)$ per query, giving a squared norm bound of $O\big(\frac \ell {KN}+\frac{T^2}K\big)$.

Substituting $\ell=O(S(T+1))$ into the claim yields
\[
    \varepsilon_{\mathrm{OWS}}
    \leq O\!\left(\frac{T^2+\sqrt{S(T+1)}}{K}\right),
\]
completing our proof of \cref{thm:main-ows-intro}. 

\subsection{Sketches of our additional results}
We defer detailed discussions of our additional results to their respective technical sections, but we briefly summarize the additional ideas required to obtain them. We emphasize that for each result below, we apply the high-level strategy of bounding the trace moments of an appropriate matrix $Y_R$.

\paragraph{Post-quantum PRGs.}
For a random oracle $R:[N]\to[M]$, let $A_R=(O_RU)^T$ denote the
adversary's online computation and let $0\preceq\Pi\preceq I$ be its
final acceptance operator. The matrix of interest is now
\[
    Y_R :=
    \mathbb E_{x\in[N]}
    \bra{R(x)}A_R^\dagger\Pi A_R\ket{R(x)}
    - \mathbb E_{r\in[M]}
    \bra r A_R^\dagger\Pi A_R\ket r.
\]
In particular, 
$|\bra\varphi Y_R\ket\varphi|$ is equal to the adversary's distinguishing advantage using advice
$\ket\varphi$, so optimizing over advice again gives
$\mathbb E_R\|Y_R\|_{\mathrm{op}}$. Compared to the OWS (or OWF) setting, there is not an obvious operational interpretation of ``repeating $Y_R$'', but we can nevertheless apply the trace moment method (and inductive analysis from earlier) to bound its operator norm! 

As before, each application of the compressed operator $\widehat Y$
increases the database size by at most $2T+2$. The analogue of our key
claim is that, for any normalized state $\ket\phi$ on the advice and
database registers supported on databases of size at most $\ell$,
\[
    \left|\bra\phi\widehat Y\ket\phi\right|
    \leq O\!\left(\frac{T^2}{N}+\sqrt{\frac{\ell}{N}}\right).
\]
To prove this bound, we introduce a resampling hybrid: instead of
providing $R(x)$ as the challenge, we sample an independent
$r\in[M]$ and replace the oracle value at the hidden seed $x$ by $r$.
The contribution from the initial database is
$O(\sqrt{\ell/N})$, while the change in the average acceptance probability
caused by reprogramming a uniform position is $O(T^2/N)$. Substituting $\ell=O(S(T+1))$ into the small database bound
gives
\[
    \epsilon_{\mathrm{PRG}}
    \leq O\!\left(\frac{T^2}{N}
          +\sqrt{\frac{S(T+1)}{N}}\right),
\]
proving Theorem~1.2.

\paragraph{1PRS with $T=0$.}
This setting follows the setup from the 1OWS case but works with a single test $\Pi$ (rather than a measurement $\{\Pi_k\})$ and considers the centered matrix 
\[
    Y_R := \mathbb E_k
    \bra{\psi_{R,k}}\Pi\ket{\psi_{R,k}}
    - \frac1N\sum_{x\in[N]}\bra x\Pi\ket x.
\]
This is the difference between the real and ideal acceptance operators
on the $S$-qubit advice register, so maximizing
$|\bra\phi Y_R\ket\phi|$ over advice gives
$\|Y_R\|_{\mathrm{op}}$. 

The corresponding ``small database'' bound, presented in \cref{sec:prs-small-database}, is a significant departure from the other proof techniques in this paper. Specifically, we introduce an \emph{alternative} purification of compressed oracle databases as uniform superpositions of \emph{ordered} databases, and describe the effect of challenge generation in terms of \emph{swaps} between the existing database positions and a new one, increasing the database size by one. We refer the reader to \cref{sec:quantum-program-1prs} for more details. 

\paragraph{One-query unitary synthesis.}
To recover the random binary phase state result from \cite{DLM26}, we apply their spectral relaxation but give a new compressed oracle analysis of the resulting expected operator norm. Writing the adversary as a fixed isometry $V:\mathbb C^N\to\mathbb C^M$,
a phase query $O_{f_R}$, and a final measurement, the weight vector
decomposition of~\cite{STOC:LomMaWri24,DLM26} gives
$V\ket{\psi_{R,k}}=D_{R,k}\ket{\mathrm{wt}_V}$, where
$\ket{\mathrm{wt}_V}$ is fixed and $D_{R,k}$ is diagonal. Since
$D_{R,k}$ commutes with $O_{f_R}$, search success is a quadratic form
in $O_{f_R}\ket{\mathrm{wt}_V}$ with matrix
\[
    Y_R := \mathbb E_k D_{R,k}^\dagger\Pi_k D_{R,k}.
\]
Thus, it suffices to bound 
$\mathbb E_R \|Y_R\|_{\mathrm{op}}$. Crucially,
$D_{R,k}=\sum_x R(k,x)D_x$ for fixed diagonal matrices $D_x$, so the matrix $Y_R$ is Hermitian PSD and quadratic in $R$! Thus, we can apply the same methodology and attempt to purify $R$ and bound the resulting trace moments in the compressed oracle basis. For the small database
bound, we (similarly to the 1OWS case with $T=0$) split the rescaling into addition and deletion terms and use
$\sum_xD_x^\dagger D_x=\sum_xD_xD_x^\dagger\preceq I$ in place of
normalization of challenge preparation. This gives a small database search
bound of $O((\ell+1)/K)$. We refer the reader to \cref{sec:unitary-synthesis} for more details. 

\paragraph{The hybrid unitary synthesis model.}
In the hybrid unitary synthesis model, there are two types of queries to the random oracle $R$ that must be treated differently. For queries to $R$ that occur \emph{before} the function query $f_R$, we generically remove them with a sub-optimal but polynomial security loss, and leave a tight analysis of these queries for future work. 

The main part of the analysis is handling the combination of one function query $f_R$ with $T$ queries to $R$ afterwards. Letting $A_R$ be the adversary's $T$-query algorithm, we arrive at the matrix 
\[
    Y_R := \mathbb E_k
    D_{R,k}^\dagger A_R^\dagger\Pi_k A_R D_{R,k}.
\]
Crucially, $Y_R$ still has degree at most $2T+2$ in $R$, so the methodology (conceivably) applies.
The key idea in the small database bound is to swap the challenge row with an
unused independent row. The initial swap contributes
$O(\sqrt{\ell/K})$ to the success amplitude, while the subsequent
query hybrid contributes $O(T/\sqrt K)$. This gives the desired bound
$O((1+T^2+\ell)/K)$.

\section{Preliminaries}
\label{sec:preliminaries}

We use $N=2^n$ for the dimension of the main $n$-qubit register and $M=2^m$ for the dimension of the adversary's post-isometry workspace. For a positive integer $K$, we write $[K]:=\{1,\dots,K\}$. When we write $x+y$ or $x\cdot y$ for $x,y\in\{0,1\}^n$, the operations are taken over $\mathbb{F}_2^n$.

\subsection{Quantum random oracle model and compressed oracles}

In the quantum random oracle model, a hash function is modeled as a random classical function $R:\{0,1\}^{n_1}\to \{0,1\}^{n_2}$ (with the below sections, we might consider different input and output lengths). The function $R$ is sampled at the beginning of any considered game and then fixed. A quantum oracle algorithm with oracle access to $R$ can make queries to $R$, with oracle access defined by a unitary $O_R:\ket{x,y}\mapsto\ket{x,y\oplus R(x)}$, together with intermediate unitaries between queries as $U_1,...,U_{t+1}$. That is, a quantum oracle algorithm with access to $R$ can be defined as a sequence of unitaries: $U_1, O_R,U_2,O_R,...,U_t,O_R,U_{t+1}$. The number of queries in this case is $t$. Without loss of generality (by adding a register as a counter), we can assume that the intermediate unitaries are the same $U$.

\paragraph{Compressed oracles \cite{C:Zhandry19}.} We use the Fourier representation of the compressed oracle. That is, we consider both the uncompressed phase oracle
\[ \ket{x} \otimes \ket{R} \mapsto R(x) \ket{x}\otimes \ket{R}
\]
for $R: [N]\rightarrow \{\pm 1\}$, and the uncompressed standard oracle
\[ \ket{x,y}\otimes \ket{R} \mapsto \ket{x, y\oplus R(x)} \otimes \ket{R},
\]
and convert the purification register of each of these into the Fourier basis. In the Fourier basis, we describe the purification register as containing a \emph{database}, and we have that a compressed oracle query changes at most one entry in a database. We refer the reader to the original paper \cite{C:Zhandry19} for more details.

\subsection{Canonical quantum bit commitments}
\label{sec:canonical-commitments}

We follow the canonical quantum bit commitments as defined from \cite{AC:Yan22,EC:HhaMorYam23}.
We use the formulation with quantum auxiliary inputs, and later consider its quantitative security in the quantum random oracle model.
This section is mostly taken verbatim from Section 3.2 of \cite{EC:HhaMorYam23}.

\begin{definition}[Canonical quantum bit commitments]
\label{def:canonical-commitment}
A \emph{canonical quantum bit commitment scheme} is specified by
polynomial-time computable unitaries
$\{Q_0(\lambda),Q_1(\lambda)\}_{\lambda\in\mathbb N}$
on polynomial-size registers $\mathsf{C}$ (commit) and $\mathsf{W}$ (decommit, as witness).
For $b\in\{0,1\}$, define
\[
    \ket{\psi_b}_{\mathsf{C}\mathsf{W}}:=Q_b\ket{0}_{\mathsf{C}\mathsf{W}},
    \qquad
    \Pi_b:=\ket{\psi_b}\!\bra{\psi_b}_{\mathsf{C}\mathsf{W}},
    \qquad
    \rho_b:=\Tr_\mathsf{W}(\Pi_b).
\]
(For simplification we omit $\lambda$ and simply write $Q_0,Q_1$ to mean $Q_0(\lambda)$ and $Q_1(\lambda)$.)
\end{definition}

\begin{remark}[Protocol]
To commit to $b$, the sender prepares $\ket{\psi_b}_{\mathsf{C}\mathsf{W}}$,
transmits $\mathsf{C}$, and retains $\mathsf{W}$.
To open, the sender transmits $(b,\mathsf{W})$; the receiver accepts according
to the projector $\Pi_b$.
\end{remark}

\begin{definition}[Hiding]
\label{def:commitment-hiding}
For \emph{computational hiding}, every non-uniform quantum polynomial-time
distinguisher $\mathcal A$ with polynomial-size quantum advice must satisfy
\[
    \left|
      \Pr\bigl[\mathcal A(1^\lambda,\rho_0)=1\bigr]
      -
      \Pr\bigl[\mathcal A(1^\lambda,\rho_1)=1\bigr]
    \right|
    =\negl(\lambda).
\]
\emph{Statistical hiding} and \emph{perfect hiding} require, respectively,
\[
    \frac12\lVert\rho_0-\rho_1\rVert_1=\negl(\lambda)
    \qquad\text{and}\qquad
    \rho_0=\rho_1.
\]
\end{definition}

\begin{definition}[Binding]
\label{def:commitment-binding}
\emph{Computational binding} requires every polynomial-time computable
unitary $U_{\mathsf{W}\mathsf{Z}}$ and every normalized auxiliary state $\ket{\tau}_\mathsf{Z}$
on polynomially many qubits to satisfy
\[
    \left\|
      (\Pi_1\otimes I_\mathsf{Z})
      (I_\mathsf{C}\otimes U_{\mathsf{W}\mathsf{Z}})
      \bigl(\ket{\psi_0}_{\mathsf{C}\mathsf{W}}\otimes\ket{\tau}_\mathsf{Z}\bigr)
    \right\|_2
    =\negl(\lambda).
\]
\emph{Statistical binding} removes the efficiency restriction on $U$.
\emph{Perfect binding} requires the norm to equal zero for every such
$U$ and $\ket{\tau}$.
\end{definition}

\subsection{1OWS and 1PRS}
\label{sec:OWSG}
Let $\mathsf{StateGen}$ be a quantum polynomial-time algorithm that,
on input $(1^\lambda,k)$ for $k\in[K]$, outputs an $n(\lambda)$-qubit
pure state $|\psi_k\rangle$, where $K=2^{\kappa(\lambda)}$ and
$N=2^{n(\lambda)}$. Here $\kappa$ and $n$ are polynomially bounded.
We consider non-uniform quantum polynomial-time adversaries with
polynomial-size quantum advice independent of the challenge. For
simplicity, we omit $\lambda$ from the notation when it is clear.

One-way state generators were introduced by Morimae and
Yamakawa~\cite{C:MorYam22}. Their standard definition uses verification of the adversary's output key, while~\cite[Remark~4.1]{C:MorYam22}
also considers the alternative requirement of recovering the
sampled key itself. We use the single-copy version of this
key-recovery formulation. Its success condition also matches
that of the oracle state search game in~\cite[Definition~4.1]{DLM26}.

\begin{definition}[Single-copy one-way states, 1OWS]
We say that $\mathsf{StateGen}$ generates single-copy one-way states
($\mathsf{1OWS}$), in the key-recovery sense, if every such adversary
$\mathcal A$ satisfies
\[
  \Pr_{k\leftarrow[K]}
  [\mathcal A(1^\lambda,|\psi_k\rangle)=k]
  =\operatorname{negl}(\lambda).
\]
The adversary receives one copy of $|\psi_k\rangle$ and succeeds if
it outputs the sampled key $k$. 
\end{definition}
The above definition allows 1OWS to be information-theoretically
secure. 
For the random binary phase states considered in this paper, each pair of distinct states has expected squared overlap \(1/N\). Summing over the \(K-1\) wrong keys gives a total expected squared overlap of \((K-1)/N\le 1\). The overlap bound for the pretty-good measurement (or quantum union bound on testing a constant fraction of all possible keys)  therefore implies that an adversary knowing the full oracle table can recover the key with constant probability, averaged over \(R\). Thus, this family is information-theoretically insecure, and we study the resources (queries, advice) required to recover the key.\\

A distinguishing version gives pseudorandom states. Pseudorandom quantum states were introduced by Ji, Liu, and Song~\cite{C:JiLiuSon18}, whose definition requires computational indistinguishability from Haar-random states even given polynomially many copies. We consider the single-copy variant studied by Morimae and Yamakawa~\cite[Definition~2.2]{C:MorYam22} and Lombardi, Ma, and
Wright~\cite[Definition~5.1]{STOC:LomMaWri24}, where the adversary receives only one copy of the challenge state.

\begin{definition}[Single-copy pseudorandom states, 1PRS]
We say that $\mathsf{StateGen}$ generates single-copy pseudorandom
states ($\mathsf{1PRS}$) if $n(\lambda)\ge\kappa(\lambda)+1$ and
every such adversary $\mathcal A$ satisfies
\[
  \left|
  \Pr_{k\leftarrow[K]}
    [\mathcal A(1^\lambda,|\psi_k\rangle)=1]
  -\Pr_{|\psi\rangle\leftarrow\mu_N}
    [\mathcal A(1^\lambda,|\psi\rangle)=1]
  \right|
  =\operatorname{negl}(\lambda),
\]
where $\mu_N$ is the Haar distribution on pure states in
$\mathbb C^N$. In both experiments the adversary receives one copy
of the challenge state.
\end{definition}

In this paper, we consider the binary phase state family \cite{STOC:LomMaWri24} generated from a random oracle $R:[K]\times[N]\to\{\pm 1\}$, $|\psi_{R,k}\rangle
  :=\frac{1}{\sqrt N}\sum_{x\in[N]}{R(k,x)}|x\rangle$.
For a known key $k$, this state can be generated using one query to $R$. We study the key-recovery probability and distinguishing advantage for this family under non-uniform (with advice) oracle-access bounds specified in each game.

\subsection{The unitary synthesis problem}

The unitary synthesis problem asks whether an arbitrary quantum unitary can be (approximately) implemented efficiently given oracle access to a suitable Boolean function~\cite{STOC:LomMaWri24}. We recall the oracle-circuit formulation from \cite{STOC:LomMaWri24}.

\begin{definition}[Approximating a unitary, \cite{STOC:LomMaWri24}]\label[definition]{def:approx_unitary}
Let $U$ be an $n$-qubit unitary, and let $\Phi_U$ be the associated quantum channel. Let $\Phi_{\mathrm{approx}}$ be a quantum channel with $n$-qubit input and output registers. We say that $\Phi_{\mathrm{approx}}$ $\varepsilon$-approximates $U$ if
\[
D_\diamond(\Phi_{\mathrm{approx}},\Phi_U)\le \varepsilon.
\]
\end{definition}

\begin{definition}[Channel implemented by an oracle circuit, \cite{STOC:LomMaWri24}]
\label{def:channel_implemented_by_oracle_circuit}
Given a $t$-query oracle circuit $\cA^{(\cdot)}$ with an $n$-qubit input register, an $(m-n)$-qubit ancilla register, intermediate unitaries $U_1,\dots,U_{t+1}$ on $m$ qubits, and a Boolean function $f:\{0,1\}^{m}\to\{\pm 1\}$, the induced $n$-qubit channel $\Phi_{\cA^f}$ acts as follows.
\begin{enumerate}[leftmargin=1.5em]
    \item On input $\ket{\psi}$, prepare
    \[
    U_{t+1} \cdot O_f \cdot U_t \cdots O_f \cdot U_2 \cdot O_f \cdot U_1(\ket{\psi}\ket{0^{m-n}}).
    \]
    \item Output the first $n$ qubits and discard the remaining $(m-n)$ qubits.
\end{enumerate}
\end{definition}

\begin{definition}[The unitary synthesis problem, \cite{STOC:LomMaWri24}]
Fix an error parameter $\varepsilon(n)=2^{-\Omega(n)}$. The unitary synthesis problem asks whether there exists a universal oracle circuit $\cA^{(\cdot)}$, computable by a $\poly(n)$-size quantum circuit and making $\poly(n)$ oracle queries, such that for every $n$-qubit unitary $U$ there exists a Boolean function $f:\{0,1\}^*\to\{\pm 1\}$ with
\[
D_\diamond(\Phi_{\cA^f},\Phi_U)\le \varepsilon(n).
\]
\end{definition}

In terms of proving unitary synthesis hardness, instead of dealing with diamond norm hardness, we study through two related oracle state games: the oracle state distinguishing game from \cite{STOC:LomMaWri24} and the oracle state search game from \cite{DLM26}. The reductions in \cite[Section~3.3]{STOC:LomMaWri24} and \cite[Section~4.1]{DLM26} connect hardness of these games, for suitable state families, to the hardness of unitary synthesis. We formally define the search and distinguishing games in \Cref{sec:unitary-synthesis-search,sec:unitary-synthesis-distinguish}, respectively, and analyze them for random binary phase states.

\subsection{One-query oracle algorithms}

The lower bounds in this paper all concern one-query algorithms in terms of the arbitrary phase oracle, so we record a normal form. This also helps capture  one-query algorithms when additional access to the underlying random oracle is allowed.

\begin{definition}[One-query unitary synthesis algorithm]\label{def:one-query-algorithm}
A one-query unitary synthesis algorithm on $n$-qubit inputs is specified by
\begin{itemize}[leftmargin=1.5em]
    \item an oracle $O_f$,
    \item an isometry $V:\mathbb{C}^{N}\to\mathbb{C}^{M}$, representing the computation before the oracle query, and
    \item a unitary $U$ on $\mathbb{C}^{M}$, representing the computation after the oracle query.
\end{itemize}
On input $\ket{\psi}\in\mathbb{C}^N$, the algorithm prepares $U\,O_f\,V\ket{\psi}$ and then outputs the designated $n$-qubit subsystem.
\end{definition}

Fixing the computational basis on the $M$-dimensional workspace, every isometry $V:\mathbb{C}^{N}\to\mathbb{C}^{M}$ can be written as
\[
V=\sum_{i\in[M]} \ket{i}\bra{v_i},
\]
where the vectors $\ket{v_i}\in\mathbb{C}^N$ satisfy
\[
\sum_{i\in[M]} \ketbra{v_i}=\Id_N.
\]
We will repeatedly use this row-vector description of $V$ below.

\subsection{The weight-vector decomposition}\label{sec:weight-vector-decomposition}
We now record the diagonal decomposition of the pre-query isometry. This is mostly taken verbatim from Section 3.3 of \cite{DLM26}. While \cite[Lemma~3.10]{DLM26} uses a weight vector attached to both $V$ and a distribution on the input states, here we consider the distribution as Haar random, since the input state from this paper is mostly a random binary phase state.

\begin{lemma}[Weight-vector decomposition, modified from {\cite[Lemma~3.10]{DLM26}}]
\label{lem:weight-vector-decomposition}
Let $V:\mathbb{C}^N\to\mathbb{C}^M$ be an isometry, and write
\[
V=\sum_{i\in[M]} \ket{i}\bra{v_i}.
\]
Let $\mu$ be the Haar random distribution on pure states in $\mathbb{C}^N$, and define
\[
p_i:=\mathbb{E}_{\psi\sim\mu}\bigl[\abs{\bra{v_i}\psi\rangle}^2\bigr]=\frac{\left<v_i|v_i\right>}{N},
\qquad
\ket{\wt_{V}}:=\sum_{i\in[M]} \sqrt{p_i}\,\ket{i}.
\]
For each pure state $\ket{\psi}$, define the diagonal matrix
\[
D_{V,\psi}
:=
\sum_{i\in[M]:\,p_i>0}
\frac{\bra{v_i}\psi\rangle}{\sqrt{p_i}}\ketbra{i},
\]
with diagonal entry $0$ when $p_i=0$. Then $\ket{\wt_{V}}$ is a unit vector, and for every $\ket{\psi}$ we have
\[
V\ket{\psi}=D_{V,\psi}\ket{\wt_{V}}.
\]
\end{lemma}

\begin{proof}
Since $V$ is an isometry,
\[
\sum_{i\in[M]} \abs{\bra{v_i}\psi\rangle}^2
=
\bra{\psi}\left(\sum_{i\in[M]}\ketbra{v_i}\right)\ket{\psi}
=
\bra{\psi}\psi\rangle
=1
\]
for every unit vector $\ket{\psi}$. From the definition of $p_i$, $\sum_i p_i=1$, so $\ket{\wt_{V}}$ has unit norm.

For the decomposition itself, if $p_i=0$ then the nonnegative random variable $\abs{\bra{v_i}\psi\rangle}^2$ has expectation $0$, and hence vanishes for every $\ket{\psi}$. Therefore, for every $\ket{\psi}$,
\[
D_{V,\psi}\ket{\wt_{V}}
=
\sum_{i\in[M]} \frac{\bra{v_i}\psi\rangle}{\sqrt{p_i}}\ketbra{i}
\left(\sum_{j\in[M]}\sqrt{p_j}\ket{j}\right)
=
\sum_{i\in[M]} \bra{v_i}\psi\rangle\ket{i}
=
V\ket{\psi}.
\]
\end{proof}

\section{Our use of the trace moment method}
 
All of our proofs follow the same high-level strategy: bound the winning probability (or distinguishing advantage) by the expected operator norm of a random Hermitian matrix.
\begin{itemize}
    \item For quantum advice, this matrix arises by optimizing over the advice state, and view the value as the operator norm.
    \item For one-query unitary synthesis lower bound, it arises from the spectral relaxation from the weight vector decomposition trick in \Cref{sec:weight-vector-decomposition}. This expresses the maximum over all $f_R$ as the maximum quadratic form over a diagonally updated weight vector, which is upper bounded by an operator norm.  
\end{itemize}  
Once this reduction is made, the same trace-moment argument applies: Let $Y_R$ be a random Hermitian matrix acting on a $d$-dimensional space. For every positive integer $t$,
\[
\mathbb E_R\|Y_R\|
\leq
\left(\mathbb E_R\operatorname{Tr}(Y_R^{2t})\right)^{1/2t}.
\]
We express the trace moment using the compressed oracle representation of $Y_R$, with the database initially empty. 
\begin{itemize}
    \item The trace can be viewed as an algorithm that starts with a uniformly random classical input, and each application of $Y$ changes the database size by a bounded amount.
    \item It therefore suffices to control the corresponding operator on states supported on small databases, as this will help us bound the growth of the norm. 
    \item Taking $t$ on the order of the input qubits will make the factor contributed by the trace dimension constant (or say, cancel out the union bound or guessing of the starting classical input).
\end{itemize}

\section{Non-uniform security of 1OWS in QROM}
\label{sec:time-advice-tradeoff-search}

In this section, we prove a non-uniform security of single-copy one-way states (specifically, the binary phase states) in the quantum random oracle model, via the trace moment method. Specifically, we show that any $T$-query quantum algorithm with $S$-qubit advice can invert a single copy of a random binary phase state with winning probability at most $O\left(\frac{T^2+1+\sqrt{S(T+1)}}{K}\right)$.

\begin{definition}[1OWS security game]\label{def:1ows-security-game}
In the non-uniform security game of 1OWS:
\begin{itemize}
    \item In the offline stage, a hash function $R:[K]\times [N]\to\{\pm 1\}$ is sampled and fixed (here $K\le N$), and an algorithm is allowed to have an $S$-qubit advice that depends arbitrarily on $R$ before receiving the challenge.
    \item Then, in the online stage, a challenger will sample a random key $k\gets [K]$, and generate a  single-copy of state $\ket{\psi_{R,k}}:=\frac{1}{\sqrt{N}}\sum_x R(k,x)\ket{x}$, and send it to the algorithm.
    \item The algorithm, after receiving the challenge $\ket{\psi_{R,k}}$, is allowed to make $T$ queries to $R$ and is asked to output a guess $k'$. It succeeds if $k'=k$. 
\end{itemize}
We consider security game in QROM, for $R$ as a random oracle, and the algorithm can have oracle access to $O_R$. The non-uniform security (in terms of $S$-qubit, $T$-query) of 1OWS in QROM is defined as the maximum winning probability, parametrized by $(S,T,K,N)$.
\end{definition}

\begin{theorem}[Non-uniform security of 1OWS in QROM]\label{thm:1ows-ts-tradeoff}
For integer $S,T\ge 0, K\le N$, the non-uniform security (with $S$-qubit advice and $T$-query) of 1OWS in QROM is $O\left(\frac{T^2+1+\sqrt{S(T+1)}}{K}\right)$; that is, the maximum winning probability is bounded by $O\left(\frac{T^2+1+\sqrt{S(T+1)}}{K}\right)$.
\end{theorem}

\paragraph{Proof overview.}
Via the trace-moment method, we first express the search probability as the expected norm of an operator on the advice register, bound it by a trace moment, and also view them in the compressed oracle database mode (\Cref{sec:1OWS-trace-moment-setup}). Then, the main step is to bound the corresponding search probability when the algorithm starts with an arbitrary state
entangled with a compressed database of size $\le \ell$ (\Cref{sec:1OWS-trace-moment-on-small-database}).

For the search case with initial database of size $\le\ell$, a weaker bound of $O((T^2+\ell)/K)$ would be easy to obtain, by introducing a  hybrid that resamples the challenge key row (resampling the $R(k,\cdot)$ part). Compared to the real case, the difference (in terms of amplitude) can be bounded by $O(\sqrt{\ell/K})$ for the initial database where the database on challenge key row $k$ is non-empty. Compared to an ideal case where the challenge is fully random, the difference can be bounded by $O(T/\sqrt{K})$ through a modified BBBV-style hybrid argument. 

To prove this stronger bound $O((T^2+\sqrt{\ell})/K)$ (\Cref{lem:1ows-l-database-t-query-reform}, proof in \Cref{sec:1OWS-bound-with-initial-database}), we use a finer analysis that checks contributions from the following components:
\begin{itemize}
    \item Is the final database on the challenge key $k$ large (or not)?
    \begin{itemize}
        \item If the final database on the challenge key $k$ is large, then the initial database on $k$ should also be large. Then this can be bounded by the size of the initial database. Specifically, for each initial database bounded by size $\ell$, it cannot have $h$-size databases on $\ge \ell/h$ keys. We set $h=\sqrt{\ell}$ here.
    \end{itemize}
    \item Does the ``challenge entry'' stay in the final database? Or say, for the $x$ part in the final state, where it is from $\ket{x}$ in the superposition of the initial challenge state, is $(k,x)$ in the final database?
    \begin{itemize}
        \item If $(k,x)$ is in the final database, note that the final database on challenge key $k$ is small now (if large, it lands in the case already discussed above). Therefore, in the corresponding coherent sum over $x$ inside the final database, paying only that small number can make it into a sum of individual squared norms, which would be easy to upper bound by completeness of the output measurement and unitarity.
        \item If $(k,x)$ is not in the final database, this means $(k,x)$ is in the initial database (which would be expected to bound by $O(\ell/KN)$), or is queried within the $T$ queries (which would be expected to bound by $O(T^2/K)$). This is because $(k,x)$ is added (by $\oplus$) to the database from the challenge state before the first query. 
    \end{itemize}
\end{itemize}

\subsection{1OWS trace-moment setup}
\label{sec:1OWS-trace-moment-setup}
For any quantum oracle algorithm in the 1OWS security (search) game, denote its intermediate unitary between queries as $U$, final projective measurement as $\{\Pi_k\}_{k\in [K]}$, and suppose of hash function $R$ its advice is $\ket{\phi_R}$. We can write its maximum winning probability as the following:
\[
\mathbb{E}_R\max_{\phi_R}\mathbb{E}_k\bra{\phi_R}
\bra{\psi_{R,k}}\bra{0^w}_{\mathsf W}(U^\dagger O_R^\dagger)^T\Pi_k(O_RU)^T\ket{\psi_{R,k}}
\ket{\phi_R}\ket{0^w}_{\mathsf W}.
\]
Here $\ket{0^w}_{\mathsf W}$ denotes the algorithm's workspace of $w$ qubits, which starts with an all $0$ state.
From the definition of the operator norm, this is exactly equal to 
\begin{align}\label{eq:1ows-goal-probability}
    \mathbb{E}_R\|Y_R\|:=
    \mathbb{E}_R\left\|
        \mathbb{E}_k\bra{\psi_{R,k}}\bra{0}_{\mathsf W}(U^\dagger O_R^\dagger)^T\Pi_k(O_RU)^T\ket{\psi_{R,k}}\ket{0}_{\mathsf W}
    \right\|.
\end{align}
Here we define $Y_R$ as the matrix inside the norm.
Now our goal is to bound $\mathbb{E}_R\|Y_R\|$. 

Via the trace moment method,  by Jensen's inequality, we have
\begin{align*}
    \mathbb{E}_{R}\|Y_{R}\|
    \le \left(\mathbb{E}_{R}\|Y_{R}\|^{2t}\right)^{1/2t}\le \left(\mathbb{E}_{R}\Tr(Y_{R}^{2t})\right)^{1/2t}.
\end{align*}
Therefore, now our goal is to upper bound
\begin{align}\label{eq:1ows-goal-trace-moment}
    \mathbb{E}_{R} \Tr(Y_{R}^{2t})=\sum_{z\in [2^S]} \mathbb{E}_R\bra{z}Y_R^{2t}\ket{z}.
\end{align}

\paragraph{Rewrite in compressed oracle framework.} For random function $R$, we can write all things above in the compressed oracle framework. 

Note that in the un-compressed world, the isometry that writes $\ket{\psi_{R,k}}$ into a fresh challenge register can be viewed as   
\[
\ket{R}\mapsto 
\ket{\psi_{R,k}}\otimes\ket{R}=
\frac{1}{\sqrt{N}}\sum_{x} \ket{x}\otimes {R(k,x)}\ket{R}.
\]
Define $T_{k,x}$ on the compressed oracle register that flips the presence of the database entry $(k,x)$; that is, $T_{k,x}\ket{D}=\ket{D\oplus \{(k,x)\}}$.
Now the isometry that generates the fresh challenge using
one compressed query can be viewed as
\[
\ket{D}\mapsto  \left(\frac{1}{\sqrt{N}}\sum_x\ket{x}\otimes T_{k,x}\right) \ket{D}.
\]
Therefore, we can define $\widehat Y$ in the compressed oracle database mode as
\begin{align*}
    \widehat Y
    =\frac1N \cdot\mathbb{E}_k \sum_{x_1,x_2}\bra{x_1}\bra{0^w} (U^\dagger \cO^\dagger)^T \Pi_k (\cO \; U)^T \ket{x_2}\ket{0^w}\otimes T_{k,x_1}T_{k, x_2}
\end{align*}
with compressed phase oracle $\cO$. This will help us rewrite \cref{eq:1ows-goal-trace-moment} as
\begin{align}\label{eq:1ows-goal-compressO}
    \mathbb{E}_R\Tr(Y_R^{2t})=\sum_{z\in [2^S]}\bra{z}\bra{\emptyset}\widehat Y^{2t}\ket{z}\ket{\emptyset}.
\end{align}

\subsection{Bounding the trace moment on small databases}
\label{sec:1OWS-trace-moment-on-small-database}
\begin{lemma}\label{lem:1ows-l-database-t-query}
For $\widehat Y$ defined from the 1OWS security game, if normalized state $\ket{\psi}$ (with both the algorithm's register and the database register) is supported on databases of size $\le \ell$, then 
\[
|\bra\psi \widehat Y\ket\psi|\le c\cdot \frac{T^2+\sqrt{\ell}}{K}
\]
for $1\le\ell\le K^2$ and constant $c$ (actually we can have $c=10$; for $\ell=0$, we can replace $\sqrt{\ell}$ by $\max\{\sqrt{\ell},1\}$ from the $\ell=1$ case).
\end{lemma}

With \cref{lem:1ows-l-database-t-query}, now we can bound \cref{eq:1ows-goal-compressO}. 

\begin{proof}[Proof of \cref{thm:1ows-ts-tradeoff}]
For each $z$, each $\widehat Y$ will append the database by size at most $(2T+2)$. Therefore,  $\ket{\psi_z^i}:=\widehat Y^i\ket{z}\ket{\emptyset}$ is supported on database of size $\le (2T+2)i$, and for $P_{\le i}$ defined as the projector on database of size $\le i$,
\begin{align*}
    \|\ket{\psi_z^i}\|
    &\le 
    \left\|P_{\le (2T+2)i}\widehat Y P_{\le (2T+2)(i-1)}\ket{\psi_z^{i-1}}\right\|
    \\
    &\le
    \left\|P_{\le (2T+2)i}\widehat Y P_{\le (2T+2)i}\right\|\cdot \left\|\ket{\psi_z^{i-1}}\right\|
    \\
    &\le c\cdot \frac{T^2+\sqrt{(2T+2)i}}{K}\cdot \left\|\ket{\psi_z^{i-1}}\right\|
\end{align*}
Therefore, for each $z\in [2^S]$, $\bra{z}\bra{\emptyset}Y^{2t}\ket{z}\ket{\emptyset}
=\|\ket{\psi_z^t}\|^2$, and we can upper bound \cref{eq:1ows-goal-compressO} by
\[
\sum_{z\in [2^S]}\bra{z}\bra{\emptyset}Y^{2t}\ket{z}\ket{\emptyset}
\le 2^S\cdot \left(c\cdot \frac{T^2+\sqrt{t(2T+2)}}{K}\right)^{2t}.
\]
This will help us upper bound the winning probability in \cref{eq:1ows-goal-probability} as 
\begin{align*}
    \mathbb{E}_R\|Y_R\|\le 
2^{S/2t}\cdot c\cdot \frac{T^2+\sqrt{t(2T+2)}}{K}.
\end{align*}
By setting $t=\max\{S/2,1\}$, we have winning probability upper bounded by $O\left(\frac{T^2+1+\sqrt{S(T+1)}}{K}\right)$.
\end{proof}

Now we are only left to prove \cref{lem:1ows-l-database-t-query}.

\paragraph{Reformulation of \cref{lem:1ows-l-database-t-query}.} For normalized state $\ket{\psi}$ (with both the algorithm's register and the database register), by our definition of $\widehat Y$ in the compressed oracle database mode: 
\begin{align*}
    \left|\bra{\psi}\widehat Y\ket{\psi}\right|
    &=
    \left|
    \bra{\psi}
        \frac1N \cdot\mathbb{E}_k \sum_{x_1,x_2}\bra{x_1}\bra{0^w} (U^\dagger \cO^\dagger)^T \Pi_k (\cO \; U)^T \ket{x_2}\ket{0^w}\otimes T_{k,(x_1\oplus x_2)}
    \ket{\psi}
    \right|
    \\
    &=
    \mathbb{E}_k
    \left\|
    \Pi_k (\cO \; U)^T \left(\frac{1}{\sqrt{N}}\sum_x\ket{x}\otimes T_{k,x}\right)
    \ket{\psi}\ket{0^w}
    \right\|^2.
\end{align*}
This is exactly the winning probability of the following 1OWS search game:
\begin{enumerate}
    \item The algorithm starts with a normalized state $\ket{\psi}$ with support on databases of size $\le \ell$.
    \item The algorithm will receive a challenge input state, which is a single copy of $\ket{\psi_{R,k}}$, for a uniformly random unknown $k\gets [K]$.
    \item Then the algorithm is allowed to make $T$ queries to the (compressed) oracle $R$ ($O_R$), and it needs to guess the underlying challenge key $k$.
\end{enumerate}
Therefore, now we only need to prove the maximum winning probability of this game, and this will give a proof of \cref{lem:1ows-l-database-t-query}. That is, now we only need to prove the following \cref{lem:1ows-l-database-t-query-reform}:
\begin{lemma}[Reformulation of \cref{lem:1ows-l-database-t-query}]
\label{lem:1ows-l-database-t-query-reform}
    In the 1OWS security game where the algorithm starts with state
    $\sum_D\alpha_D\ket{\varphi_D}\ket{D}$ supported on databases of size $\le \ell$ in the compressed oracle database mode, the winning probability of the 1OWS game is bounded by $c\cdot \frac{T^2+\sqrt{\ell}}{K}$. 
\end{lemma}

\subsection{Proof of the 1OWS bound for states with initial database}
\label{sec:1OWS-bound-with-initial-database}
In this section, we will prove \cref{lem:1ows-l-database-t-query-reform} (and thus \cref{lem:1ows-l-database-t-query}). 

That is, start with state $\ket{\psi}=\sum_D\alpha_D \ket{\varphi_D}\ket{D}$ in the compressed oracle database mode, where $\sum_D |\alpha_D|^2=1$, $\{\ket{\varphi_D}\}_D$ is a set of normalized states on the algorithm's register, and $\ket{D}$ is the database register, now the goal is to upper bound
\begin{align}\label{eq:goal-for-l-t-search-bound}
    &\mathbb{E}_k
    \left\|
    \Pi_k (\cO \; U)^T \left(\frac{1}{\sqrt{N}}\sum_x\ket{x}\otimes T_{k,x}\right)
    \ket{\psi}
    \right\|^2
    \notag\\
    &\quad=
    \frac{1}{KN}\norm{
        \sum_{k,D,x}\Pi_k(\cO\cdot U)^T\cdot
        \alpha_D\ket{x}\ket{\varphi_D}\ket{D\oplus(k,x)}\ket{k}
    }^2,
\end{align}
where $\{\Pi_k\}$ are the final measurement projectors with $\sum_k\Pi_k=I$ and will only apply on the first 2 registers, and we also introduce a $K$-dimensional key register $\mathsf K$ with basis $\{\ket{k}:k\in [K]\}$ as the last register. As mentioned, we are summing over $D$ such that $|D|\leq\ell$.

\paragraph{Feynman path formulation.}
For $\cO$ as the compressed oracle query, it is defined as $\cO\ket{a,y}\ket{D}=\ket{a,y}\ket{D\oplus(a,y)}$ for a query on key $a$ and input $y$. Write $Q_{a,y}:=\ket{a,y}\bra{a,y}\otimes I$. We define $p_t=(a_j,y_j)_{j=1}^t$ as a Feynman path (for the query part) of length $t$, and define
\[
    \mathcal{O}_{p_t}:=Q_{a_t,y_t}U\cdots Q_{a_1,y_1}U,
    \qquad S_{p_t}:=\bigoplus_{i=1}^t(a_i,y_i).
\]
This operator $\mathcal O_{p_t}$ never touches the database register. From this, we can rewrite the state after $t$ queries as
\begin{align*}
    &\frac{1}{\sqrt{KN}}\sum_{k,D,x}(\cO\cdot U)^t\cdot \alpha_D\ket{x}\ket{\varphi_D}\ket{D\oplus(k,x)}\ket{k}
    \\
    &\quad =
    \frac{1}{\sqrt{KN}}\sum_{k,D,x}\sum_{(a_j,y_j)_j}
    \left[
    \left(\prod_{i=1}^t \cO\cdot (\ketbra{a_i,y_i}\otimes I)\cdot U\right)
    \cdot \alpha_D\ket{x}\ket{\varphi_D}\ket{D\oplus(k,x)}
    \right]
    \ket{k}
    \\
    &\quad =
    \frac{1}{\sqrt{KN}}\sum_{k,D,x}\sum_{p_t}
    \left(
    \mathcal{O}_{p_t}
    \cdot \alpha_D\ket{x}\ket{\varphi_D}
    \right)
    \ket{D\oplus(k,x)\oplus S_{p_t}}\ket{k}
    \\
    &\quad =
    \frac{1}{\sqrt{KN}}\sum_{k,S,x}\sum_{p_t}
    \left(
    \mathcal{O}_{p_t}
    \cdot \alpha_{S\oplus S_{p_t}}\ket{x}\ket{\varphi_{S\oplus S_{p_t}}}
    \right)
    \ket{S\oplus(k,x)}\ket{k}.
\end{align*}

For simplicity, we define
\[
    \ket{v_{t,S,x}}:=\sum_{p_t}\mathcal O_{p_t}\cdot
    \alpha_{S\oplus S_{p_t}}\ket{x}\ket{\varphi_{S\oplus S_{p_t}}}.
\]
Here the empty path, as the initial state, gives $\ket{v_{0,S,x}}=\alpha_S\ket{x}\ket{\varphi_S}$.
By unitarity, and by writing the last query separately, we have
\[
    \sum_S\norm{\ket{v_{t,S,x}}}^2=1,
    \qquad
    \ket{v_{t,S,x}}=\sum_{a,y}Q_{a,y}U\ket{v_{t-1,S\oplus(a,y),x}}.
\]

By orthogonality on the last 2 registers, we can rewrite our goal \cref{eq:goal-for-l-t-search-bound} as
\begin{align}\label{eq:goal-rewrite-for-l-t-search-bound}
    \frac{1}{KN}\sum_{k,S}
    \norm{\Pi_k\sum_x\ket{v_{T,S\oplus(k,x),x}}}^2.
\end{align}

\paragraph{Getting rid of ``large $S_k$''.}
Write $S_k:=\{x:(k,x)\in S\}$. For the sum over $k,S$ such that $|S_k|>\sqrt\ell+T$, we can upper bound it by first getting rid of $\Pi_k$. Write this upper bound as $H_T$, where we define $H_t$ as
\[
    H_t:=\frac{1}{KN}\sum_{\substack{k,S\\|S_k|>\sqrt\ell+t}}
    \norm{\sum_x\ket{v_{t,S\oplus(k,x),x}}}^2.
\]
We can write it as a sum by the last query, due to orthogonality. By letting $S'=S\oplus(a,y)$, as one query changes $|S_k|$ by at most 1, we get
\[
\begin{aligned}
    H_t
    &=\frac{1}{KN}
    \sum_{\substack{k,S',(a,y)\\|(S'\oplus(a,y))_k|>\sqrt\ell+t}}
    \norm{Q_{a,y}U\sum_x\ket{v_{t-1,S'\oplus(k,x),x}}}^2\\
    &\leq\frac{1}{KN}
    \sum_{\substack{k,S'\\|S'_k|>\sqrt\ell+t-1}}
    \norm{\sum_x\ket{v_{t-1,S'\oplus(k,x),x}}}^2
    =H_{t-1}.
\end{aligned}
\]
This gives us a way to go from path of length $t$ to $t-1$. Initially, the $\ket{x}$'s are orthogonal, so
\begin{align}\label{ineq:l-t-search-large-Sk}
    H_T\leq H_0
    =\frac{1}{KN}\sum_{\substack{k,S,x\\|S_k|>\sqrt\ell}}
      |\alpha_{S\oplus(k,x)}|^2
    \leq\frac{\ell}{K\lfloor\sqrt\ell\rfloor}
    \leq\frac{2\sqrt\ell}{K}.
\end{align}
Here $D=S\oplus(k,x)$ must satisfy $|D_k|\geq\lfloor\sqrt\ell\rfloor$, and since $|D|\leq\ell$, each $D$ is counted for at most $\ell/\lfloor\sqrt\ell\rfloor$ such keys.

\paragraph{Getting rid of ``final database with $x$''.}
By triangle inequality, the sum over all $x$ can be divided into 2 parts, based on whether $x\in S_k$. For the part where we sum over $x$ such that $x\in S_k$ and  $|S_k|\leq\sqrt\ell+T$, by Cauchy--Schwarz, we have
\begin{align}
    \frac{1}{KN}\sum_{\substack{k,S\\|S_k|\leq\sqrt\ell+T}}
    \norm{\Pi_k\sum_{x\in S_k}\ket{v_{T,S\oplus(k,x),x}}}^2
    &\leq\frac{\sqrt\ell+T}{KN}\sum_{k,S,x}
    \norm{\Pi_k\ket{v_{T,S\oplus(k,x),x}}}^2 \notag
    \\
    &=\frac{\sqrt\ell+T}{KN}\sum_{S',x}
    \norm{\ket{v_{T,S',x}}}^2
    =\frac{\sqrt\ell+T}{K}.
    \label{ineq:l-t-search-finaldb-with-x}
\end{align}
Here we let $S'=S\oplus(k,x)$, and then sum over the measurement projectors $\Pi_k$.

\paragraph{For the last part of ``final database without $x$''.}
We will upper bound the remaining part by getting rid of $\Pi_k$ and the restriction on $|S_k|$. Define $G_t$ as the following, and write it as a sum by the last query:
\[
\begin{aligned}
    G_t:=\frac{1}{KN}\sum_{k,S}
    \norm{\sum_{x\notin S_k}\ket{v_{t,S\oplus(k,x),x}}}^2
    =\frac{1}{KN}\sum_{k,S,(a,y)}
    \norm{Q_{a,y}U\sum_{x\notin S_k}
        \ket{v_{t-1,S\oplus(a,y)\oplus(k,x),x}}}^2.
\end{aligned}
\]
We can write the sum over $x\notin S_k$ as a sum of 2 parts, according to whether $(k,x)$ is the same as the last query $(a,y)$. Write $A_t$ and $B_t$ for the corresponding squared norms (including the factor $1/(KN)$) according to $x$, for $\{x:(k,x)\neq(a,y)\}$ and 
$\{x:(k,x)=(a,y)\}$, respectively. That is,
\begin{align*}
    A_t&:=
    \frac{1}{KN}\sum_{k,S,(a,y)}
    \norm{Q_{a,y}U\sum_{\substack{x\notin S_k,\\ (k,x)\neq (a,y)}}
        \ket{v_{t-1,S\oplus(a,y)\oplus(k,x),x}}}^2,
    \\
    B_t&:=
    \frac{1}{KN}\sum_{k,S,(a,y)}
    \norm{Q_{a,y}U\sum_{\substack{x\notin S_k,\\ (k,x)= (a,y)}}
        \ket{v_{t-1,S\oplus(a,y)\oplus(k,x),x}}}^2
    =
    \frac{1}{KN}\sum_{k,S}
    \norm{\sum_{\substack{x\notin S_k}} Q_{k,x}U\ket{v_{t-1,S,x}}}^2.
\end{align*}
By triangle inequality, this gives $\sqrt{G_t}\leq\sqrt{A_t}+\sqrt{B_t}$.

\textbf{For the part $(k,x)\neq(a,y)$ in $A_t$:}
We first go one query back by letting $S'=S\oplus(a,y)$. Then $x\notin S_k$ is the same as $x\notin S'_k$. Write the sum with $(k,x)\neq(a,y)$ as the unrestricted sum over $x\notin S'_k$ (which will correspond to $G_{t-1}$) minus the $\{x:x\in S'_k, (k,x)=(a,y)\}$ part (which will correspond to a newly defined term $C_t$): 
\begin{align*}
G_{t-1}&:=\frac{1}{KN}\sum_{k,S'}
    \norm{\sum_{x\notin S'_k}\ket{v_{t-1,S'\oplus(k,x),x}}}^2
=
\frac{1}{KN}\sum_{k,S',(a,y)}
    \norm{Q_{a,y}U\sum_{\substack{x\notin S'_k}}
        \ket{v_{t-1,S'\oplus(k,x),x}}}^2,
\\
C_t&:=
\frac{1}{KN}\sum_{k,S',(a,y)}
    \norm{Q_{a,y}U\sum_{\substack{x\notin S'_k,\\ (k,x)= (a,y)}}
        \ket{v_{t-1,S'\oplus(k,x),x}}}^2
=
\frac{1}{KN}\sum_{\substack{k,S'}}
    \norm{\sum_{x\notin S'_k} Q_{k,x}U\ket{v_{t-1,S'\oplus(k,x),x}}}^2.
\end{align*}
By triangle inequality and orthogonality of the query projectors, this gives $\sqrt{A_t}\leq\sqrt{G_{t-1}}+\sqrt{C_t}$.

\textbf{For the part $(k,x)=(a,y)$ in $B_t$:}
By letting $S'=S\oplus(k,x)$, where this $(k,x)$ is from the last query, we get
\[
    B_t=\frac{1}{KN}\sum_{\substack{k,S'}}
    \norm{\sum_{x\in S'_k}Q_{k,x}U\ket{v_{t-1,S'\oplus(k,x),x}}}^2.
\]

\textbf{Observation for $B_t+C_t$:}
By letting $S''=S'\oplus(k,x)$, as $\sum_kQ_{k,x}\preceq I$ for every $x$, we have
\[
\begin{aligned}
    B_t+C_t
    &=\frac{1}{KN}\sum_{k,S',x}
    \norm{Q_{k,x}U\ket{v_{t-1,S'\oplus(k,x),x}}}^2
    =\frac{1}{KN}\sum_{k,S'',x}
    \norm{Q_{k,x}U\ket{v_{t-1,S'',x}}}^2\\
    &\leq\frac{1}{KN}\sum_{S'',x}\norm{\ket{v_{t-1,S'',x}}}^2
    =\frac{1}{K}.
\end{aligned}
\]

\textbf{Concluding the induction of $G_T$:}
Now we can get an upper bound of $G_T$:
\[
    \sqrt{G_t}
    \leq\sqrt{G_{t-1}}+\sqrt{C_t}+\sqrt{B_t}
    \leq\sqrt{G_{t-1}}+\sqrt{2(B_t+C_t)}
    \leq\sqrt{G_{t-1}}+\sqrt{\frac{2}{K}}.
\]
Since we can also bound $G_0$ as
\[
    G_0=\frac{1}{KN}\sum_{k,D,x\in D_k}|\alpha_D|^2
    =\frac{1}{KN}\sum_D|\alpha_D|^2\cdot|D|
    \leq\frac{\ell}{KN},
\]
this will give an upper bound for $G_T$:
\begin{align}\label{ineq:l-t-search-finaldb-without-x}
    G_T \leq\left(\sqrt{G_0}+T\cdot\sqrt{\frac{2}{K}}\right)^2
    \leq2\cdot\left(\frac{\ell}{KN}+\frac{2T^2}{K}\right).
\end{align}

\paragraph{Conclusion.} From all of above, we can bound \cref{eq:goal-for-l-t-search-bound} through \cref{eq:goal-rewrite-for-l-t-search-bound} and
\begin{align*}
    &\frac{1}{KN}\sum_{k,S}
    \norm{\Pi_k\sum_x\ket{v_{T,S\oplus(k,x),x}}}^2
    =H_T + \frac{1}{KN}\cdot \sum_{\substack{k,S\\|S_k|\leq\sqrt\ell+T}}
    \norm{\Pi_k\sum_{x} \ket{v_{T,S\oplus(k,x),x}}}^2
    \\
    &\quad \le \frac{2\sqrt{\ell}}{K}
    +\frac{2}{KN}\cdot\left(
    \sum_{\substack{k,S\\|S_k|\leq\sqrt\ell+T}}
    \norm{\Pi_k\sum_{x\in S_k} \ket{v_{T,S\oplus(k,x),x}}}^2
    +
    \sum_{\substack{k,S\\|S_k|\leq\sqrt\ell+T}}
    \norm{\Pi_k\sum_{x\notin S_k} \ket{v_{T,S\oplus(k,x),x}}}^2
    \right)
    \tag{by \cref{ineq:l-t-search-large-Sk} and triangle inequality}
    \\
    &\quad \le \frac{2\sqrt{\ell}}{K}+\frac{2(\sqrt{\ell}+T)}{K}+2\cdot G_T\le \frac{10T^2}{K}+\frac{4\sqrt{\ell}}{K}+\frac{4\ell}{KN}.
    \tag{by \cref{ineq:l-t-search-finaldb-with-x} and \cref{ineq:l-t-search-finaldb-without-x}}
\end{align*}

\section{Non-uniform security of 1OWS in QROM, with classical advice}
\label{sec:time-advice-tradeoff-search-classical}

In this section, we bound the non-uniform security of 1OWS in QROM with classical advice. That is, the security game follows the same as in \cref{def:1ows-security-game}, but in the offline stage, an algorithm can only prepare an $S$-bit classical advice string that depends on $R$, before the online stage. The non-uniform security (in terms of $S$-bit, $T$-query) of 1OWS in QROM is defined as the maximum winning probability in this case, parametrized by $S,T,K,N$. we
show that any $T$-query quantum algorithm with $S$-bit classical advice can invert a single copy of a random
binary phase state with winning probability at most $O\left(\frac{T^2+1}{K}+\frac{S(T+1)}{KN}\right)$.

\begin{theorem}[Non-uniform security of 1OWS in QROM, with classical advice]\label{thm:1ows-ts-tradeoff-classical}
For integer $S,T\ge 0$ and $K\le N$, 
the non-uniform security (with $S$-bit advice and $T$-query) of 1OWS in QROM is $O\left(\frac{T^2+1}{K}+\frac{S(T+1)}{KN}\right)$; that is, the maximum winning probability is bounded by $O\left(\frac{T^2+1}{K}+\frac{S(T+1)}{KN}\right)$.
\end{theorem}

When \(T=o(\sqrt K)\), \Cref{thm:1ows-ts-tradeoff-classical} implies that constant success probability requires \(ST=\Omega(KN)\) with classical advice, whereas \Cref{thm:1ows-ts-tradeoff} gives the necessary condition \(ST=\Omega(K^2)\) with quantum advice. These conditions agree when \(K=N\). When \(K<N\), the classical-advice lower bound is stronger by a factor of \(N/K\).

\paragraph{Proof overview.}
The improvement for classical advice comes from fixing the advice string before applying the moment argument. For each fixed string $z$, the success probability is a scalar, and its compressed-oracle representation acts only on the database register. Thus, although the database may be in an arbitrary superposition, the algorithm's initial state is fixed and unentangled with it.

For the key step (\Cref{lem:1ows-l-database-t-query-classical}), the key observation is, there is a controlled change of purification on inaccessible registers (database register and challenge key register) that allows us to compare the challenge with a maximally mixed state independent of the hidden key (\Cref{lem:1ows-classical-controlled-unitary}).
\begin{itemize}
    \item From that, for a database of size at most $\ell$, the resulting initial-state error is $O(\sqrt{\ell/(KN)})$.
    \item The change of purification also conjugates the oracle used by the online algorithm, but the two oracles differ only on queries to the hidden row $k$. A query-by-query hybrid therefore contributes $O(T/\sqrt K)$ to the success amplitude.
\end{itemize}

\subsection{1OWS scalar-moment setup}
For any quantum oracle algorithm in the 1OWS security (search) game, denote its intermediate unitary between queries as $U$, final projective measurement as $\{\Pi_k\}_{k\in [K]}$, and suppose of hash function $R$ its advice is $\ket{z_R}$, where $z_R\in\{0,1\}^S$ is a classical string of length $S$. We can write its maximum winning probability as the following (with algorithm's workspace initialized as $\ket{0^w}$):
\[
\mathbb{E}_R\max_{z_R}\mathbb{E}_k\bra{z_R}
\bra{\psi_{R,k}}\bra{0^w}(U^\dagger O_R^\dagger)^T\Pi_k(O_RU)^T\ket{0^w}\ket{\psi_{R,k}}
\ket{z_R}.
\]
Define $Y_{R,z}:=\mathbb{E}_k\bra{z}
\bra{\psi_{R,k}}\bra{0^w}(U^\dagger O_R^\dagger)^T\Pi_k(O_RU)^T\ket{0^w}\ket{\psi_{R,k}}
\ket{z}$, and by Jensen's inequality, 
\begin{align}\label{eq:1ows-goal-probability-classical}
    \mathbb{E}_{R}\max_z Y_{R,z}
    \le \left(\mathbb{E}_{R}\max_z Y_{R,z}^{2t}\right)^{1/2t}\le \left(\sum_{z\in\{0,1\}^S}\mathbb{E}_{R} Y_{R,z}^{2t}\right)^{1/2t}.
\end{align}

\paragraph{Rewrite in compressed oracle framework.} For random function $R$, we can write all things above in the compressed oracle framework. Similarly we define $T_{k,x}:\ket{D}\mapsto \ket{D\oplus\{(k,x)\}}$. We also define $C_{k,z}$ for any fixed advice string $z$, where
\[
C_{k,z}:\ket{D}\mapsto \left(\frac{1}{\sqrt{N}}\sum_x \ket{x}\ket{z}\ket{0^w}T_{k,x}\right)\ket{D}.
\]
Then we can define $\widehat Y_z$ in the compressed oracle database mode as
\begin{align*}
    \widehat Y_z
    =\mathbb{E}_k C_{k,z}^\dagger (U^\dagger \cO^\dagger)^T \Pi_k (\cO \; U)^T C_{k,z}
\end{align*}
with compressed phase oracle $\cO$. This will help us rewrite the scalar moment as 
\begin{align}\label{eq:1ows-goal-compressO-classical}
    \mathbb{E}_R Y_{R,z}^{2t}=\bra{\emptyset}\widehat{Y}_z^{2t}\ket{\emptyset}.
\end{align}

\subsection{Bounding the scalar moment on small databases}
\begin{lemma}\label{lem:1ows-l-database-t-query-classical}
For $\widehat Y_z$ defined from the 1OWS security game with fixed classical advice string $z$, if normalized state $\ket{\chi}$ (only on the database register) is supported on databases of size $\le \ell$, then 
\[
|\bra\chi \widehat Y_z\ket\chi|\le c\cdot \left(\frac{T^2+1}{K}+\frac{\ell}{KN}\right)
\]
for constant $c$ (actually we can have $c=18$).
\end{lemma}

With \cref{lem:1ows-l-database-t-query-classical}, now we can bound \cref{eq:1ows-goal-compressO-classical}.
\begin{proof}[Proof of \cref{thm:1ows-ts-tradeoff-classical}]
    Each $\widehat Y_z$ will append the database by size at most $(2T+2)$. Therefore, $\ket{\chi^i}:=\widehat Y_z^i\ket{\emptyset}$ is supported on database of size $\le (2T+2)i$, and for $P_{\le i}$ defined as the projector on database of size $\le i$,
    \begin{align*}
    \|\ket{\chi^i}\|
    &\le 
    \left\|P_{\le (2T+2)i}\widehat Y_z P_{\le (2T+2)(i-1)}\ket{\chi^{i-1}}\right\|
    \\
    &\le
    \left\|P_{\le (2T+2)i}\widehat Y_z P_{\le (2T+2)i}\right\|\cdot \left\|\ket{\chi^{i-1}}\right\|
    \\
    &\le c\cdot \left(\frac{T^2+1}{K}+\frac{(2T+2)i}{KN}\right)\cdot \left\|\ket{\chi^{i-1}}\right\|
\end{align*}

Therefore, $\bra{\emptyset}\widehat Y_z^{2t}\ket{\emptyset}
=\|\ket{\chi^t}\|^2$, and we can upper bound \cref{eq:1ows-goal-compressO-classical} by
\[
\bra{\emptyset}\widehat Y_z^{2t}\ket{\emptyset}
\le \left(c\cdot \frac{(T^2+1)\cdot N+t(2T+2)}{KN}\right)^{2t}.
\]
This will help us upper bound the winning probability in \cref{eq:1ows-goal-probability-classical} as 
\begin{align*}
    \mathbb{E}_R\max_z Y_{R,z}\le 
2^{S/2t}\cdot c\cdot \left(\frac{T^2+1}{K}+\frac{t(2T+2)}{KN}\right).
\end{align*}
By setting $t=S/2$, we have winning probability upper bounded by $O\left(\frac{T^2+1}{K}+\frac{S(T+1)}{KN}\right)$.
\end{proof}

Now we are only left to prove \cref{lem:1ows-l-database-t-query-classical}.

\subsection{Proof of the 1OWS bound for an initial database state}
For normalized state $\ket{\chi}$ only on the database register, we can write it as $\ket{\chi}:=\sum_D \alpha_D\ket{D}$. With our definition of $\widehat Y_z$ in the compressed oracle database mode, we have
\begin{align*}
    \left|\bra{\chi}\widehat Y_z\ket{\chi}\right|
    &=
    \left|\bra{\chi}\mathbb{E}_k C_{k,z}^\dagger (U^\dagger \cO^\dagger)^T \Pi_k (\cO \; U)^T C_{k,z}\ket{\chi}\right|
    \\
    &=\mathbb{E}_k\left\|
        \Pi_k (\cO\; U)^T \left(\frac{1}{\sqrt{N}}\sum_x\ket{x}\ket{z}\ket{0^w}\otimes T_{k,x}\right)\left(\sum_D\alpha_D\ket{D}\right)
    \right\|^2.
\end{align*}
That is, now the goal is to upper bound the winning probability of this search game, starting with $\ket z\ket{0^w}\otimes\ket\chi$ in the database mode before receiving the challenge. Below we can view $\ket{z}$ as $\ket{z,0^w}$, since the upper bound we prove here in \cref{lem:1ows-l-database-t-query-classical} is independent of the length of $z$.

\paragraph{Real and reference states.}
Introduce an additional $N$-dimensional register
$\mathsf B$, initially in a fixed unit vector $\ket0$, and a $K$-dimensional key register
$\mathsf K$ with basis $\{\ket k:k\in[K]\}$. Together with the database register, these three registers are inaccessible to the algorithm. 

Define the following two normalized states, with registers:
\begin{align}
\ket{\xi_{\mathrm{real}}}
&:=\frac1{\sqrt{KN}}\sum_{k,x}
\ket x\ket z\otimes T_{k,x}\ket{\chi}\ket{0}_\mathsf B\ket{k}_\mathsf K,
\label{eq:1ows-classical-real-state}\\
\ket{\xi_{\mathrm{ref}}}
&:=\frac1{\sqrt{KN}}\sum_{k,x}
\ket x\ket z\otimes \ket\chi \ket{x}_\mathsf B\ket{k}_\mathsf K.
\label{eq:1ows-classical-ref-state}
\end{align}
The real state generates the challenge from $k$ (from $\mathsf{K}$), while in the reference state the challenge is maximally mixed and independent of the database (in particular, the reference state is a product with $\frac{1}{\sqrt{K}}\sum_k\ket k_{\mathsf K}$). 

The success projector on all the registers can be defined as $\Pi:=\sum_k\Pi_k\otimes I\otimes\ket k\bra k_{\mathsf K}$, where $I$ acts as the identity on the database register and $\mathsf B$. Thus the probability we want to bound is
\begin{align}
\label{eq:1ows-classical-real-success}
p:=\left|\bra\chi\widehat Y_z\ket\chi\right|
=\norm{\Pi(\cO\; U)^T\ket{\xi_{\mathrm{real}}}}^2.
\end{align}

For comparison, as the reference computation does not touch $\mathsf K$, its success probability is
\begin{align}
\label{eq:1ows-classical-reference-success}
\norm{\Pi(\cO U)^T\ket{\xi_{\mathrm{ref}}}}^2
=\frac1K\sum_k\norm{\Pi_k(\cO\; U)^T
\left(\frac1{\sqrt N}\sum_x\ket x\ket z\ket\chi\ket x\right)}^2
=\frac1K.
\end{align}

\paragraph{A controlled change of purification.}
For a database $D$, define $D_k:=\{x:(k,x)\in D\}$ as the database on $k$, and define
\[
\mu:=\sum_D|\alpha_D|^2|D|\le\ell.
\]
Thus $\mu$ is the expected total number of entries in the initial database. 
\begin{lemma}[Controlled change of purification]
\label{lem:1ows-classical-controlled-unitary}
For every normalized database vector $\ket\chi=\sum_D\alpha_D\ket D$, there exist
unitaries $\{W_k\}_k$ acting on the database register and $\mathsf B$, such that the unitary
\[
W:=\sum_{k\in[K]}W_k\otimes\ket k\bra k_{\mathsf K}
\]
has the following properties:
\begin{enumerate}
\item For every $a\ne k$ and $x\in[N]$, $[W_k,T_{a,x}]=0$, $[W,U]=0$, $[W,\Pi]=0$. Here $[A,B]:=AB-BA$ is the commutator. In particular, conjugating the oracle by $W$ (that is, $W \cO W^\dagger$) can change a query only when its queried key agrees with the key in $\mathsf K$.
\item The initial-state distance satisfies
\begin{align}
\label{ineq:1ows-classical-W-initial-gap}
\norm{W\ket{\xi_{\mathrm{real}}}-\ket{\xi_{\mathrm{ref}}}}
\le(1+\sqrt3)\sqrt{\frac{\mu}{KN}}.
\end{align}
\end{enumerate}
\end{lemma}

Proof of \Cref{lem:1ows-classical-controlled-unitary} will be left to \Cref{sec:proof-of-classical-controlled-unitary}. 

From \cref{lem:1ows-classical-controlled-unitary}, now we can prove \cref{lem:1ows-l-database-t-query-classical}.

\begin{proof}[Proof of \cref{lem:1ows-l-database-t-query-classical}]
Let $Q_{a,y}$ be the projector on the algorithm's query registers at address $(a,y)$. Within the compressed oracle framework (with compressed phase oracle),
\[
\cO=\sum_{a\in[K],\,y\in[N]}Q_{a,y}\otimes T_{a,y}.
\]
Define the projectors $\{Q_k\}_k$ on the queried key and the controlled matching projector $Q$ as
\[
Q_k:=\sum_yQ_{k,y},\qquad
Q:=\sum_kQ_k\otimes I\otimes\ket k\bra k_{\mathsf K}.
\]
Then $\sum_kQ_k\preceq I$. Write
\[
\widetilde{\cO}:=W\cO W^\dagger.
\]
It is a unitary. By the first property in \cref{lem:1ows-classical-controlled-unitary}, it agrees with $\cO$
on the complement of $Q$, $\widetilde{\cO}-\cO
=(\widetilde{\cO}-\cO)Q$. Thus for every vector $\ket v$, by triangle inequality,
\begin{align}
\norm{(\widetilde{\cO}-\cO)\ket v}=\norm{(\widetilde{\cO}-\cO)Q\ket{v}}
\le2\norm{Q\ket v}.
\label{ineq:1ows-classical-query-locality}
\end{align}
Also, since $[W,U]=0$, $W(\cO U)^T=(\widetilde{\cO}U)^T W$.

\paragraph{The query hybrid.}
For $0\le i\le T$, define
\begin{align*}
    \ket{\Psi_0}&:=(\cO U)^{T}\ket{\xi_{\mathrm{ref}}},
    \\
    \ket{\Psi_i}&:=(\widetilde{\cO}U)^i(\cO U)^{T-i}\ket{\xi_{\mathrm{ref}}},
    \\
    \ket{\Psi_{T+1}}&:=W(\cO U)^T\ket{\xi_{\mathrm{real}}}.
\end{align*} 
For $0\le i<T$, cancelling the common
unitary suffix and using \cref{ineq:1ows-classical-query-locality}, we have
\begin{align*}
\norm{\ket{\Psi_{i+1}}-\ket{\Psi_i}}
&=\norm{(\widetilde{\cO}-\cO)U(\cO U)^{T-i-1}\ket{\xi_{\mathrm{ref}}}}\\
&\le2\norm{QU(\cO U)^{T-i-1}\ket{\xi_{\mathrm{ref}}}}.
\end{align*}
Since the reference computation never touches $\mathsf K$ before $Q$, on register $\mathsf K$ is still a product with $\frac{1}{\sqrt{K}}\sum_k\ket k$.
Therefore,
\[
\norm{QU(\cO U)^{T-i-1}\ket{\xi_{\mathrm{ref}}}}^2
=\sum_k\norm{(Q_k\otimes \ketbra{k})U(\cO U)^{T-i-1}\ket{\xi_{\mathrm{ref}}}}^2\le\frac1K.
\]
Summing over $\{i:0\le i<T\}$, we have
\begin{align}
\label{ineq:1ows-classical-BBBV-bound}
\norm{\ket{\Psi_T}-\ket{\Psi_0}}
\le\sum_{i=0}^{T-1}\norm{\ket{\Psi_{i+1}}-\ket{\Psi_i}}
\le\frac{2T}{\sqrt K}.
\end{align}
For the last step, \cref{ineq:1ows-classical-W-initial-gap} gives
\begin{align}
\norm{\ket{\Psi_{T+1}}-\ket{\Psi_T}}
&=\norm{W\ket{e_{\mathrm{real}}}-\ket{e_{\mathrm{ref}}}}\notag\\
&\le(1+\sqrt3)\sqrt{\frac{\mu}{KN}}.
\label{ineq:1ows-classical-last-hybrid}
\end{align}

\paragraph{Conclusion.}
Since $W$ commutes with $\Pi$, the real success amplitude is $\sqrt p=\norm{\Pi\ket{\Psi_{T+1}}}$. Therefore, by
\cref{eq:1ows-classical-reference-success,ineq:1ows-classical-BBBV-bound,ineq:1ows-classical-last-hybrid},
\begin{align*}
\sqrt p
&\le\norm{\Pi\ket{\Psi_0}}+\norm{\ket{\Psi_{T+1}}-\ket{\Psi_0}}\\
&\le\frac{2T+1}{\sqrt K}+(1+\sqrt3)\sqrt{\frac{\mu}{KN}}\\
&\le\frac{2T+1}{\sqrt K}+(1+\sqrt3)\sqrt{\frac{\ell}{KN}}\\
&\le \sqrt{18\cdot \left(\frac{T^2+1}{K}+\frac{\ell}{KN}\right)}.
\end{align*}
\end{proof}

\section{Quantum bit commitments in QROM}
\label{sec:commitments}

In this section, we construct a computationally binding and statistically hiding canonical quantum bit commitment scheme in QROM, from the non-uniform security of single-copy one-way states in \Cref{sec:time-advice-tradeoff-search}. 
Applying the conversion of
\cite[Theorem~6.1]{EC:HhaMorYam23}, this will also give a computationally hiding and
statistically binding canonical quantum bit commitment scheme.

\paragraph{Construction.}
Let $R:[K]\times [N]\to\{\pm 1\}$ be the random function, where $K=N^{1-\Omega(1)}$. Let $\mathsf{C}$ and $\mathsf{W}$ be
polynomial-size registers as the commitment register and
the reveal register, respectively. We define polynomial-time
computable unitaries $Q_0^R$ and $Q_1^R$ on $\mathsf{C}\mathsf{W}$ as
\begin{equation}
\begin{aligned}\label{eq:commitment-states}
    Q_0^R\lvert 0\rangle_{\mathsf{C}\mathsf{W}}
        &= \lvert\Psi_0^R\rangle_{\mathsf{C}\mathsf{W}}
         := \frac{1}{\sqrt{K}}\sum_{k\in\cK}\ket{k}_\mathsf{C}\ket{\psi_{R,k}}_\mathsf{W},\\
    Q_1^R\lvert 0\rangle_{\mathsf{C}\mathsf{W}}
        &= \lvert\Psi_1^R\rangle_{\mathsf{C}\mathsf{W}}
         := \frac{1}{\sqrt{K}}\sum_{k\in\cK}\ket{k}_\mathsf{C}\ket{k}_\mathsf{W}.
\end{aligned}
\end{equation}
where $\{\ket{\psi_{R,k}}\}_{k\in\cK}$ are the set of binary phase random states defined from the 1OWS security game \Cref{def:1ows-security-game}. Both $Q_0^R$ and $Q_1^R$ are efficiently computable: for $Q_0^R\ket{0}$, one would only need to prepare $\propto\sum_{k,x}\ket{k,x}$ and then make 1 query to $R$. 

To commit to $b\in\{0,1\}$, the sender prepares
$\lvert\Psi_b^R\rangle_{\mathsf{C}\mathsf{W}}$, sends $\mathsf{C}$ to the receiver, and
keeps $\mathsf{W}$. To open, the sender sends $b$ together with witness register $\mathsf{W}$. The receiver applies $(Q_b^R)^\dagger$ to $\mathsf{C}\mathsf{W}$ and accepts if measuring all qubits
in the computational basis gives the all-zero outcome.
Equivalently, the acceptance projector is
\(\Pi_b^R:= \lvert\Psi_b^R\rangle\langle\Psi_b^R\rvert_{\mathsf{CW}}\).

\begin{theorem}
\label{thm:shcb-commitment-security}
The scheme $(Q_0^R,Q_1^R)$
is a perfectly correct, statistically hiding, and computationally binding canonical quantum bit commitment in QROM. 

Specifically, any adversary breaking binding that starts with an $S$-qubit advice and makes $T$ queries to $O_R$ can win with probability no more than $O\left(\frac{T^2+\sqrt{S(T+1)}}{K}\right)$.
\end{theorem}

\begin{proof}
For correctness, an honest opening is accepted with probability one, since
$(Q_b^R)^\dagger\lvert\Psi^R_b\rangle=\lvert 0\rangle$.

For statistical hiding, let $\rho_b^R:= \operatorname{Tr}_\mathsf{W} \bigl(\lvert\Psi^R_b\rangle\langle\Psi^R_b\rvert\bigr)$ be the receiver's state after the commitment phase. We show that $D(\rho_0^R,\rho_1^R):=\frac12\|\rho_0^R-\rho_1^R\|_1 = \operatorname{negl}(\lambda)$. Note that
\[
\rho_0^R-\rho_1^R=\frac1K\sum_{k\neq\ell}\langle\psi_{R,\ell}|\psi_{R,k}\rangle\,|k\rangle\langle\ell|.
\]
Using $\|A\|_1^2\le K\operatorname{Tr}(A^\dagger A)$ for a $K\times K$ matrix, we obtain
\[
\mathbb E_R \frac14\left\|\rho_0^R-\rho_1^R\right\|_1^2
\le \frac K4\,\mathbb E_R\operatorname{Tr}[(\rho_0^R-\rho_1^R)^2]
=\frac1{4K}\sum_{k\ne\ell}
\mathbb E_R|\langle\psi_{R,\ell}|\psi_{R,k}\rangle|^2
=\frac{K-1}{4N}.
\]
By Markov's inequality, $\Pr_R[D(\rho_0^R,\rho_1^R)>\varepsilon]
\le \frac{K-1}{4N\varepsilon^2}$, and thus for $K=N^{1-\Omega(1)}$, $D(\rho_0^R,\rho_1^R)=\negl(n)$ with $1-\negl(n)$ probability.

For computational binding, for any polynomial-time computable unitary $U^R$ (with oracle access to $R$) acting on $\mathsf{W}\mathsf{Z}$,
where $\mathsf{Z}$ is an arbitrary polynomial-size auxiliary register,
and any auxiliary state $\lvert\tau^R\rangle_\mathsf{Z}$, 
from \Cref{def:commitment-binding}, it suffices to show that
\[
\left\|
        (\langle\Psi_1^R\rvert_{\mathsf{C}\mathsf{W}}\otimes I_\mathsf{Z})
        (I_\mathsf{C}\otimes U^R_{\mathsf{W}\mathsf{Z}})
        (\lvert\Psi_0^R\rangle_{\mathsf{C}\mathsf{W}}\otimes\lvert\tau^R\rangle_\mathsf{Z})
    \right\|
    = \negl(\lambda).
\]
From \cref{eq:commitment-states}, by Cauchy-Schwarz, 
\begin{align*}
    \left\|
        (\langle\Psi_1^R\rvert_{\mathsf{C}\mathsf{W}}\otimes I_\mathsf{Z})
        (I_\mathsf{C}\otimes U^R_{\mathsf{W}\mathsf{Z}})
        (\lvert\Psi_0^R\rangle_{\mathsf{C}\mathsf{W}}\otimes\lvert\tau^R\rangle_\mathsf{Z})
    \right\|^2
    \le 
    \frac{1}{K}\sum_k 
    \left\|
        (\langle k\rvert_{\mathsf{W}}\otimes I_\mathsf{Z})
        U^R_{\mathsf{W}\mathsf{Z}}
        (\lvert\psi_{R,k}\rangle_{\mathsf{W}}\otimes\lvert\tau^R\rangle_\mathsf{Z})
    \right\|^2.
\end{align*}
If we write $U^R$ in terms of an oracle query algorithm with $T$ queries to $O_R$, this is exactly the non-uniform security of the 1OWS game defined in \Cref{def:1ows-security-game}. Therefore, from \Cref{thm:1ows-ts-tradeoff}, for any $T$-query algorithm in QROM, with auxiliary register $\ket{\tau^R}_{\mathsf Z}$ of $S$ qubits, this is upper bounded by $O\left(\frac{T^2+\sqrt{S(T+1)}}{K}\right)$.
\end{proof}

Similarly from \Cref{thm:1ows-ts-tradeoff-classical}, we can bound the security when the adversary's auxiliary register $\ket{\tau^R}$ is restricted to a classical string. 
\begin{theorem}[Binding security with classical auxiliary input]
\label{thm:shcb-commitment-security-classical-advice}
The scheme $(Q_0^R,Q_1^R)$
is a perfectly correct, statistically hiding, and computationally binding canonical quantum bit commitment in QROM. Any adversary breaking binding that starts with an $S$-bit advice and makes $T$ queries to $O_R$ can win with probability no more than $O\left(\frac{T^2}{K}+\frac{S(T+1)}{KN}\right)$.
\end{theorem}

\paragraph{Obtaining computational hiding.}
The flavor-conversion theorem \cite[Theorem~6.1]{EC:HhaMorYam23} converts a computationally binding and statistically hiding canonical quantum bit commitment scheme into a computationally hiding and statistically binding canonical quantum bit commitment scheme, where the computational security goes from $\epsilon$ to $\sqrt{\epsilon}$. Hence we obtain the following consequence.

\begin{corollary}
\label{cor:statistically-binding-commitment}
There exists a computationally hiding and statistically binding canonical quantum bit commitment scheme. 

Specifically, any adversary breaking hiding that starts with an $S$-qubit advice and makes $T$ queries to $O_R$ can win with advantage no more than $O\left(\frac{T+(S(T+1))^{1/4}}{\sqrt{K}}\right)$. If the advice is restricted to an $S$-bit classical string, the advantage over $1/2$ is bounded by $O\left(\frac{T}{\sqrt{K}}+\sqrt{\frac{S(T+1)}{KN}}\right)$.
\end{corollary}

\begin{proof}
Apply \cite[Theorem~6.1]{EC:HhaMorYam23} to the scheme of
\Cref{thm:shcb-commitment-security}. Computational binding of the original scheme gives computational hiding of the converted
scheme, while statistical hiding of the original scheme gives statistical binding of the converted scheme. Then the security follows from \Cref{thm:shcb-commitment-security} and \Cref{thm:shcb-commitment-security-classical-advice}.
\end{proof}

\section{Non-uniform security of PRGs in QROM}

In this section, we prove a tighter non-uniform security of pseudorandom generators in the quantum random oracle model via the trace moment method. Specifically, we show that any $T$-query quantum algorithm with $S$-qubit advice can distinguish between a random image and a random element in the range, with an winning probability at most $\frac12+O\left(\frac{T^2}{N}+\sqrt{\frac{S(T+1)}{N}}\right)$, which improves the bound in \cite{EC:Liu23}.
In particular, the current best classical result (distinguishing advantage bounded by $O(T/N+\sqrt{S(T+1)/N})$ for classical queries and classical advice, from \cite{EC:CDGS18}), it recovers the same dependence $\sqrt{S(T+1)/N}$, while allowing both quantum advice and quantum oracle access. 

\begin{definition}[PRG security game]\label{def:prg-security-game}
In the non-uniform security game of PRG:
\begin{itemize}
    \item In the offline stage, a hash function $R:[N]\to[M]$ is sampled and fixed, and an algorithm is allowed to have an $S$-qubit advice state that depends arbitrarily on $R$ before receiving the challenge.
    \item Then, in the online stage, a challenger will sample a challenge bit $b\gets \{0,1\}$.
    \begin{itemize}
        \item If $b=0$, the challenger samples $x\gets [N]$ and sends $R(x)$ to the algorithm.
        \item If $b=1$, the challenger samples $y\gets[M]$ and sends $y$ to the algorithm.
    \end{itemize}
    \item The algorithm, after receiving the challenge, is allowed to make $T$ queries to $R$ and is asked to output a guess $b'$. It succeeds if $b'=b$. 
\end{itemize}
We consider security game in QROM, for $R$ as a random oracle, and the algorithm can have oracle access to $O_R$. The non-uniform security (in terms of $S$-qubit, $T$-query) of PRG in QROM is defined as the maximum distinguishing advantage, parametrized by $(S,T,N,M)$.
\end{definition}

\begin{theorem}[Non-uniform security of PRG in QROM]\label{thm:prg-ts-tradeoff}
The non-uniform security (with $S$-qubit advice and $T$-query) of PRG in QROM is $O\left(\frac{T^2}{N}+\sqrt{\frac{S(T+1)}{N}}\right)$; that is, the maximum distinguishing advantage is bounded by $O\left(\frac{T^2}{N}+\sqrt{\frac{S(T+1)}{N}}\right)$.
\end{theorem}

\paragraph{Proof overview.} 
Via the trace-moment method, we first express the distinguishing advantage as the expected norm of a Hermitian operator on the advice register and bound it by a trace moment (\Cref{sec:prg-trace-moment}). Then, viewing everything in the compressed oracle database mode, the main step is to bound the corresponding distinguishing gap when the algorithm starts with a state with compressed databases of size at most $\ell$ (\Cref{sec:prg-trace-moment-on-small-database}).

To prove this bound (\Cref{lem:prg-l-database-t-query}), we introduce a resampling hybrid: instead of generating the challenge as $R(x)$, we sample an independent challenge $r$ and
replace the oracle value at the hidden seed $x$ by $r$. 
\begin{itemize}
    \item Compared with the real experiment, one can tell the difference only in the case where the initial database contains $x$. Averaging over the hidden seed bounds this contribution by $O(\sqrt{\ell/N})$ for database of size $\le \ell$.
    \item Compared with an ideal experiment, where the challenge is independent of the oracle, a modified BBBV-style \cite{BBBV97} hybrid argument will give a $O(T^2/N)$ bound.
\end{itemize}
Finally, a trace moment of order $S$ gives databases of size about $O(ST)$, and this gives our bound.

\subsection{PRG trace-moment setup}
\label{sec:prg-trace-moment}
For any quantum oracle algorithm in the distinguishing game, denote its intermediate unitary between queries as $U$, final distinguishing measurement as $0\preceq\Pi\preceq I$, and suppose on hash function $R$ its advice is $\ket{\phi_R}$. Define $A_R=(O_RU)^T$, we can write its maximum distinguishing advantage as the following (with algorithm's workspace of $w$ qubits initialized as $\ket{0^w}_\mathsf W$):
\begin{align*}
    \mathbb{E}_R \max_{\phi_R}
    \left|\mathbb{E}_x
    \bra{\phi_R,0^w}\bra{R(x)}A_R^\dagger \Pi A_R\ket{R(x)} 
    \ket{\phi_R,0^w}
    - 
    \mathbb{E}_r\bra{\phi_R,0^w}
    \bra{r}A_R^\dagger \Pi A_R\ket{r} 
    \ket{\phi_R,0^w}
    \right|.
\end{align*}
From the definition of the operator norm, this is exactly equal to
\begin{align}\label{eq:prg-goal-advantage}
    \mathbb{E}_R \left\|Y_R\right\|:=
\mathbb{E}_R\left\|
    \mathbb{E}_x
    \bra{R(x),0^w}A_R^\dagger \Pi A_R\ket{R(x),0^w}
    - 
    \mathbb{E}_r
    \bra{r,0^w}A_R^\dagger \Pi A_R\ket{r,0^w}
    \right\|.
\end{align}
Here we define $Y_{R}$ as the matrix inside the norm. Now our goal is to bound $\mathbb{E}_{R}\|Y_{R}\|$.

Via the trace moment method, by Jensen's inequality, we have
\begin{align*}
    \mathbb{E}_{R}\|Y_{R}\|
    \le \left(\mathbb{E}_{R}\|Y_{R}\|^{2t}\right)^{1/2t}\le \left(\mathbb{E}_{R}\Tr(Y_{R}^{2t})\right)^{1/2t}.
\end{align*}
Therefore, now our goal is to upper bound
\begin{align}\label{eq:prg-goal-trace-moment}
    \mathbb{E}_{R} \Tr(Y_{R}^{2t})=\sum_{z\in [2^S]} \mathbb{E}_R\bra{z}Y_R^{2t}\ket{z}.
\end{align}

\paragraph{Rewrite in compressed oracle framework.} For random function $R$, we can write all things above in the compressed oracle framework. 

Note that in the un-compressed world, the isometry that writes $R(x)$ into a fresh challenge register can be viewed as   
\[
\ket{R}\mapsto 
\ket{R(x)}\otimes\ket{R}=
\frac{1}{\sqrt{M}}\sum_{y} H\ket{y}\otimes (-1)^{y\cdot R(x)}\ket{R}.
\]
Define $T_{x,y}$ on the compressed oracle register that flips the presence of the
database entry $(x,y)$, that is, $T_{x,y}\ket{D}=\ket{D\oplus \{(x,y)\}}$.
Now the isometry that generates the fresh challenge using
one compressed query can be viewed as
\[
\ket{D}\mapsto  \left(\frac{1}{\sqrt{M}}\sum_yH\ket{y}\ket{0^w}\otimes T_{x,y}\right) \ket{D}.
\]
Therefore, by having $\widehat A:= (\cO\cdot U)^T$, we can define $\widehat Y$ in the compressed oracle database mode as
\begin{align*}
    \widehat Y
    =\frac1M \cdot\mathbb{E}_x \sum_{y_1,y_2}\bra{y_1,0^w}H \widehat A^\dagger \Pi \widehat A H\ket{y_2,0^w}\otimes T_{x,(y_1\oplus y_2)}
    -
    \mathbb{E}_r
    \bra{r,0^w}\widehat A^\dagger \Pi \widehat A\ket{r,0^w}\otimes I_D
\end{align*}
with compressed phase oracle $\cO$. This will help us rewrite \cref{eq:prg-goal-trace-moment} as
\begin{align}\label{eq:prg-goal-compressO}
    \mathbb{E}_R\Tr(Y_R^{2t})=\sum_{z\in [2^S]}\bra{z}\bra{\emptyset}\widehat Y^{2t}\ket{z}\ket{\emptyset}.
\end{align}

\subsection{Bounding the trace moment on small databases}
\label{sec:prg-trace-moment-on-small-database}

\begin{lemma}\label{lem:prg-l-database-t-query}
For $\widehat Y$ defined from the PRG security game, if normalized state $\ket{\psi}$ (with both the algorithm's register and the database register) is supported on databases of size $\le \ell$, then
\[
|\bra\psi \widehat Y\ket\psi|\le 
8\cdot \frac{T^2}{N}+4\cdot \sqrt{\frac{\ell}{N}}=O\left(\frac{T^2}{N}+\sqrt{\frac{\ell}{N}}\right)
\]
\end{lemma}

With \cref{lem:prg-l-database-t-query}, now we can bound \cref{eq:prg-goal-compressO}. 

\begin{proof}[Proof of \cref{thm:prg-ts-tradeoff}]
For each $z$, each $Y$ will append the database by size at most $(2T+2)$. Therefore,  $\ket{\psi_z^i}:=\widehat Y^i\ket{z}\ket{\emptyset}$ is supported on database of size $\le (2T+2)i$, and for $P_{\le i}$ defined as the projector on database of size $\le i$,
\begin{align*}
    \|\ket{\psi_z^i}\|
    &\le 
    \left\|P_{\le (2T+2)i}\widehat Y P_{\le (2T+2)(i-1)}\ket{\psi_z^{i-1}}\right\|
    \\
    &\le
    \left\|P_{\le (2T+2)i}\widehat Y P_{\le (2T+2)i}\right\|\cdot \left\|\ket{\psi_z^{i-1}}\right\|
    \\
    &\le 8\cdot \left(\frac{T^2}{N}+\sqrt{\frac{(2T+2)i}{N}}\right)\cdot \left\|\ket{\psi_z^{i-1}}\right\|
\end{align*}
Therefore, for each $z\in [2^S]$, $\bra{z}\bra{\emptyset}Y^{2t}\ket{z}\ket{\emptyset}
=\|\ket{\psi_z^t}\|^2$, and we can upper bound \cref{eq:prg-goal-compressO} by
\[
\sum_{z\in [2^S]}\bra{z}\bra{\emptyset}Y^{2t}\ket{z}\ket{\emptyset}
\le 2^S\cdot 8^{2t}\left(\frac{T^2}{N}+\sqrt{\frac{t(2T+2)}{N}}\right)^{2t}.
\]
This will help us upper bound the distinguishing advantage in \cref{eq:prg-goal-advantage} as 
\begin{align*}
    \mathbb{E}_R\|Y_R\|\le 
2^{S/2t}\cdot 8\cdot \left(\frac{T^2}{N}+\sqrt{\frac{t(2T+2)}{N}}\right).
\end{align*}
By setting $t=S/2$, we have distinguishing advantage $O\left(\frac{T^2}{N}+\sqrt{\frac{S(T+1)}{N}}\right)$.
\end{proof}

Now we are only left to prove \cref{lem:prg-l-database-t-query}, which is left to \cref{sec:append-prg}.
\section{One-query unitary synthesis lower bound from compressed oracles}
\label{sec:unitary-synthesis}

In this section, we apply the trace-moment method to two oracle state games to reprove one-query unitary synthesis lower bounds: the oracle state search game \cite{DLM26} and the oracle state distinguishing game \cite{STOC:LomMaWri24}. 
We give a common compressed-oracle approach to bounding those security, recovering the asymptotic bounds of \cite{DLM26,STOC:LomMaWri24}.

To apply the trace moment method, both arguments begin with the weight-vector decomposition, which converts the optimization over the arbitrary phase oracle into a spectral relaxation. We then bound the resulting operator norm through trace moments in the compressed-oracle representation.

\subsection{Oracle state search game hardness}
\label{sec:unitary-synthesis-search}
In this subsection, we upper bound a search game winning probability of a one-query (unitary synthesis) adversary, which gives an alternate proof for Theorem A.1 in \cite{DLM26}. Instead of using matrix concentration inequalities for matrix Rademacher series, we upper bound the winning probability via the compressed oracle techniques. Together with the reduction in \cite[Section~4.1]{DLM26}, this will give a one-query unitary synthesis lower bound by the compressed oracle techniques.

\begin{definition}[Oracle state search game, from {\cite[Definition~4.1]{DLM26}}]\label{def:state-search-game}
Fix a random variable $\mathbf R$ and, for every $R$ in its support, define a family of normalized states $\bigl\{\ket{\psi_{R,k}}:k\in[K]\bigr\}\subseteq \mathbb{C}^{N}$.
Without loss of generality, we may describe $\ket{\psi_{R, k}}$ in the computational basis with the following normalization:
\[ \ket{\psi_{R, k}} = \frac 1 {\sqrt N}\sum_{x\in [N]} R(k,x) \ket{x}
\]
for some random variables $\{\mathbf R(k,x) \in \mathbb C\}_{k\in [K], x\in [N]}$. 

In the oracle state search game, the challenger samples $R$ together with a uniformly random key $k\in[K]$, gives the adversary one copy of $\ket{\psi_{R,k}}$, and the adversary must output $k$ after making one oracle query. The oracle can depend on the variable $R$ but not on $k$. 
\end{definition}

Together with the one-query normal form of a unitary synthesis algorithm as in \Cref{def:one-query-algorithm}, now we can denote the winning probability of this search game conditioned on sampling $R$ as
\begin{align*}
    \Win(\cA\mid R):= \max_f\mathbb{E}_k\left\|
    \Pi_k \cdot O_f \cdot V\ket{\psi_{R,k}}
    \right\|^2,
\end{align*}
where the final projective measurement $\{\Pi_k\}_{k\in [K]}$ is also independent of $R$, and the post query unitary $U$ can be absorbed in the projective measurement by defining $\{\Pi'_k:=U^\dagger \Pi_k U\}_k$. $O_f$ is the arbitrary phase oracle, where the underlying boolean function is $f$.

\begin{theorem}[One-query search bound for binary phase states, {\cite[Theorem~A.1]{DLM26}}]
\label{thm:binary-phase-search}
For the oracle state search game associated with the family $\{\ket{\psi_{R,k}}\}_{k\in[K]}$ defined as in \Cref{def:state-search-game} with $R=(R(k,x))_{k\in[K],\,x\in[N]}$ as a family of independent Rademacher random variables, every one-query adversary with workspace dimension $M$ satisfies
\[
\mathbb{E}_R\bigl[\Win(\cA\mid R)\bigr]
\le \frac{2\bigl(1+\log(2M)\bigr)}{K}.
\]
\end{theorem}

\paragraph{Proof overview.}
By the weight-vector decomposition together with the spectral relaxation, the search probability can be upper bounded by the expected norm of an operator on the algorithm's register. Via the trace-moment method and rewriting in the compressed oracle database mode, the key remaining step is the small-database analysis (\Cref{lem:1-query-compressed-oracle-search}).

After the weight vector decomposition, we can view the reweighting diagonal operator $D_{V,R,k}$ (or $\widehat D_k$ in compressed database mode) as generating a challenge, and similarly, we can group terms by whether the final database contains the ``challenge entry''; that is, view $\widehat D_k$ as a sum of operators $D_x$ for different $x$, and for the $x$ part with $D_x$ applying on states, we group according to whether $(k,x)$ is in the final database.  
\begin{itemize}
    \item If $(k,x)$ is in the final database, note that the final database on challenge key $k$ is bounded by $\ell$. Therefore, in the corresponding coherent sum over $x$ inside the final database, paying only $O(\ell)$ can make it into a sum of individual squared norms, which would be easy to upper bound.
    \item If $(k,x)$ is not in the final database, this means $(k,x)$ should appear in the initial database, where one would expect to bound it by $\ell/K$.
\end{itemize}
Since each application of the moment operator changes the database size by at most two, a moment of order $1+\log M$ gives the desired asymptotic search bound.

\subsubsection{Spectral relaxation and the compressed oracle trace moment}
\label{sec:trace-moment-search}

Our goal is to upper bound the search game success probability of a one-query adversary, which is to upper bound
\[
\mathbb{E}_R[\Win(\cA\mid R)]= 
\mathbb{E}_R\max_f\mathbb{E}_k\left\|
    \Pi_k \cdot O_f \cdot V\ket{\psi_{R,k}}
    \right\|^2.
\]
By \Cref{lem:weight-vector-decomposition}, we can apply a generic spectral relaxation as in \cite[Section~6.1]{DLM26}:
\begin{align*}
    \mathbb{E}_R\max_f\mathbb{E}_k\left\|
    \Pi_k \cdot O_f \cdot V\ket{\psi_{R,k}}
    \right\|^2
    &=
    \mathbb{E}_R\max_f\mathbb{E}_k\left\|
    \Pi_k D_{R,k}\cdot O_f \ket{\wt_V}
    \right\|^2
    \\
    &=
    \mathbb{E}_R\max_f\bra{\wt_V}O_f^\dagger\cdot 
    \left(\mathbb{E}_k D_{R,k}^\dagger \Pi_k D_{R,k}\right) \cdot O_f\ket{\wt_V}
    \\
    &\le 
    \mathbb{E}_R \norm{\mathbb{E}_k D_{R,k}^\dagger \Pi_k D_{R,k}}.
\end{align*}
here we omit $V$ in $D_{R,k}:=D_{V,R,k}$ from \Cref{lem:weight-vector-decomposition}.
Define
\[
    W_R := \frac1K\sum_{k=1}^K D_{R,k}^\dagger \Pi_k D_{R,k}.
\]
Then $W_R\succeq 0$, and for every positive integer $t$,
\begin{align}
    \mathbb{E}_R \norm{W_R}
    \leq \Big(\mathbb{E}_R \norm{W_R}^{2t}\Big)^{1/(2t)}\leq \Big(\mathbb{E}_R \Tr(W_R^{2t})\Big)^{1/(2t)} . \label{ineq:trace-moment-search}
\end{align}

Now we only need to upper bound this trace moment, and we rewrite it within the compressed oracle methodology. 

\paragraph{Rewrite in compressed oracle framework. } First we write $D_{R,k}$ as
\[
    D_{R,k}=\sum_{x\in[N]} R(k,x)D_x,
\]
where the fixed diagonal matrices $\{D_x\}_x$ satisfy $\sum_x D_xD_x^\dagger = \sum_x D_x^\dagger D_x = \mathbb{E}_R D_{R,k}^\dagger D_{R,k}\preceq \Id$.
Define $T_{k,x}$ on the compressed oracle register that flips the presence of the database entry $(k,x)$; that is, $T_{k,x} \ket{S} = \ket{S\oplus \{(k,x)\}}$, and thus define
\begin{align}\label{eqdef:trace-synthesis-D-compress}
    \widehat D_k := \sum_{x\in[N]} D_x\otimes T_{k,x}.
\end{align}
This means that $\widehat D_k\ket{\phi}\ket S
=\sum_x D_x\ket{\phi}\otimes \ket{S\oplus\{(k,x)\}}$.

Therefore, the compressed oracle purification of $W_R$ is then expressed as
\[
    \widehat W
    :=\frac1K\sum_{k=1}^K
    \widehat D_k^\dagger(\Pi_k\otimes \Id)\widehat D_k,
\]
and our goal from \cref{ineq:trace-moment-search} can be written as
\begin{align}
    \mathbb{E}_R\Tr\big(W_R^{2t}\big)
    &= \sum_{z\in[M]} \bra{z}\big(\mathbb{E}_R W_R^{2t}\big)\ket{z} 
    = \sum_{z\in[M]} \bra{z}\bra{\emptyset}\widehat W^{2t}\ket{z}\ket{\emptyset}. \label{eq:trace-to-compressed-search}
\end{align}
It remains to bound the last quantity for each fixed $z$.

\subsubsection{Bounding the trace moment on small databases}
\label{subsec:trace-search-small-database}

\begin{lemma}\label{lem:1-query-compressed-oracle-search}
For $\widehat{W}$ defined as above, if normalized state $\ket{\psi}$ (with both the algorithm's register and the database register) is supported on databases of size $\le\ell$, then
\[
    \bra{\psi}\widehat W\ket{\psi}
    \leq \frac{4(\ell+1)}{K}.
\]
\end{lemma}

With \Cref{lem:1-query-compressed-oracle-search}, now we can bound \cref{eq:trace-to-compressed-search}.

\begin{proof}[Proof of \Cref{thm:binary-phase-search}]
For each $z$, each $\widehat W$ will append the database by size at most $2$. Therefore, $\ket{\psi_z^i}:=\widehat W^i\ket{z}\ket{\emptyset}$ is supported on database of size $\le 2i$, and for $P_{\le i}$ defined as the projector on database of size $\le i$,
\begin{align*}
    \norm{\ket{\psi_z^i}}
    &=\norm{P_{\leq 2i}\widehat W P_{\leq 2i-2}\ket{\psi_z^{i-1}}} \\
    &\leq \norm{P_{\leq 2i}\widehat W P_{\leq 2i}}\cdot \norm{\ket{\psi_z^{i-1}}} \\
    &\leq \frac{4(2i+1)}{K}\,\norm{\ket{\psi_z^{i-1}}}.
\end{align*}
Therefore, for each $z\in[M]$,
\[
\bra{z}\bra{\emptyset}\widehat W^{2t}\ket{z}\ket{\emptyset}
    =\norm{\ket{\psi_t}}^2
    \leq \prod_{i=1}^t\left(\frac{4(2i+1)}{K}\right)^2
    \leq \left(\frac{8t+4}{K}\right)^{2t}.
\]
From \cref{ineq:trace-moment-search,eq:trace-to-compressed-search}, this gives
\begin{align*}
    \mathbb{E}_R\norm{W_R}
    \leq
    \left(M\cdot \left(\frac{8t+4}{K}\right)^{2t}\right)^{1/2t} 
    = M^{1/2t}\cdot \frac{8t+4}{K}.
\end{align*}
Taking $t=\log M$ gives
\[
    \mathbb{E}_R\norm{W_R}\leq O\left(\frac{\log M}{K}\right).
\]
\end{proof}

Now we are left to prove \Cref{lem:1-query-compressed-oracle-search}.

\subsubsection{Proof of the small database search bound}

\begin{proof}[Proof of \Cref{lem:1-query-compressed-oracle-search}]
By definition,
\begin{align*}
    \bra{\psi}\widehat W\ket{\psi}
    =\frac1K\sum_k \norm{(\Pi_k\otimes \Id)\widehat D_k\ket{\psi}}^2.
\end{align*}
Write
\[
    \ket{\psi}=\sum_S \alpha_S\ket{\varphi_S}\ket S,
\]
where each $\ket{\varphi_S}$ is normalized, and where $\alpha_S = 0$ whenever $|S|> \ell$. Expanding $\widehat D_k$ and by the orthogonality of the final database $S^+$, this gives
\begin{align*}
    \bra{\psi}\widehat W\ket{\psi}
    &=\frac1K\sum_{k,S^+}
      \norm{\Pi_k\sum_{x,S:\, S^+=S\oplus\{(k,x)\}}
      \alpha_SD_x\ket{\varphi_S}}^2 \\
    &\leq \frac2K\sum_{k,S^+}\Bigg(
      \norm{\Pi_k\sum_{\substack{x:(k,x)\in S^+\\ S=S^+\setminus\{(k,x)\}}}
      \alpha_SD_x\ket{\varphi_S}}^2 \nonumber +
      \norm{\Pi_k\sum_{\substack{x:(k,x)\notin S^+\\ S=S^+\cup\{(k,x)\}}}
      \alpha_SD_x\ket{\varphi_S}}^2\Bigg).
\end{align*}
We bound the two sums separately.

\paragraph{For the first sum.} The inner sum has at most $|S^+|\leq \ell+1$ terms, so by Cauchy-Schwarz,
\begin{align*}
    \sum_{k,S^+}
      \norm{\Pi_k\sum_{\substack{x:(k,x)\in S^+\\ S=S^+\setminus\{(k,x)\}}}
      \alpha_SD_x\ket{\varphi_S}}^2
    &\leq (\ell+1)\sum_{k,S^+}
      \sum_{\substack{x:(k,x)\in S^+\\ S=S^+\setminus\{(k,x)\}}}
      |\alpha_S|^2\norm{\Pi_kD_x\ket{\varphi_S}}^2 \\
    &\leq (\ell+1) \sum_{k,x,S}|\alpha_S|^2\norm{\Pi_kD_x\ket{\varphi_S}}^2 \\
    &=(\ell+1)\sum_{x,S}|\alpha_S|^2
      \bra{\varphi_S}D_x^\dagger\Big(\sum_k\Pi_k\Big)D_x\ket{\varphi_S} \\
    &\leq (\ell+1)\sum_S |\alpha_S|^2
      \bra{\varphi_S}\Big(\sum_xD_x^\dagger D_x\Big)\ket{\varphi_S}
    \\
    &\leq \ell+1.
\end{align*}

\paragraph{For the second sum.} We use the operator-valued Cauchy-Schwarz inequality
\[
    \norm{\sum_j B_jv_j}^2
    \leq \norm{\sum_j B_jB_j^\dagger}\sum_j\norm{v_j}^2.
\]
Applying this with $B_x=\Pi_kD_x$ and using $\sum_x \Pi_kD_xD_x^\dagger\Pi_k\preceq \Pi_k\preceq I$,
we will obtain
\begin{align*}
    \sum_{k,S^+}
      \norm{\Pi_k\sum_{\substack{x:(k,x)\notin S^+\\ S=S^+\cup\{(k,x)\}}}
      \alpha_SD_x\ket{\varphi_S}}^2
    \leq
      \sum_{k,S^+}
      \sum_{\substack{x:(k,x)\notin S^+\\ S=S^+\cup\{(k,x)\}}}|\alpha_S|^2 \|\ket{\varphi_S}\|^2
    =\sum_S |S|\,|\alpha_S|^2
      \leq \ell. 
\end{align*}
\paragraph{Combining the two sums.} From the above upper bounds, 
\[
    \bra{\psi}\widehat W\ket{\psi}
    \leq \frac{2}{K}\big((\ell+1)+\ell\big)
    \leq \frac{4(\ell+1)}{K}.
\]
\end{proof}

\subsection{Oracle state distinguishing game hardness}
\label{sec:unitary-synthesis-distinguish}

In this subsection, we upper bound a distinguishing game winning probability of a one-query (unitary synthesis) adversary, which gives an alternate proof for Theorem 2.2 in \cite{STOC:LomMaWri24} (and thus an alternate proof of a one-query unitary synthesis lower bound). Instead of using matrix concentration inequalities for matrix Rademacher series, we upper bound the winning probability via the compressed oracle techniques.

\begin{definition}[Oracle state distinguishing game, from {\cite[Definition~3.8]{STOC:LomMaWri24}}]
\label{def:state-distinguishing-game}
Fix a random variable $\mathbf R$ and, for every $R$ in its support, define a family of normalized states $\bigl\{\ket{\psi_{R,k}}:k\in[K]\bigr\}\subseteq \mathbb{C}^{N}$.
Without loss of generality, we may describe $\ket{\psi_{R,k}}$ in the computational basis with the following normalization:
\[
\ket{\psi_{R,k}}
=
\frac{1}{\sqrt N}\sum_{x\in[N]} R(k,x)\ket{x}
\]
for some random variables
$\{\mathbf R(k,x)\in\mathbb C\}_{k\in[K],\,x\in[N]}$.

In the oracle state distinguishing game, the challenger samples $R$ together with an independent uniformly random bit $b\in\{0,1\}$.
If $b=0$, the challenger samples a uniformly random key $k\in[K]$ and gives the adversary one copy of $\ket{\psi_{R,k}}$.
If $b=1$, the challenger samples a uniformly random $x\in[N]$ and gives the adversary $\ket{x}$.
The adversary must output a bit $b'\in\{0,1\}$ after making one oracle query, and wins if $b'=b$.
The oracle can depend on $R$ but not on the challenge bit $b$, the sampled key $k$, or the sampled basis index $x$.
\end{definition}

Together with the one-query normal form of a unitary synthesis algorithm as in \Cref{def:one-query-algorithm}, now we can denote the distinguishing advantage conditioned on sampling $R$ as
\begin{align*}
    \Delta_\cA(R):= \max_f
    \left| \mathbb{E}_k\bra{\psi_{R,k}}V^\dagger O_f^\dagger \Pi O_f V\ket{\psi_{R,k}}
    -
    \mathbb{E}_h\bra{\psi_{h}}V^\dagger O_f^\dagger \Pi O_f V\ket{\psi_{h}}
    \right|
    ,
\end{align*}
where the final projective measurement  $\{\Pi, I-\Pi\}$ for distinguishing is also independent of $R$, and the post query unitary $U$ can be absorbed in the projective measurement by defining $\Pi':=U^\dagger \Pi U$. $O_f$ is the arbitrary phase oracle, where the underlying boolean function is $f$. $h$ is an independent random binary phase vector, and defines $\ket{\psi_h}$ as a random binary phase state.

\begin{theorem}[One-query distinguishing bound for binary phase states, {\cite[Theorem~4.17]{STOC:LomMaWri24}}]
\label{thm:binary-phase-distinguish}
For the oracle state distinguishing game associated with the family $\{\ket{\psi_{R,k}}\}_{k\in[K]}$ defined as in \Cref{def:state-distinguishing-game} with $R=(R(k,x))_{k\in[K],\,x\in[N]}$ as a family of independent Rademacher random variables, every one-query adversary with workspace dimension $M$ satisfies
\[
\mathbb{E}_R\bigl[\Delta_\cA(R)\bigr]
\le O\left(\sqrt{\frac{\log M}{K}}\right)
\]
for $\Delta_\cA(R)$ as the distinguishing advantage when $R$ is sampled.
\end{theorem}

\paragraph{Proof overview.}
By the weight-vector decomposition together with the spectral relaxation, the distinguishing advantage can be upper bounded by the expected norm of an operator on the algorithm's register. 
\begin{itemize}
    \item For distinguishing, the operator is centered. To keep this cancellation, we introduce an independent copy of the oracle. By symmetrization and convexity of the operator norm, it then suffices to bound the expected norm of the difference between the two associated acceptance operators.
\end{itemize}
Via the trace-moment method and rewriting in the compressed oracle database mode, the key remaining step is the small-database analysis (\Cref{lem:1-query-compressed-oracle-distinguish}).

In the compressed-oracle representation, we compare these operators by swapping the two copies of the challenge key row. 
\begin{itemize}
    \item This swap leaves any database component unchanged when both rows are empty, so its effect can be charged to components in which at least one of the rows is occupied.
    Therefore, averaging over the hidden key, one would expect to bound some squared difference by $O(\ell/K)$ (and later Cauchy-Schwarz would give a small-database distinguishing bound of $O(\sqrt{\ell/K})$).
\end{itemize}
Applying the trace-moment argument at order $1+\log M$ proves the theorem.

\subsubsection{Spectral relaxation and the compressed oracle trace moment}
\label{sec:trace-moment-distinguish}

Our goal is to upper bound the distinguishing advantage of a one-query adversary, which is to upper bound
\begin{align*}
    \mathbb{E}_R[\Delta_\cA(R)]
    =
    \mathbb{E}_R\max_f
    \left| \mathbb{E}_k\bra{\psi_{R,k}}V^\dagger O_f^\dagger \Pi O_f V\ket{\psi_{R,k}}
    -
    \mathbb{E}_h\bra{\psi_{h}}V^\dagger O_f^\dagger \Pi O_f V\ket{\psi_{h}}
    \right|.
\end{align*}
By \Cref{lem:weight-vector-decomposition}, we can apply a generic spectral relaxation similar to the search case in \Cref{sec:trace-moment-search}:
\begin{align*}
    &\mathbb{E}_R\max_f
    \left| \mathbb{E}_k\bra{\psi_{R,k}}V^\dagger O_f^\dagger \Pi O_f V\ket{\psi_{R,k}}
    -
    \mathbb{E}_h\bra{\psi_{h}}V^\dagger O_f^\dagger \Pi O_f V\ket{\psi_{h}}
    \right|
    \\
    &\quad =
    \mathbb{E}_R\max_f 
    \left|
    \bra{\wt_V}O_f^\dagger\cdot 
    \left(
        \mathbb{E}_k D_{R,k}^\dagger \Pi D_{R,k}-\mathbb{E}_h D_h^\dagger \Pi D_h
    \right)
    \cdot O_f\ket{\wt_V}
    \right|
    \\
    &\quad 
    \le \mathbb{E}_R
    \norm{\mathbb{E}_k D_{R,k}^\dagger \Pi D_{R,k} - \Gamma}, \qquad \text{for }\Gamma \coloneqq \mathbb{E}_{h} D_h^\dagger \Pi D_h,
\end{align*}
Define
\[ X_R \coloneqq \mathbb{E}_k D_{R,k}^\dagger \Pi D_{R,k} - \Gamma. \]
By symmetrization with an independent copy $R'$, with $W_R := \mathbb{E}_k D_{R,k}^\dagger \Pi D_{R,k}$ and by convexity of the operator norm, we have 
\begin{align*}
    \mathbb{E}_R \norm{X_R} 
    = \mathbb{E}_R \norm{W_R - \mathbb{E}_{R'} W_{R'}} 
    \leq \mathbb{E}_{R,R'} \norm{W_R - W_{R'}}.
\end{align*}
Therefore, define
\[ Y_{R,R'} \coloneqq W_R - W_{R'} = \mathbb{E}_k\big(D_{R,k}^\dagger \Pi D_{R,k} - D_{R',k}^\dagger \Pi D_{R',k}\big), \]
now it suffices to bound $\mathbb{E}_{R,R'} \norm{Y_{R,R'}}$.

Now $Y_{R,R'}$ is Hermitian.
As before, we use an even trace power. For every positive integer $t$,
\begin{align}\label{ineq:trace-moment-distinguishing}
    \mathbb{E}_{R,R'} \norm{Y_{R,R'}}
    \leq \Big(\mathbb{E}_{R,R'} \norm{Y_{R,R'}}^{2t}\Big)^{1/2t}
    \leq \Big(\mathbb{E}_{R,R'} \Tr(Y_{R,R'}^{2t})\Big)^{1/2t}.
\end{align}

We will now upper bound this trace moment using the compressed oracle methodology.

\paragraph{Rewrite in the compressed oracle framework.} Similarly as in \cref{eqdef:trace-synthesis-D-compress}, 
for key $k$, define
\[
    \widehat D_k := \sum_{x\in[N]} D_x\otimes T_{k,x},\qquad \widehat D'_k = \sum_x D_x \otimes T'_{k,x}.
\]
This allows us to define the compressed oracle purification
\begin{align*}
    \widehat Y
    := \frac{1}{K} \sum_{k=1}^K \big(\widehat D_k^\dagger \Pi \widehat D_k - \big(\widehat D'_{k}\big)^\dagger \Pi \widehat D'_k\big).
\end{align*}
Then, similarly \cref{ineq:trace-moment-distinguishing} can be written as 
\begin{align}
    \mathbb{E}_{R,R'} \Tr\big(Y_{R,R'}^{2t}\big)
    = \sum_{z\in[M]} \bra{z}\bra*{\emptyset} \widehat Y^{2t} \ket{z}\ket{\emptyset},
    \label{eq:trace-to-compressed-distinguish}
\end{align}
where now $\ket{\emptyset}$ denotes the empty compressed database for the two independent copies $R$ and $R'$. It remains to bound the last quantity for each fixed $z$. 

\subsubsection{Bounding the trace moment on small databases}
\label{subsec:trace-distinguish-small-database}

\begin{lemma}\label{lem:1-query-compressed-oracle-distinguish}
For $\widehat Y$ defined as above, if normalized state is supported on databases of size $\leq \ell$, where $\ell+1\leq K$, then
\begin{align*}
    \abs{\bra{\psi} \widehat Y \ket{\psi}} \leq 20 \sqrt{\frac{\ell+1}{K}}.
\end{align*}  
\end{lemma}

With \Cref{lem:1-query-compressed-oracle-distinguish}, now we can bound \cref{eq:trace-to-compressed-distinguish}.

\begin{proof}[Proof of \Cref{thm:binary-phase-distinguish}]
This proceeds exactly as in \Cref{subsec:trace-search-small-database}. 
For each $z$, each $\widehat Y$ will append the database by size at most $2$. Therefore, $\ket{\psi_z^i}:=\widehat Y^i\ket{z}\ket{\emptyset}$ is supported on database of size $\le 2i$, and for $P_{\le i}$ defined as the projector on database of size $\le i$, as long as $2i+1\le K$, 
\begin{align*}
    \norm{\ket{\psi_z^i}}
    &=\norm{P_{\leq 2i}\widehat Y P_{\leq 2i-2}\ket{\psi_z^{i-1}}} \\
    &\leq \norm{P_{\leq 2i}\widehat Y P_{\leq 2i}}\cdot \norm{\ket{\psi_z^{i-1}}} \\
    &\leq 20\sqrt{\frac{2i+1}{K}}\,\norm{\ket{\psi_z^{i-1}}}.
\end{align*}

Suppose that $K > 2\log M$ (or we will trivially bound the distinguishing advantage in the regime $K \le 2 \log M$ by $1$). Then we set $t = \log M$, and 
\begin{align*}
    \bra{z}\bra*{\emptyset} \widehat Y^{2t} \ket{z}\ket{\emptyset}
    = \norm{\ket{\psi_t}}^2
    \leq \prod_{i=1}^t \frac{400(2i+1)}{K}
    \leq \left(\frac{1600t}{K}\right)^t.
\end{align*}
Summing over $z\in[M]$, this gives
\begin{align*}
    \mathbb{E}_R \norm{X_R}
    \leq \mathbb{E}_{R,R'} \norm{Y_{R,R'}} 
    \leq \left( M \cdot \left(\frac{1600t}{K}\right)^t \right)^{1/2t} 
    &\leq O\!\left(\sqrt{\frac{\log M}{K}}\right).
\end{align*}

\end{proof}

\subsubsection{Proof of the small database distinguishing bound}

\begin{proof}[Proof of \Cref{lem:1-query-compressed-oracle-distinguish}]
For each $k$, define $p_k \coloneqq \norm{\Pi \widehat D_k \ket{\psi}}^2$ and $q_k \coloneqq \norm{\Pi \widehat D'_k \ket{\psi}}^2$, so that
\begin{align*}
    \bra{\psi} \widehat Y \ket{\psi} = \frac{1}{K} \sum_k (p_k - q_k).
\end{align*}

For a fixed choice of $k$ and for $b\in \{0,1\}$, define
\begin{align*}
    n_k^{(b)}(S) \coloneqq \#\{x : (b,k,x) \in S\}.
\end{align*}
We first prove the following bound.
\begin{claim}\label{claim:trace-distinguish-row-by-row}
    Suppose $\ket{\psi} = \sum_S \alpha_S \ket{\varphi_S}\ket{S}$. Then for every $k \in [K]$, $b\in \{0,1\}$, we have
    \begin{align*}
        \norm{\Pi \widehat D_k^{(b)} \ket{\psi}}^2 \leq 2 \sum_S (2n_k^{(b)}(S)+1) \cdot |\alpha_S|^2,
    \end{align*}
    where $\widehat D_k^{(0)} = \widehat D_k$ and $\widehat D_k^{(1)} = \widehat D_k'$.
\end{claim}

\begin{proof}[Proof of \Cref{claim:trace-distinguish-row-by-row}]
The $b=0$ and $b=1$ cases are symmetric, so we only need to analyze the case for $b=0$. Set $n_k(S) = n_k^{(0)}(S)$. For $S^+$ defined as the final database, we expand
\begin{align*}
    \norm{\Pi \widehat D_k \ket{\phi}}^2
    = \sum_{S^+}
    \bigg\| \Pi \cdot \sum_{x,S:S^+=S\oplus\{(k,x)\}} \alpha_S D_x \ket{\phi_S} \bigg\|^2 
    \leq 2(A_k+B_k),
\end{align*}
where
\begin{align*}
    A_k \coloneqq \sum_{S^+}
    \bigg\| \Pi \cdot \sum_{\substack{x:(k,x)\in S^+\\ S=S^+\setminus\{(k,x)\}}} \alpha_S D_x \ket{\phi_S} \bigg\|^2, \quad
    B_k \coloneqq \sum_{S^+}
    \bigg\| \Pi \cdot \sum_{\substack{x:(k,x)\notin S^+\\ S=S^+\cup\{(k,x)\}}} \alpha_S D_x \ket{\phi_S} \bigg\|^2.
\end{align*}

\paragraph{Bound $A_k$.} Note that for fixed $S^+$, the inner sum contains exactly $n_k(S^+)$ terms, and thus,
\begin{align*}
    A_k
    &\leq \sum_{S^+} n_k(S^+) \cdot \sum_{\substack{x:(k,x)\in S^+\\ S=S^+\setminus\{(k,x)\}}} |\alpha_S|^2 \cdot \norm{\Pi D_x \ket{\phi_S}}^2 \\
    &\leq \sum_S (n_k(S)+1) \cdot \sum_x |\alpha_S|^2 \cdot \norm{\Pi D_x \ket{\phi_S}}^2 \\
    &= \sum_S (n_k(S)+1) \cdot |\alpha_S|^2 \cdot \bra{\phi_S} \sum_x D_x^\dagger \Pi D_x \ket{\phi_S} \\
    &\leq \sum_S (n_k(S)+1) \cdot |\alpha_S|^2,
\end{align*}
where in the last line we used $\Pi \preceq I$ together with $\sum_x D_x^\dagger D_x \preceq I$.

\paragraph{Bound $B_k$.} We use the same operator-valued Cauchy-Schwarz as in the proof of \Cref{lem:1-query-compressed-oracle-search}:
\begin{align*}
    B_k
    &\leq \sum_{S^+}
    \left\| \sum_{\substack{x:(k,x)\notin S^+}} \Pi D_x D_x^\dagger \Pi \right\|
    \cdot
    \sum_{\substack{x:(k,x)\notin S^+\\ S=S^+\cup\{(k,x)\}}} |\alpha_S|^2 \\
    &\leq \sum_{S^+}
    \sum_{\substack{x:(k,x)\notin S^+\\ S=S^+\cup\{(k,x)\}}} |\alpha_S|^2 \\
    &= \sum_S n_k(S) \cdot |\alpha_S|^2,
\end{align*}
where we used that the $D_x$ are diagonal, and therefore $\sum_x D_x D_x^\dagger = \sum_x D_x^\dagger D_x \preceq \Id$.

\paragraph{Combining the two estimates. } This will give
\begin{align*}
    \norm{\Pi \widehat D_k \ket{\phi}}^2 \leq 2 \sum_S (2n_k(S)+1) \cdot |\alpha_S|^2,
\end{align*}
which proves \Cref{claim:trace-distinguish-row-by-row}.
\end{proof}

We now return to the proof of the lemma. 

Applying Cauchy-Schwarz exactly as in the first section,
\begin{align*}
    \abs{\bra{\psi} \widehat Y \ket{\psi}}
    &= \frac{1}{K} \abs{\sum_k (p_k - q_k)} 
    \leq \left(\frac{1}{K}\sum_k (\sqrt{p_k}-\sqrt{q_k})^2\right)^{1/2}
    \cdot
    \left(\frac{1}{K}\sum_k (\sqrt{p_k}+\sqrt{q_k})^2\right)^{1/2}.
\end{align*}
We bound the two factors separately.

\paragraph{The second factor.} Applying \Cref{claim:trace-distinguish-row-by-row} with $b=0$ and $b=1$ respectively, we obtain
\begin{align*}
    \frac{1}{K} \sum_k p_k
    &\leq \frac{2}{K} \sum_k \sum_S (2n_k(S)+1) \cdot |\alpha_S|^2 
    = 2 \sum_S \left(1 + \frac{2}{K}\sum_k n_k(S)\right) |\alpha_S|^2 
    \leq 2 + \frac{4\ell}{K},
\end{align*}
and this holds similarly for $q_k$, with $\frac1K\sum_k q_k\le 2+\frac{4\ell}{K}$. Therefore, for $\ell+1\le K$,
\begin{align*}
    \frac{1}{K}\sum_k (\sqrt{p_k}+\sqrt{q_k})^2
    &\leq \frac{2}{K} \sum_k (p_k+q_k)
    \leq 
    8+\frac{16\ell}{K}\le 24.
\end{align*}

\paragraph{The first factor.}
For each $k$, let $S_k$ be the unitary that swaps $(0,k)$ and $(1,k)$ in the compressed database; that is, swapping the database $D_k$ and $D_k'$. Let $B_k$ denote the projector onto database states $S$ satisfying $ n_k^{(0)}(S)+n_{k}^{(1)}(S) > 0$; that is, projecting on database with $D_k\cup D_k'$ non-empty.

Write
$\ket{a_k} \coloneqq (I-B_k)\ket{\psi}$ and $\ket{b_k} \coloneqq B_k \ket{\psi}$.
Then $S_k \ket{a_k} = \ket{a_k}$, and $S_k \widehat D'_k = \widehat D_k S_k$. This implies
\begin{align*}
    \sqrt{p_k}
    = \norm{\Pi \widehat D_k (\ket{a_k}+\ket{b_k})}
    = \norm{\Pi \widehat D'_k (\ket{a_k} + S_k \ket{b_k})},
\end{align*}
while
\begin{align*}
    \sqrt{q_k} = \norm{\Pi \widehat D'_k (\ket{a_k}+\ket{b_k})}.
\end{align*}
Therefore,
\begin{align*}
    \abs{\sqrt{p_k}-\sqrt{q_k}}^2
    &\leq \norm{\Pi \widehat D'_k (S_k \ket{b_k} - \ket{b_k})}^2
    \leq 2\norm{\Pi \widehat D'_k S_k \ket{b_k}}^2 + 2\norm{\Pi \widehat D'_k \ket{b_k}}^2 \\
    &= 2\norm{\Pi \widehat D_k \ket{b_k}}^2 + 2\norm{\Pi \widehat D'_k \ket{b_k}}^2.
\end{align*}
Applying \Cref{claim:trace-distinguish-row-by-row} to state $\ket{b_k}$, we have
\begin{align*}
    \norm{\Pi \widehat D_k \ket{b_k}}^2
    &\leq 2 \sum_S (2n_k^{(0)}(S)+1) \cdot \mathbf{1}[n_k^{(0)}(S)+n_k^{(1)}(S)>0] \cdot |\alpha_S|^2,
\end{align*}
and similarly for $\widehat D_k'$. Therefore
\begin{align*}
    \frac{1}{K}\sum_k (\sqrt{p_k}-\sqrt{q_k})^2
    &\leq \frac{4}{K}\sum_S |\alpha_S|^2 \cdot \sum_{k:n_k^{(0)}(S)+n_k^{(1)}(S)>0} \big(2n_k^{(0)}(S)+1+2n_k^{(1)}(S)+1\big)
    \\
    &\leq \frac{4}{K}\sum_S |\alpha_S|^2 \cdot 4|S|
    \le \frac{16\ell}{K}.
\end{align*}

\paragraph{Combining the two factors.} This will give
\begin{align*}
    \abs{\bra{\psi} \widehat Y \ket{\psi}}
    \leq \sqrt{\frac{16\ell}{K}} \cdot \sqrt{24} 
    \leq 20 \sqrt{\frac{\ell+1}{K}},
\end{align*}
which proves \Cref{lem:1-query-compressed-oracle-distinguish}.
\end{proof}

\section{One-query unitary synthesis with additional query access to $R$}
\label{sec:1-query-post-query}

In this section, we extend the oracle state search and distinguishing bounds of \cref{sec:unitary-synthesis} to adversaries with additional access to the underlying random oracle $R$. After receiving the challenge, the adversary may make $t$ queries to $R$, one query to an arbitrary phase oracle $O_{f_R}$, and then $T$ further queries to $R$. As before, $f_R$ may depend arbitrarily on $R$ but not on the challenge, while all non-oracle operations are fixed independently of $R$.

We first consider the post-query setting, where $t=0$. Here the isometry preceding the arbitrary phase query remains independent of $R$, so the spectral relaxation from \cref{sec:unitary-synthesis} still applies.
We extend the compressed-oracle analysis to handle the subsequent queries, proving the search bound in \cref{sec:post-query-search} and the distinguishing bound in \cref{sec:post-query-distinguishing}. 
Pre-queries make the preceding isometry depend on $R$, so this analysis no longer applies directly.
In \cref{sec:pre-query}, we use a resampling argument to reduce that setting to the post-query case. The resulting bounds are quantitatively weaker, but still give negligible success probability and distinguishing advantage when $K$ is exponential and the numbers of queries and workspace qubits are polynomial.

With some abuse of notation, in the post-query proofs of \cref{sec:post-query-search,sec:post-query-distinguishing}, the number of pre-queries is 0, and $t$ denotes a freely chosen trace-moment parameter, with moments of order $2t$; in \cref{sec:pre-query}, we reuse $t$ for the number of queries to $R$ before the arbitrary phase query.

\subsection{Oracle state search game with post-query  access}
\label{sec:post-query-search}

We begin with the search game in the post-query setting, where there is no pre-query. The following theorem extends the bound from \cref{sec:unitary-synthesis-search} to adversaries making $T$ queries to $R$ after the arbitrary phase query. 

\begin{theorem}[One-query search bound for binary phase states, with post queries]
\label{thm:binary-phase-search-post}
For the variant of the oracle state search game associated with the family $\{\ket{\psi_{R,k}}\}_{k\in[K]}$ with $R=(R(k,x))_{k\in[K],\,x\in[N]}$ as a family of independent Rademacher random variables, where the adversary is also allowed to make $T$ queries to $R$ after the arbitrary phase query,
every such adversary with workspace dimension $M$ satisfies
\[
\mathbb{E}_R\bigl[\Win(\cA\mid R)\bigr]
\le O\left(\frac{T^2+(T+1)\log M}{K}\right).
\]
\end{theorem}

\paragraph{Proof overview.}
The spectral relaxation from the one-query search game \Cref{sec:unitary-synthesis-search} still applies, since the isometry preceding the arbitrary phase query is independent of $R$. It therefore remains to prove the corresponding small-database estimate in the presence of post-queries (\Cref{lem:post-query-sparse-search}). 

The main additional difficulty is that the subsequent $T$ queries can use the correlation between the challenge state and the random oracle. We handle this correlation by introducing a fresh auxiliary row $0\notin[K]$, which is never queried by the algorithm, and swapping it with the challenge row $k$. This helps us compare in these ways:
\begin{itemize}
    \item The swap changes the initial compressed database only on components where database on key $k$ is nonempty. For a database of size at most $\ell$, this contributes $O(\sqrt{\ell/K})$ to the success amplitude.
    \item Conjugating the post-query computation by the swap, this changes the oracle only on queries with address on row $k$. Replacing these queries one at a time contributes $O(T/\sqrt K)$.
\end{itemize}
Also, in the ideal challenge case where the challenge is generated from the fresh row, the computation is independent of the hidden key, so its success probability is at most $1/K$. Combining the three amplitude contributions gives the small-database bound $O((\ell+T^2+1)/K)$.

This extends to the distinguishing case in \Cref{sec:post-query-distinguishing}, as we can follow the same hybrids in the analysis for the centered distinguishing operator.

\subsubsection{Spectral relaxation with post-query access}
We define
\[
    A_R := A_R^{(T)} =(O_R\cdot U)^T
\]
as the $T$-query post-processing unitary, where $U$ is a fixed unitary independent of $R$. The weight vector decomposition in \Cref{sec:weight-vector-decomposition}  still holds, since all $T$ queries are applied after the arbitrary phase oracle query $O_f$.
Therefore, the corresponding spectral relaxation, modified from \Cref{sec:trace-moment-search} by adding $A_R$, is the spectral norm of the matrix
\[
    W_R^{(T)}
    :=\frac1K\sum_{k=1}^K D_{R,k}^\dagger A_R^\dagger\Pi_k A_R D_{R,k}.
\]

Since $W_R^{(T)}\succeq0$, for every positive integer $t$,
\begin{align}\label{ineq:post-query-trace-moment}
    \mathbb{E}_R\norm{W_R^{(T)}}
    &\leq \Big(\mathbb{E}_R\Tr\big((W_R^{(T)})^{2t}\big)\Big)^{1/2t}. 
\end{align}

For random $R$, we can write all things in the compressed oracle framework. For compressed phase oracle $\cO$, we can rewrite $A_R^{(T)}$ and $W_R^{(T)}$ by
\[
    \widehat A=(\cO\,U)^T, \qquad \widehat W
    :=\frac1K\sum_{k=1}^K
    \widehat D_k^\dagger \widehat A^\dagger (\Pi_k\otimes\Id)\widehat A\widehat D_k .
\]
And therefore, our goal from \cref{ineq:post-query-trace-moment} can be written as
\begin{align}
    \Tr\big(\mathbb{E}_R (W_R^{(T)})^{2t}\big)
    &=\sum_{z\in[M]}\bra z\bra\emptyset \widehat W^{2t}\ket z\ket\emptyset . \label{eq:post-query-compressed-trace}
\end{align}

\subsubsection{From the small-database estimate to the trace moment}
\begin{lemma}\label{lem:post-query-sparse-search}
Let $\ket u_{\mathsf K}=K^{-1/2}\sum_{k\in[K]}\ket k_\mathsf K$, and define
\[
    P_{\mathrm{win}}:=\sum_{k=1}^K\ketbra{k}_{\mathsf K}\otimes(\Pi_k\otimes\Id),
    \qquad
    D_{\mathrm{ctrl}}:=\sum_{k=1}^K\ketbra{k}_{\mathsf K}\otimes\widehat D_k.
\]
If $\ket\psi$ is normalized, supported on databases of size at most $\ell$, and has no support on row $0$, then
\[
    p:=\norm{P_{\mathrm{win}}(\Id_{\mathsf K}\otimes \widehat A)D_{\mathrm{ctrl}}(\ket u_{\mathsf K}\otimes\ket\psi)}^2
    \leq \frac{(1+2T+\sqrt{14\ell})^2}{K}.
\]
In particular,
\[
    p=O\left(\frac{\ell+T^2+1}{K}\right).
\]
\end{lemma}

From the definition of $\widehat W$, \Cref{lem:post-query-sparse-search} gives an upper bound for $\bra\psi\widehat W\ket\psi$. Similarly, this gives us a way to bound \cref{eq:post-query-compressed-trace}.

\begin{proof}[Proof of \Cref{thm:binary-phase-search-post}]
For each $z\in [M]$, each $\widehat W$ will append (or change) the database by size at most $2(T+1)$: 2 in $\widehat D_k$ and ${\widehat D_k}^\dagger$, and $2T$ in $\widehat A$ and ${\widehat A}^\dagger$. Therefore, $\ket{\psi_z^i}:=\widehat W^i\ket{z}\ket{\emptyset}$ is supported on database of size $\le 2i(T+1)$. 

For $P_{\leq \ell}^{\mathrm{no0}}$ defined as the projector onto databases of size at most $\ell$ that contain no entry in row $0$ (row $0$ will be used only as a fresh auxiliary row in the proof), we have
\begin{align*}
    \norm{\ket{\psi_z^i}}
    &=\norm{P^{\mathrm{no0}}_{\leq 2i(T+1)}\widehat W P^{\mathrm{no0}}_{\leq 2(i-1)(T+1)}\ket{\psi_z^{i-1}}} \\
    &\leq \norm{P^{\mathrm{no0}}_{\leq 2i(T+1)}\widehat W P^{\mathrm{no0}}_{\leq 2i(T+1)}}\cdot \norm{\ket{\psi_z^{i-1}}} \\
    &\leq \frac{(1+2T+\sqrt{28i(T+1)})^2}{K}\,\norm{\ket{\psi_z^{i-1}}}.
\end{align*}
Therefore, for each $z\in[M]$, for some universal constant $C$,
\[
    \bra z\bra\emptyset(\widehat W)^{2t}\ket z\ket\emptyset
    =\norm{\psi_z^t}^2
    \leq \left(\frac{C(T^2+(T+1)t)}{K}\right)^{2t}.
\]
From \cref{ineq:post-query-trace-moment,eq:post-query-compressed-trace}, and taking $t=\log M$, this gives
\[
    \mathbb{E}_R\norm{W_R^{(T)}}
    \leq O\left(\frac{T^2+(T+1)\log M}{K}\right).
\]
\end{proof}

Now we are left to prove \Cref{lem:post-query-sparse-search}.

\subsubsection{Proof of the search bound with post-query}
\label{sec:proof-of-post-query-sparse-search}

We first prove the following bound. As above, row $0$ is an auxiliary row
and is not queried by $\cO$.

\begin{claim}\label{claim:post-query-row-zero}
Let
\[
    \widehat D_0:=\sum_x D_x\otimes T_{0,x},
    \qquad
    n_0(S):=\#\{x:(0,x)\in S\}.
\]
For any state $\ket{\psi}=\sum_S\beta_S\ket{\varphi_S}\ket S$,
where each $\ket{\varphi_S}$ is normalized, we have
\begin{equation}\label{eq:post-query-row-zero-general}
    \norm{\widehat D_0\ket{\psi}}^2
    \le 2\sum_S(2n_0(S)+1)|\beta_S|^2.
\end{equation}
Moreover, if every database in the support of $\ket{\varphi}$ has row $0$
empty, then
\begin{equation}\label{eq:post-query-row-zero-empty}
    \norm{\widehat D_0\ket{\psi}}^2
    \le \norm{\ket{\psi}}^2.
\end{equation}
\end{claim}

\begin{proof}[Proof of \Cref{claim:post-query-row-zero}]
The proof of \cref{eq:post-query-row-zero-general} is the same as
the proof of \Cref{claim:trace-distinguish-row-by-row}, with $\Pi=\mathrm{Id}$ and row $0$ in place of
row $k$. For \cref{eq:post-query-row-zero-empty}, since row $0$ is empty,
the final databases $S\cup\{(0,x)\}$ are distinct for distinct pairs
$(S,x)$. Therefore,
\[
    \norm{\widehat D_0\ket{\psi}}^2
    =\sum_S|\beta_S|^2\sum_x\norm{D_x\ket{\varphi_S}}^2
    \le\sum_S|\beta_S|^2
    =\norm{\ket{\psi}}^2,
\]
where we used $\sum_xD_x^\dagger D_x\preceq\mathrm{Id}$.
\end{proof}

We now return to the proof of \Cref{lem:post-query-sparse-search}.

\begin{proof}[Proof of \Cref{lem:post-query-sparse-search}]
Set $\ket\Psi:=\ket u_\mathsf K\otimes\ket\psi$. For each $k\in[K]$, let
$S_k$ be the unitary that swaps rows $0$ and key $k$ part in the compressed
database. Define
\[
    S:=\sum_{k=1}^K\ket k\!\bra k_\mathsf K\otimes S_k,
    \qquad
    \widehat A':=S(\mathrm{Id}_\mathsf K\otimes\widehat A)S.
\]
By definition, $D_{\mathrm{ctrl}}S
    =S(\mathrm{Id}_\mathsf K\otimes\widehat D_0)$.
Since $P_{\mathrm{win}}$ acts trivially on the database, it commutes
with $S$. Therefore, from the definition of $p$ as in \Cref{lem:post-query-sparse-search}, it can be rewritten as 
\begin{align}\label{eq:post-query-search-prob-compressed}
    p
    =\norm{P_{\mathrm{win}}(\mathrm{Id}_\mathsf K\otimes\widehat A)
        D_{\mathrm{ctrl}}\ket\Psi}^2
    =\norm{P_{\mathrm{win}}\widehat A'
        (\mathrm{Id}_\mathsf K\otimes\widehat D_0)S\ket\Psi}^2.
\end{align}
Define modified winning probability $p_0$ (for acting always on row $0$) and $p_*$ (for starting unentangled with the challenge key):
\begin{align}
    p_0&:=\norm{P_{\mathrm{win}}\widehat A'
        (\mathrm{Id}_\mathsf K\otimes\widehat D_0)\ket\Psi}^2,
        \label{eq:post-query-search-prob-compress-0}
        \\
    p_*&:=\norm{P_{\mathrm{win}}(\mathrm{Id}_\mathsf K\otimes\widehat A)
        (\mathrm{Id}_\mathsf K\otimes\widehat D_0)\ket\Psi}^2.
        \label{eq:post-query-search-prob-compress-*}
\end{align}
By triangle inequality,
\[
    \sqrt p
    \le |\sqrt p-\sqrt{p_0}|
       +|\sqrt{p_0}-\sqrt{p_*}|+\sqrt{p_*}.
\]
We bound the three terms separately.

\paragraph{The first term $|\sqrt p-\sqrt{p_0}|$.}
By triangle inequality and unitarity of $\widehat A'$,
\[
    |\sqrt p-\sqrt{p_0}|
    \le\norm{(\mathrm{Id}_\mathsf K\otimes\widehat D_0)
        (S-\mathrm{Id})\ket\Psi}.
\]
Write $\ket\psi=\sum_S\alpha_S\ket{\varphi_S}\ket S$, $n_k(S):=\#\{x:(k,x)\in S\}$, where each $\ket{\varphi_S}$ is normalized, and $\alpha_S=0$ whenever $|S|>\ell$ or $n_0(S)>0$.
Let $B_k$ be the projector onto databases with row $k$ (key $k$ part) non-empty.
Since row $0$ is empty in the support of $\ket\psi$, we have
\[
    (S_k-\mathrm{Id})\ket\psi
    =(S_k-\mathrm{Id})B_k\ket\psi.
\]
The state $B_k\ket\psi$ has row $0$ empty, while $S_kB_k\ket\psi$
has $n_k(S)$ entries in row $0$ for each original database $S$.
Applying \Cref{claim:post-query-row-zero}, this gives
\[
\begin{aligned}
    \norm{\widehat D_0(S_k-\mathrm{Id})\ket\psi}^2
    &\le 2\norm{\widehat D_0S_kB_k\ket\psi}^2
        +2\norm{\widehat D_0B_k\ket\psi}^2\\
    &\le 4\sum_{S:n_k(S)>0}(2n_k(S)+1)|\alpha_S|^2
        +2\sum_{S:n_k(S)>0}|\alpha_S|^2\\
    &\le 14\sum_S n_k(S)|\alpha_S|^2.
\end{aligned}
\]
Therefore, by orthogonality of the register $K$,
\[
    \norm{(\mathrm{Id}_\mathsf K\otimes\widehat D_0)
        (S-\mathrm{Id})\ket\Psi}^2
    =\frac1K\sum_{k=1}^K
        \norm{\widehat D_0(S_k-\mathrm{Id})\ket\psi}^2
    \le\frac{14}{K}\sum_S|S|\,|\alpha_S|^2
    \le\frac{14\ell}{K}.
\]
This gives $|\sqrt p-\sqrt{p_0}|\le\sqrt{14\ell/K}$.

\paragraph{The second term $|\sqrt{p_0}-\sqrt{p_*}|$.}
We compare $p_0$ and $p_*$ by replacing the $T$ queries one at a time. 

Define $\cO':=S(\mathrm{Id}_\mathsf K\otimes\cO)S$.
Since $\mathrm{Id}_\mathsf K\otimes U$ commutes with $S$, the circuit
$\widehat A'$ is obtained from $\mathrm{Id}_\mathsf K\otimes\widehat A$ by
replacing each $\mathrm{Id}_\mathsf K\otimes\cO$ with $\cO'$.
For $j=0,1,\ldots,T$, let $\widehat A_j$ be the circuit in which the
first $j$ queries are $\mathrm{Id}_\mathsf K\otimes\cO$, and the remaining
$T-j$ queries are $\cO'$. Then
$\widehat A_0=\widehat A'$ and
$\widehat A_T=\mathrm{Id}_\mathsf K\otimes\widehat A$. Therefore,
\[
    |\sqrt{p_0}-\sqrt{p_*}|
    \le\sum_{j=0}^{T-1}
        \norm{(\widehat A_j-\widehat A_{j+1})
        (\mathrm{Id}_\mathsf K\otimes\widehat D_0)\ket\Psi}.
\]

For a fixed $j$, immediately before the differing query, all previous
queries are $\mathrm{Id}_\mathsf K\otimes\cO$. Thus the operations already
applied act trivially on register $\mathsf K$, and the state has the form
\[
    \ket u_\mathsf{K}\otimes\ket\Theta,
    \qquad \norm{\ket\Theta}\le1,
\]
where the norm bound follows from unitarity of the previous operations. 
Note that the two oracle calls $\cO,\cO'$ differ only when the queried row $r$ equals the challenge key $k$.
For
\[
    B:=\sum_{k=1}^K\ket k\!\bra k_\mathsf K\otimes\ket k\!\bra k_\mathsf Q,
\]
for the queried row register $\mathsf Q$, we will see
\[
    \norm{(\widehat A_j-\widehat A_{j+1})
        (\mathrm{Id}_\mathsf K\otimes\widehat D_0)\ket\Psi}
    \le2\norm{B(\ket u_\mathsf K\otimes\ket\Theta)}
    \le\frac{2}{\sqrt{K}}.
\]
Therefore, by summing over $j=0$ to $T-1$, we will obtain
\[
    |\sqrt{p_0}-\sqrt{p_*}|\le\frac{2T}{\sqrt K}.
\]

\paragraph{The third term $\sqrt{p_*}$.}
Since $\widehat D_0$ and $\widehat A$ act trivially on register $K$,
\[
    (\mathrm{Id}_\mathsf K\otimes\widehat A)
    (\mathrm{Id}_\mathsf K\otimes\widehat D_0)\ket\Psi
    =\ket u_\mathsf K\otimes\ket\xi
\]
for some $\ket\xi$ with $\norm{\ket\xi}\le1$. Therefore, as $\sum_k\Pi_k\preceq\mathrm{Id}$,
\[
    p_*
    =\norm{P_{\mathrm{win}}(\ket u_\mathsf K\otimes\ket\xi)}^2
    =\frac1K\sum_{k=1}^K
        \norm{(\Pi_k\otimes\mathrm{Id})\ket\xi}^2
    \le\frac1K.
\]

\paragraph{Combining the three terms.}
From the above upper bounds,
\[
    \sqrt p
    \le\sqrt{\frac{14\ell}{K}}+\frac{2T}{\sqrt K}
        +\frac1{\sqrt K}
    =\frac{1+2T+\sqrt{14\ell}}{\sqrt K}.
\]
Squaring both sides, this gives $p\le{(1+2T+\sqrt{14\ell})^2}/{K}$,
which proves \Cref{lem:post-query-sparse-search}.
\end{proof}

\subsection{Oracle state distinguishing game with post-query access}\label{sec:post-query-distinguishing}

From the proof of \Cref{thm:binary-phase-search-post} for the search case, in fact we can simply obtain a bound for the corresponding distinguishing case. This extends the one-query distinguishing lower bound in \cite{STOC:LomMaWri24} to hold even when the adversary is allowed to make polynomial queries to $R$ after the one query.

\begin{theorem}[One-query distinguishing bound for binary phase states, with post queries]
\label{thm:binary-phase-distinguish-post}
For variant of the oracle state distinguishing game associated with the family $\{\ket{\psi_{R,k}}\}_{k\in[K]}$ with $R=(R(k,x))_{k\in[K],\,x\in[N]}$ as a family of independent Rademacher random variables, where the adversary is also allowed to make $T$ queries to $R$ after the arbitrary phase query, every such adversary with workspace dimension $M$ satisfies
\[
\mathbb{E}_R\bigl[\Delta_\cA(R)\bigr]
\le O\left(\sqrt{\frac{T^2+(T+1)\log M}{K}}\right)
\]
for $\Delta_\cA(R)$ as the distinguishing advantage when $R$ is sampled.
\end{theorem}

\begin{proof}[Proof sketch.]

For the distinguishing game, the advantage is upper bounded by
$\mathbb{E}_R\|X_R\|$, where
\[
\begin{aligned}
X_R
&:= \frac{1}{K}\sum_{k=1}^K
D_{R,k}^{\dagger}A_R^{\dagger}\Pi A_RD_{R,k}-\Gamma_R,\\
\Gamma_R
&:= \mathbb{E}_h D_h^{\dagger}A_R^{\dagger}\Pi A_RD_h
 = \sum_x D_x^{\dagger}A_R^{\dagger}\Pi A_RD_x.
\end{aligned}
\]
Here $\Pi$ is the distinguishing projector and $h$ is an independent
random binary phase vector, and the spectral relaxation is modified from \Cref{sec:trace-moment-distinguish} by adding $A_R$. In particular, both experiments use the
same oracle $R$ for the $T$ post-queries. Let $\widehat X$ be the
compressed-oracle representation of $X_R$.

\paragraph{Bounding through small database (as \Cref{lem:1-query-compressed-oracle-distinguish} for $T=0$).}
Let $|\psi\rangle$ be normalized, supported on databases of size at
most $\ell$, and have row $0$ empty. Define
\[
\begin{aligned}
p'&:=\frac{1}{K}\sum_{k=1}^K
\bigl\|(\Pi\otimes\mathrm{Id})\widehat A
\widehat D_k|\psi\rangle\bigr\|^2,\\
p_*'&:=\bigl\|(\Pi\otimes\mathrm{Id})\widehat A
\widehat D_0|\psi\rangle\bigr\|^2.
\end{aligned}
\]
Since row $0$ is empty and is not queried by $\widehat A$, averaging
over that row gives
$\left|\langle\psi|\widehat X|\psi\rangle\right|=|p'-p_*'|$.
Moreover, by \cref{eq:post-query-row-zero-empty}, $p_*'\le1$.

By defining $P_{\mathrm{win}}=\mathrm{Id}_{K}\otimes
(\Pi\otimes\mathrm{Id})$, $p'$ and $p'_*$ will correspond to the $p$ and $p_*$ as in \cref{eq:post-query-search-prob-compressed,eq:post-query-search-prob-compress-*} in \Cref{sec:proof-of-post-query-sparse-search}. Moreover, since upper bounding $|\sqrt{p}-\sqrt{p_0}|$ and $|\sqrt{p_0}-\sqrt{p_*}|$ is independent of $P_\mathrm{win}$ in \Cref{sec:proof-of-post-query-sparse-search}, by triangle inequality, we can obtain
\begin{align*}
    \left|\sqrt{p'}-\sqrt{p'_*}\right|
    \le \norm{(\mathrm{Id}_\mathsf K\otimes\widehat D_0)
        (S-\mathrm{Id})\ket\Psi}+\sum_{j=0}^{T-1}
        \norm{(\widehat A_j-\widehat A_{j+1})
        (\mathrm{Id}_\mathsf K\otimes\widehat D_0)\ket\Psi}
    \le \frac{\sqrt{14\ell}+2T}{\sqrt K}.
\end{align*}
Therefore, we can upper bound $\langle\psi|\widehat X|\psi\rangle$ by
\[
\left|\langle\psi|\widehat X|\psi\rangle\right|
=|p'-p'_*|
\le \left|\sqrt{p'}-\sqrt{p'_*}\right|\cdot \left(\left|\sqrt{p'}-\sqrt{p'_*}\right|+2\right)
=O\left(\frac{\sqrt{\ell}+T}{\sqrt{K}}\right).
\]

\paragraph{Bounding the trace moment (as \Cref{subsec:trace-distinguish-small-database} for $T=0$).}
Each application of $\widehat X$ changes the database size by at most $2(T+1)$ and never changes row $0$. Thus, by a similar induction as in \Cref{subsec:trace-distinguish-small-database}, for every positive integer $t$,
\[
\begin{aligned}
\mathbb{E}_R\|X_R\|
\le
\left(
\sum_{z\in[M]}
\bigl\|\widehat X^{\,t}|z\rangle|\emptyset\rangle\bigr\|^2
\right)^{1/2t}
\le M^{1/2t}
\cdot
C\cdot \frac{\sqrt{2t(T+1)}+T}{\sqrt{K}}.
\end{aligned}
\]
Taking $t=\log M$ gives
\[
\mathbb{E}_R\|X_R\|
\le O\left(\sqrt{\frac{T^2+(T+1)\log M}{K}}\right).
\]
Therefore, 
the expected distinguishing advantage is at most $O\!\left(
\min\!\left\{
1,\sqrt{\frac{T^2+(T+1)\log M}{K}}
\right\}
\right)$.
  
\end{proof}

\subsection{Extension to the setting with both post-query and pre-query access}\label{sec:pre-query}

We now allow the adversary to make $t$ ordinary queries to $R$ after receiving the challenge but before making the one arbitrary phase query, followed by $T$ queries to $R$.

With the post $T$ query part defined as $A_R$,
now the computation (before measurement) can be written as 
\[
A_R O_{f_R}V_R\ket{\psi_{R,k}}
\]
for input $\ket{\psi_{R,k}}$, where $V_R:\mathbb C^N\to \mathbb C^M$ is the isometry that includes the $t$ pre-queries (with fixed ancillas absorbed in this isometry). All gates other than the oracle calls are fixed independently of
$R$ and the hidden challenge key $k$. The function $f_R$ may depend arbitrarily on $R$, but not on the challenge. The dimension $M$ includes all workspace and retained environments.

\begin{theorem}[One-query lower bound with pre-queries and post-queries]\label{thm:1-query-pre-post-queries}
Let $\{\ket{\psi_{R,k}}\}_{k\in [K]}$ be a binary phase state family with $R=(R(k,x))_{k\in[K],x\in[N]}$ as a family of independent Rademacher random variables. For the variants of the oracle state search game and distinguishing game, where the adversary is allowed to make $t$ ($t\ge 1$) queries before and $T$ queries after the arbitrary phase query, the security satisfies:
\begin{align*}
    \mathbb{E}_R[\Win(\mathcal{A}\mid R)]
    &\le O\left(
    \min\left\{
    1, \frac{t^2}{K}+\sqrt{\frac{tT^2+t(T+1)\log M}{K}}
    \right\}
    \right) \tag{Search}
    \\
    \mathbb{E}_R[\Delta_\cA(R)]
    &\le
    O\left(
    \min\left\{
    1, \frac{t}{\sqrt{K}}+\left(\frac{tT^2+t(T+1)\log M}{K}\right)^{1/4}
    \right\}
    \right) \tag{Distinguishing}
\end{align*}
The case $t=0$ is covered separately by \Cref{thm:binary-phase-search-post,thm:binary-phase-distinguish-post}.
\end{theorem}

\paragraph{Proof overview.} We are going to reduce the $t$-pre-query case to a $0$-pre-query case, so that we can apply \Cref{thm:binary-phase-search-post,thm:binary-phase-distinguish-post}.

The proof uses a swapping or resampling trick similar to the hybrids introduced in the previous proofs, which helps reduce the security to the case where there are no pre-queries. More specifically, here we introduce two hybrids, by modifying only the pre-query part of the original algorithm:
\begin{itemize}
    \item pre-queries to $R$ are replaced by $R^{\{k\}}$, and
    \item pre-queries to $R$ are replaced by $R^{J}$, for randomly sampled subset $J$, and $k\gets J$,
\end{itemize}
and here $R^{J}$ denotes resampling $R$ by replacing every $\{R(k,x):k\in J\}$ by a uniformly random $\pm 1$; that is, resampling the rows in $J$. 
The difference between hybrids can then be bounded according to the following intuitions.
\begin{itemize}
    \item When pre-queries to $R$ replaced by $R^{\{k\}}$, the difference only comes from the case when the query address is on $k$. This gives the difference of real and $R^{(k)}$ bounded by $t/\sqrt{K}$.
    \item For the $R^J$ where $k\in J$, the difference only comes from the case when some row in $J\backslash\{k\}$ is queried, which can be managed in $\sqrt{t(|J|-1)/K}$. (This is different from the above difference, since $J\backslash\{k\}$ is randomly sampled and independent of the challenge state.)
    \item When we are in the $R^J$ case, we can view the pre-query part as a fixed unitary by fixing all $[K]\backslash J$ parts and the resampling $J$ rows in $R^J$. The security in this case is then reduced to the one-query with (only) post-query case, where we can apply \Cref{thm:binary-phase-search-post,thm:binary-phase-distinguish-post}. 
\end{itemize}

\begin{proof}[Proof of \Cref{thm:1-query-pre-post-queries}]
Fix $1\le L\le K$. Sample a uniform $L$-element subset $J\subseteq[K]$, and then a uniform $k\in J$, independently of $R$. The marginal distribution of $k$ is still uniform on $[K]$. Independently sample another table $R'$. For $B\subseteq[K]$, let $R^B$ agree with $R'$ on $B$ and with $R$ outside $B$. Define the following three prefix output states:
\[
\ket{\alpha_{\mathrm{real}}}=V_R\ket{\psi_{R,k}},\qquad
\ket\alpha_1=V_{R^{\{k\}}}\ket{\psi_{R,k}},\qquad
\ket\alpha_2=V_{R^J}\ket{\psi_{R,k}}.
\]
Only the oracle used in the pre-$t$-query is replaced; the challenge $\ket{\psi_{R,k}}$, the arbitrary query $O_{f_R}$, and the post-processing $A_R$ remain unchanged.

For $R^{\{k\}}$, the joint distribution $(R^{\{k\}},\ket{\psi_{R,k}})$ is independent of $k$. Let $\ket{u_i}$ be the state just before the $i$-th query of $V_{R^{\{k\}}}$ on challenge $\ket{\psi_{R,k}}$, and let $Q_k$ be the projector onto query address with row (key) $k$. Then, since $\ket{u_i}$ remains independent of $k$, $\mathbb{E}_k\|Q_k\ket{u_i}\|^2=1/K$. Note that oracle to $R^{\{k\}}$ and $R$ differ only at query address with row on $k$, by the triangle inequality,
\begin{align}\label{ineq:pre-query-real-and-1}
    \sqrt{\mathbb{E}_k \left\|\ket{\alpha_{\mathrm{real}}}-\ket{\alpha_1}\right\|^2}
    \le \sum_{i=1}^t\sqrt{
    \mathbb{E}_k\left\|(O_R-O_{R^{\{k\}}})\ket{u_i}\right\|^2
    }
    \le \sum_{i=1}^t 2\cdot \frac{1}{\sqrt{K}}=\frac{2t}{\sqrt{K}}.
\end{align}

For $R^J$, consider the set $J\backslash \{k\}$. Conditioned on $(k, J\backslash \{k\})$, now $R^{\{k\}}$, the challenge $\ket{\psi_{R,k}}$, and the resampled rows (the $J\backslash \{k\}$ rows in $R^J$) are mutually independent. Write $\ket{\psi_{R,k}}$ as $\ket{\phi}$. 
For each fixed normalized input $\ket{\phi}$, the $t$-query computation admits the Fourier expansion
\[
V_F\ket{\phi}
=
\sum_{|S|\le t}\chi_S(F)\ket{a_S(\phi)},
\qquad
\sum_S\|\ket{a_S(\phi)}\|^2=1,
\]
where $\chi_S(F)=\prod_{(j,x)\in S}F(j,x)$, and define $\operatorname{rows}(S)=\{j:\exists x,\ (j,x)\in S\}$.
Indeed, each query increases the Fourier degree by at most one,
and the normalization follows from Parseval. Conditional on $(\ket{\phi},J\backslash\{k\})$, the two oracle tables agree outside $J\backslash\{k\}$ and have independent uniform entries on $J\backslash\{k\}$. Thus
\[
\mathbb E\!\left[
\chi_S(R^{\{k\}})\chi_{S'}(R^J)
\,\middle|\,\ket{\phi},J\backslash\{k\}
\right]
=
\begin{cases}
1,
& S=S' \text{ and }
  \operatorname{rows}(S)\cap J\backslash\{k\}=\varnothing,\\
0,
& \text{otherwise}.
\end{cases}
\]
Therefore, averaging over $\ket{\phi}, J\backslash\{k\}$,
\begin{align}
\mathbb E\|\ket{\alpha_1}-\ket{\alpha_2}\|^2
&=2-
2\cdot \mathbb{E}\Re\left<\alpha_1|\alpha_2\right>
\notag
\\
&=2\cdot \mathbb E_{\ket\phi}\sum_S
\Pr_{J\backslash\{k\}}\!\left[
\operatorname{rows}(S)\cap J\backslash\{k\}\ne\varnothing
\right]
\cdot
\left\|\ket{a_S(\phi)}\right\|^2
\notag
\\
&\le
\frac{2(L-1)}{K}\,
\mathbb E_{\ket\phi}\sum_S
|\operatorname{rows}(S)|
\cdot \left\|\ket{a_S(\phi)}\right\|^2
\tag{union bound}
\\
&\le \frac{2t(L-1)}{K}.
\label{ineq:pre-query-1-and-2}
\end{align}
Combining \cref{ineq:pre-query-real-and-1,ineq:pre-query-1-and-2}, we conclude
\begin{align}\label{ineq:pre-query-prefix-part}
\sqrt{\mathbb E
\|\ket{\alpha_{\mathrm{real}}}-\ket{\alpha_2}\|^2}
\le
\frac{2t}{\sqrt K}
+
\sqrt{\frac{2t(L-1)}{K}}.
\end{align}

\paragraph{For the search game.} For the winning probability of the search game, by the triangle inequality,
\begin{align*}
    \mathbb{E}_R [\Win(\cA\mid R)]
    &=\mathbb{E}_{R,k}\left\|
    \Pi_k A_RO_{f_R}V_R\ket{\psi_{R,k}}
    \right\|^2
    \\
    &\le \left(
    \sqrt{
    \mathbb{E}_{R,k,J,R'} \left\|\alpha_{\mathrm{real}}-\ket{\alpha_2}\right\|^2
    }
    +\sqrt{
    \mathbb{E}_{R,k,J,R'} \left\|\Pi_kA_RO_{f_R}V_{R^J}\ket{\psi_{R,k}}\right\|^2
    }
    \right)^2.
\end{align*}
The first term can be bounded by \cref{ineq:pre-query-prefix-part}. For the second term, now we fix $J$, $R_{[K]\backslash J}, R'_J$. $V_{R^J}$ now is a fixed isometry, and queries outside of $J$ have fixed answers. Still, $f_R$ can depend arbitrarily on $R_J$. Therefore, from \Cref{thm:binary-phase-search-post}, we can bound the term by $C\cdot (T^2+(T+1)\log M)/L$ on keys in $J$ (for constant $C$), and
\begin{align*}
    \mathbb{E}_R [\Win(\cA\mid R)]
    &\le \left(
    \sqrt{
    \mathbb{E}_{R,k,J,R'} \left\|\alpha_{\mathrm{real}}-\ket{\alpha_2}\right\|^2
    }
    +\sqrt{
    \mathbb{E}_{R,k,J,R'} \left\|\Pi_kA_RO_{f_R}V_{R^J}\ket{\psi_{R,k}}\right\|^2
    }
    \right)^2
    \\
    &\le C'\cdot\left(
    \frac{t^2}{K}+\frac{tL}{K}+\frac{T^2+(T+1)\log M}{L}
    \right)
    =O\left(\frac{t^2}{K}+\sqrt{\frac{t(T^2+(T+1)\log M)}{K}}\right)
\end{align*}
by setting $L=\sqrt{K(T^2+(T+1)\log M)/t}$.

\paragraph{For the distinguishing game.}
Let $\Delta_{\cA}^{J}(R,R')$ denote the gap after replacing $V_R$
by $V_{R^J}$ in both experiments:
\[
\Delta_{\cA}^{J}(R,R')
:=
\left|
\mathbb E_{k\in J}
\left\|\Pi A_RO_{f_R}V_{R^J}\ket{\psi_{R,k}}\right\|^2
-
\mathbb E_{r\in[N]}
\left\|\Pi A_RO_{f_R}V_{R^J}\ket{r}\right\|^2
\right|.
\]
For each fixed $R$, the original distribution of $k$ can be viewed as taking the average over $J$ and then taking $k$ uniform in $J$. By convexity of absolute value, it therefore
suffices to bound the average distinguishing gap in these
restricted experiments on a random $J$.
Thus, by the triangle inequality (adding and subtracting the acceptance probabilities of the two experiments with prefix $V_{R^J}$), together with $\left|\|\Pi u\|^2-\|\Pi v\|^2\right|\le (\|\Pi u\|+\|\Pi v\|)\|\Pi(u-v)\|\le 2\|\Pi(u-v)\|$,
\begin{align*}
\mathbb E_R\Delta_{\cA}(R)
&\le
\mathbb{E}_{R,J}
\left|
\mathbb{E}_{k\in J}\left\|\Pi A_RO_{f_R}V_{R}\ket{\psi_{R,k}}\right\|^2
-
\mathbb E_{r\in[N]}
\left\|\Pi A_RO_{f_R}V_{R}\ket{r}\right\|^2
\right|
\\
&\le
\mathbb E_{R,R',J}\Delta_{\cA}^{J}(R,R')
+
2\cdot \mathbb E_{R,R',J,k\in J}
\left\|
\Pi A_RO_{f_R}(V_R-V_{R^J})\ket{\psi_{R,k}}
\right\|\\
&\quad +
2\cdot \mathbb E_{R,R',J,r\in[N]}
\left\|
\Pi A_RO_{f_R}(V_R-V_{R^J})\ket{r}
\right\|
\\
&\le C\cdot \sqrt{\frac{T^2+(T+1)\log M}{L}}+2\cdot\left(\frac{t}{\sqrt{K}}+\sqrt{\frac{2t(L-1)}{K}}\right)+
2\cdot \mathbb E_{R,R',J,r}
\left\|
(V_R-V_{R^J})\ket{r}
\right\|
\tag{by \Cref{thm:binary-phase-distinguish-post}, and Jensen's inequality followed from \cref{ineq:pre-query-prefix-part}}
\end{align*}
Similar to the calculation in \cref{ineq:pre-query-1-and-2},
we obtain
\begin{align*}
\mathbb E_{R,R',J,r}
\left\|(V_R-V_{R^J})\ket{r}\right\|^2
=
2\cdot \mathbb E_r\sum_S
\Pr_J[\operatorname{rows}(S)\cap J\ne\varnothing]\,
\|\ket{a_S(r)}\|^2
\le \frac{2tL}{K}.
\end{align*}
By Jensen's inequality, combining the above gives 
\begin{align*}
\mathbb E_R\Delta_{\cA}(R)
&\le
C\sqrt{\frac{T^2+(T+1)\log M}{L}}
+\frac{4t}{\sqrt K}
+2\sqrt{\frac{2t(L-1)}{K}}
+2\sqrt{\frac{2tL}{K}}
\\
&=O\left(
\frac{t}{\sqrt{K}}+\left(\frac{tT^2+t(T+1)\log M}{K}\right)^{1/4}
\right)
\end{align*}
by setting $L=\sqrt{K(T^2+(T+1)\log M)/t}$.

\end{proof}

\section{Quantum advice lower bound for 1PRS} \label{sec:quantum-program-1prs}
In this section, we prove a lower bound on the size of a quantum advice for distinguishing a single copy of a random binary phase state from a Haar-random state, when there are no online queries to the random oracle. This can be viewed as a distinguishing version of \Cref{thm:1ows-ts-tradeoff} with $T=0$. Specifically, we show that any $S$-qubit quantum advice (or say, quantum program) can distinguish a random binary phase state from Haar random with distinguishing advantage at most $O\left(\frac{\sqrt{S}}{K}\right)$.

\begin{definition}[1PRS security game] In the non-uniform security game of 1PRS:
\begin{itemize}
    \item In the offline stage, a hash function $R:[K]\times [N]\to\{\pm 1\}$ is sampled and fixed (here $K\le N$), and an algorithm is allowed to have an $S$-qubit advice that depends arbitrarily on $R$ before receiving the challenge.
    \item Then, in the online stage, a challenger will sample a bit $b\gets\{0,1\}$.
    \begin{itemize}
        \item If $b=0$, the challenger will sample a random key $k\gets [K]$, generate a single-copy of state $|\psi_{R,k}\rangle:=\frac{1}{\sqrt N}\sum_{x\in[N]}R(k,x)|x\rangle$, and send it to the algorithm.
        \item If $b=1$, the challenger will sample a Haar random state and send it to the algorithm.
    \end{itemize}
    \item The algorithm, after receiving the challenge state, is allowed to make $T$ queries to $R$ and is asked to output a guess $b'$. It succeeds if $b'=b$.
\end{itemize}
We consider security game in QROM, for $R$ as a random oracle, and the algorithm can have oracle access to $O_R$. The non-uniform security (in terms of $S$-qubit, $T$-query) of 1PRS in QROM is defined as the maximum distinguishing advantage, parametrized by $(S,T,K,N)$.
\end{definition}
Note here we consider quantum advice lower bound, with zero online queries. That is, now we assume $T=0$.

\begin{theorem}[Quantum advice lower bound for 1PRS]\label{thm:program-1prs}
Any (zero-query) adversary with $S$ qubit of advice has distinguishing advantage bounded by $O\left(\frac{\sqrt{S}}{K}\right)$ in QROM. Consequently, constant advantage requires $\Omega(K^2)$ qubits (for $K=\Theta(N)$, this requires $\Omega(N^2)$ qubits).
\end{theorem}

\subsection{1PRS trace-moment setup}

Let $\mathsf A$ be the advice register and $\mathsf X$ the challenge register. Absorb the fixed online computation into an acceptance operator $0\preceq\Pi\preceq I$ on $\mathsf A\otimes \mathsf X$, and set
\[
    \Pi_{x,y}:=(I_\mathsf{A}\otimes\langle x|)\Pi(I_\mathsf{A}\otimes|y\rangle).
\]
The centered acceptance matrix on the advice register is
\begin{align}
 Y_R
 &:=\mathbb E_k(I_\mathsf{A}\otimes\langle\psi_{R,k}|)
          \Pi(I_\mathsf{A}\otimes|\psi_{R,k}\rangle)
          -\frac1N\sum_x\Pi_{x,x}
    =\frac1{KN}\sum_k\sum_{x\ne y}R(k,x)R(k,y)\Pi_{x,y}.
 \label{eq:program-centered}
\end{align}
Here the diagonal terms cancel because $R(k,x)^2=1$. Thus
\[
    \Delta_{\mathcal A}(R)=|\operatorname{Tr}(\rho_RY_R)|\le\|Y_R\|,
\]
and now our goal is to bound $\mathbb E_R\|Y_R\|$.

In the compressed oracle database mode, for a database $D\subseteq[K]\times[N]$, define the Fourier character as
$\chi_D(R)=\prod_{(k,x)\in D}R(k,x)$, and define
\[
 T_{k,x}|D\rangle:=|D\oplus\{(k,x)\}\rangle,
 \qquad
 \widehat Y:=\frac1{KN}\sum_k\sum_{x\ne y}
               \Pi_{x,y}\otimes T_{k,x}T_{k,y}.
\]
The usual trace-moment method gives, for every positive integer $t$,
\begin{equation}
 \mathbb E_R\|Y_R\|
 \le\left(\mathbb E_R\operatorname{Tr}(Y_R^{2t})\right)^{1/2t},
 \qquad
 \mathbb E_R\operatorname{Tr}(Y_R^{2t})
 =\sum_{z\in[2^S]}\|\widehat Y^t|z\rangle|\emptyset\rangle\|^2.
 \label{eq:program-compressed-moment}
\end{equation}
We will use the following local bound.

\begin{lemma}[Small-database bound]
\label{lem:program-local}
For every normalized state $|\phi\rangle$ supported on databases of size at most $\ell$,
\begin{equation}
 \left|\langle\phi|\widehat Y|\phi\rangle\right|
 \le\beta(\ell):=\frac{\sqrt\ell}{K}
               +2\sqrt{\frac\ell{KN}}+\frac{3\ell}{KN}.
 \label{eq:program-local-bound}
\end{equation}
\end{lemma}

\begin{proof}[Proof of \Cref{thm:program-1prs}]
Let $P_{\le\ell}$ project onto databases of size at most $\ell$. Since $\widehat Y$ is Hermitian, \Cref{lem:program-local} implies
\[
    \|P_{\le\ell}\widehat YP_{\le\ell}\|\le\beta(\ell).
\]
Each application of $\widehat Y$ changes the database size by at most two. Therefore,
\[
 \widehat Y^t|z\rangle|\emptyset\rangle
 =(P_{\le2t}\widehat YP_{\le2t})^t|z\rangle|\emptyset\rangle,
 \qquad
 \|\widehat Y^t|z\rangle|\emptyset\rangle\|\le\beta(2t)^t.
\]
setting $t=S/2$ gives an upper bound for \cref{eq:program-compressed-moment} as
\[
 \mathbb E_R\|Y_R\|
 \le 2^{S/2t}\beta(2t)
 \le O\!\left(\frac{\sqrt S}{K}+\sqrt{\frac S{KN}}+\frac S{KN}\right).
\]
When $K\le N$ and $S\le K^2$, the last two terms are at most $\sqrt S/K$. For $S>K^2$, use the trivial advantage bound $1$.
\end{proof}

\subsection{Proof of the small-database bound}\label{sec:prs-small-database}
We prove \Cref{lem:program-local} by comparing the real challenge with an independent maximally mixed state. Note that challenge preparation either adds a new database entry or deletes an existing one, while deletions have small total weight. For additions, the main issue is that different choices of the newly added entry can lead to the same final database and hence interfere. We express this interference as a sum of swaps, which we bound in the next subsection.

\begin{proof}[Proof of \Cref{lem:program-local}]
Write the initial state as
\[
  |\varphi\rangle
  =\sum_D |\varphi_D\rangle_{\mathsf A}|D\rangle_{\mathsf{db}},
  \qquad
  \sum_D \bigl\|\,|\varphi_D\rangle\,\bigr\|^2=1,
\]
where $|\varphi_D\rangle=0$ when $|D|>\ell$. A purification of the real
challenge is
\[
  |\Psi\rangle
  =\frac{1}{\sqrt{KN}}
    \sum_{D,k,x}|\varphi_D\rangle_{\mathsf A}|x\rangle_{\mathsf X}
    |D\oplus\{(k,x)\}\rangle_{\mathsf{db}}|k\rangle_{\mathsf{key}}.
\]
The algorithm sees only $\mathsf A\mathsf X$; the database and key registers
are inaccessible.

\paragraph{Separating deletions from new entries.}
Split $|\Psi\rangle=|\Psi_0\rangle+|\Psi_1\rangle$ according to whether
$(k,x)\notin D$ or $(k,x)\in D$, respectively. These two parts are orthogonal,
and the deletion part has squared norm
\[
  q:=\bigl\|\,|\Psi_1\rangle\,\bigr\|^2
   =\frac{1}{KN}\sum_D |D|\,
       \bigl\|\,|\varphi_D\rangle\,\bigr\|^2
   \le\frac{\ell}{KN},
  \qquad
  \bigl\|\,|\Psi_0\rangle\,\bigr\|^2=1-q.
\]
Let $\sigma_b$ be the reduced, subnormalized state of $|\Psi_b\rangle$ on
$\mathsf A\mathsf X$. The independent-challenge state on these registers is
\[
  \tau:=\left(\sum_D|\varphi_D\rangle\langle\varphi_D|\right)
            \otimes\frac{I_{\mathsf X}}{N}.
\]

\paragraph{Identifying the surviving interference.}
When we trace out the hidden registers, two new-entry terms can interfere
only if they have the same final database $S$ and key $k$. They may, however,
have added different entries: adding $(k,x)$ to $S\setminus\{(k,x)\}$ and
adding $(k,y)$ to $S\setminus\{(k,y)\}$ both produce $S$. Thus the two histories exchange the roles of the challenge column and a
stored column.

Write $S_k:=\{x:(k,x)\in S\}$ and
$|\varphi_{D,x}\rangle:=|\varphi_D\rangle|x\rangle$. The off-diagonal terms
are collected in
\[
  B_\varphi:=\sum_{S,k}
      \sum_{\substack{x,y\in S_k\\x\ne y}}
      |\varphi_{S\setminus\{(k,y)\},y}\rangle
      \langle\varphi_{S\setminus\{(k,x)\},x}|.
\]
The diagonal terms recover $\tau$, except for the terms excluded by the
new-entry condition. Thus
\begin{equation}\label{eq:1prs-rewrite-diagonal}
  \sigma_0-\tau=\frac{B_\varphi}{KN}-H,
  \qquad
  H:=\frac{1}{KN}\sum_{\substack{D,k,x\\(k,x)\in D}}
        |\varphi_{D,x}\rangle\langle\varphi_{D,x}|.
\end{equation}
Both $H$ and $\sigma_1$ are positive with trace $q$. Hence, for the acceptance
operator $0\preceq\Pi\preceq I$,
$|\operatorname{Tr}(\Pi(\sigma_1-H))|\le q$. The cross terms $\bra{\Psi_0}\Pi\otimes I\ket{\Psi_1}$ between
$|\Psi_0\rangle$ and $|\Psi_1\rangle$ contribute at most
$2\sqrt{q(1-q)}\le 2\sqrt{q}$, by Cauchy-Schwarz. Consequently,
\begin{align}\label{eq:1prs-rewrite-gap}
\langle\varphi|\widehat Y|\varphi\rangle
&=\langle\Psi|(\Pi\otimes I)|\Psi\rangle
-\operatorname{Tr}(\Pi\tau)
\notag
\\
&=\frac{\operatorname{Tr}(\Pi B_\varphi)}{KN}
+\operatorname{Tr}\bigl(\Pi(\sigma_1-H)\bigr)
+2\operatorname{Re}
\langle\Psi_0|(\Pi\otimes I)|\Psi_1\rangle
\notag
\\
&\le \frac{\|B_\varphi\|_1}{KN}+q+2\sqrt q.
\end{align}
Here $\|\cdot\|_1$ denotes the trace norm. It remains to prove $\|B_\varphi\|_1\le N\sqrt\ell+2\ell$.

\paragraph{Writing the interference as swaps.}
First suppose the initial state is normalized and supported on databases of
exactly size $m\ge1$. To give each stored entry a slot, now we encode a database $D=\{(k_1,x_1),\ldots,(k_m,x_m)\}$ by the uniform superposition of its
orderings (as a different purification of the database register):
\[
  |\operatorname{ord}(D)\rangle
  :=\frac{1}{\sqrt{m!}}\sum_{\pi\in S_m}
       |k_{\pi(1)},x_{\pi(1)}\rangle\cdots
       |k_{\pi(m)},x_{\pi(m)}\rangle.
\]
Define
\[
  L:=\sum_{|D|=m}|\varphi_D\rangle\langle\operatorname{ord}(D)|.
\]
Thus $L$ maps the ordered encoding of $D$ to its advice amplitude
$|\varphi_D\rangle$. The ordered encodings are orthonormal, so
$\operatorname{Tr}(L^\dagger L)=1$.

Let $\mathsf X_0$ be the challenge-column register and
$\mathsf X_1,\ldots,\mathsf X_m$ the stored column registers. Define
\[
  J_m:=\sum_{i=1}^m\operatorname{Swap}_{0i},
  \qquad
  E_m:=\sum_{i=1}^m\sum_{x\in[N]}
       |x,x\rangle\langle x,x|_{\mathsf X_0\mathsf X_i},
\]
with identities on all other registers, including the row registers.
The operator $J_m$ sums all exchanges of the challenge column with a stored
column; $E_m$ removes the equal-column exchanges, which are diagonal terms.
We claim that
\begin{equation}\label{eq:1prs-rewrite-swaps}
  B_\varphi=(L\otimes I)(J_m-E_m)(L^\dagger\otimes I).
\end{equation}
Indeed, starting with $S\setminus\{(k,x)\}$ and visible column $x$, swapping
with the stored entry $(k,y)$ produces $S\setminus\{(k,y)\}$ and visible
column $y$. This is precisely the corresponding term of $B_\varphi$.
The swap pairs the $m!$ orderings bijectively, giving coefficient one.
Equal-column swaps are canceled by $E_m$, and a swap creating a repeated
entry is killed by $L$, since the resulting tuple is not an encoding of a
set.

The advantage of \cref{eq:1prs-rewrite-swaps} is that all dependence on the
advice is now contained in $L$. For any Hermitian $J$ and any $C$,
\[
  \|CJC^\dagger\|_1\le\operatorname{Tr}(C^\dagger C\,|J|),
\]
by splitting $J$ into its positive and negative parts. Also,
$E_m\succeq0$ and $\operatorname{Tr}_{\mathsf X_0}E_m=mI$. Therefore,
\[
  \|B_\varphi\|_1
  \le\operatorname{Tr}\!\left[
       L^\dagger L\,\operatorname{Tr}_{\mathsf X_0}|J_m|
     \right]+m.
\]
Since $L^\dagger L$ is positive with trace one, it suffices to bound the
partial trace uniformly. \Cref{lem:1prs-rewrite-swap}, proved below,
gives
\[
  \operatorname{Tr}_{\mathsf X_0}|J_m|
  \preceq(N\sqrt m+m)I,
  \qquad
  \|B_\varphi\|_1\le N\sqrt m+2m.
\]
The case $m=0$ is immediate, since there are no off-diagonal terms.

For a general state, let $p_m$ be its squared norm on databases of size $m$.
Every term of $B_\varphi$ pairs initial databases of the same size: both
have size $|S|-1$. Applying the exact-size bound to each normalized
component and then the triangle inequality gives
\[
  \|B_\varphi\|_1
  \le\sum_m p_m(N\sqrt m+2m)
  \le N\sqrt\ell+2\ell.
\]
Substituting this and $q\le\ell/(KN)$ into \Cref{eq:1prs-rewrite-gap},
we obtain
\[
  \bigl|\langle\varphi|\widehat Y|\varphi\rangle\bigr|
  \le\frac{N\sqrt\ell+2\ell}{KN}
       +\frac{\ell}{KN}+2\sqrt{\frac{\ell}{KN}}
  =\frac{\sqrt\ell}{K}
       +2\sqrt{\frac{\ell}{KN}}+\frac{3\ell}{KN},
\]
proving \Cref{lem:program-local}.
\end{proof}

\subsection{An elementary bound for the swap operator}

\begin{lemma}\label{lem:1prs-rewrite-swap}
On $(\mathbb C^N)^{\otimes(m+1)}$, let
$J_m=\sum_{i=1}^m\operatorname{Swap}_{0i}$ and
$|J_m|=(J_m^2)^{1/2}$. Then
\[
  \operatorname{Tr}_{\mathsf X_0}|J_m|
  \preceq\sqrt{m(N^2+m)}\,I
  \preceq(N\sqrt m+m)I.
\]
The same bound holds with additional registers on which all swaps act
trivially.
\end{lemma}

\begin{proof}
The case $m=0$ is immediate. A direct operator-norm bound would give only
$NmI$, since $\|J_m\|_{\mathsf{op}}=m$. To exploit the partial trace, we
introduce an auxiliary model with $m$ occupied positions and one movable
empty slot. Swaps starting with the empty slot at position $0$ move it to orthogonal positions, giving a norm of $\sqrt m$.
The main step is a general identity converting the partial trace of a permutation into a block of its auxiliary operator. 
This helps transfer back to the original operator.

\paragraph{Permutations with an empty slot.}
The auxiliary space is
\[
  \mathbb C^{m+1}\otimes(\mathbb C^N)^{\otimes m}.
\]
Its first register records the position $j\in\{0,\ldots,m\}$ of the empty
slot; the other registers store the $m$ column values. We use the following
indexing convention: a basis state $|j\rangle|x_1,\ldots,x_m\rangle$ encodes
the configuration obtained from
\[
  [\,\square,x_1,\ldots,x_m\,]
\]
by swapping positions $0$ and $j$. In particular, moving the empty slot
from $0$ to $i$ changes the label from $0$ to $i$ without changing the stored
tuple in this encoding.

For a permutation $\pi$ of the $m+1$ positions, let $Q_\pi$ permute this configuration, moving both the contents and the empty slot. These are unitary permutation matrices and satisfy
$Q_\pi Q_\sigma=Q_{\pi\sigma}$, just as the usual register-permutation operators $U_\pi$ do (while the difference is, $Q_\pi$ applies on the auxiliary space $\mathbb C^{m+1}\otimes(\mathbb C^N)^{\otimes m}$). In both cases, a value at position $j$ moves to
position $\pi(j)$.

Define the Hermitian operator
\[
  \widetilde J:=\sum_{i=1}^m Q_{(0,i)},
\]
and write $(\cdot)_{00}$ for the block with empty-slot input and output
label $0$. For every stored-register vector $|v\rangle$,
\begin{equation}\label{eq:1prs-rewrite-empty-slot}
  \begin{aligned}
    \widetilde J(|0\rangle\otimes|v\rangle)
       &=\sum_{i=1}^m |i\rangle\otimes|v\rangle,\qquad 
    (\widetilde J)_{00}=0,
    \qquad (\widetilde J^2)_{00}=mI.
  \end{aligned}
\end{equation}
The different swaps give different empty-slot labels, so their
contributions are orthogonal. This is the source of the square-root
dependence on $m$.

\paragraph{Converting the partial trace into a matrix block.}
For every permutation $\pi$,
\begin{equation}\label{eq:1prs-rewrite-permutation-trace}
  \operatorname{Tr}_{\mathsf X_0}U_\pi
  =N(Q_\pi)_{00}+(\widetilde JQ_\pi)_{00}.
\end{equation}
To see this, consider whether $\pi$ fixes position $0$. 
\begin{itemize}
    \item If it does, $\pi(0)=0$, tracing out $\mathsf X_0$ gives $N$ times the permutation of the remaining registers. In the auxiliary model the empty slot stays at $0$, so this is exactly the first term on the right; the second term is zero.
    \item Otherwise, let $i=\pi(0)\ne0$. Tracing out $\mathsf X_0$ deletes $0$ from its permutation cycle, with no factor of $N$. In the auxiliary model, $Q_\pi$ moves the empty slot to $i$, so $(Q_\pi)_{00}=0$. Exactly one summand of $\widetilde J$ returns it to $0$, namely $Q_{(0,i)}$. On the remaining registers, $(0,i)\pi$ performs precisely the permutation obtained by deleting $0$ from its cycle. Thus the second term on the right gives the required partial trace.
\end{itemize}

Because the $U_\pi$ and $Q_\pi$ obey the same multiplication rules,
expanding any polynomial $f$ into products of swaps extends
\cref{eq:1prs-rewrite-permutation-trace} to
\[
  \operatorname{Tr}_{\mathsf X_0}f(J_m)
   =N\bigl(f(\widetilde J)\bigr)_{00}
       +\bigl(\widetilde Jf(\widetilde J)\bigr)_{00}.
\]
Choose $f$ to agree with absolute value on the finite spectra of both
$J_m$ and $\widetilde J$. This gives the exact identity
\begin{equation}\label{eq:1prs-rewrite-absolute-trace}
  \operatorname{Tr}_{\mathsf X_0}|J_m|
    =N(|\widetilde J|)_{00}
       +(\widetilde J|\widetilde J|)_{00}.
\end{equation}

\paragraph{Applying Cauchy--Schwarz.}
Fix a unit vector $|v\rangle$ on the stored registers and set
$|u\rangle=|0\rangle\otimes|v\rangle$. By
\cref{eq:1prs-rewrite-empty-slot},
\[
  \langle u|\widetilde J|u\rangle=0,
  \qquad
  \bigl\|\,\widetilde J|u\rangle\,\bigr\|^2=m,
  \qquad
  \bigl\|\,|\widetilde J|\,|u\rangle\,\bigr\|^2=m,
\]
where the last equality uses $|\widetilde J|^2=\widetilde J^2$.
Equation \cref{eq:1prs-rewrite-absolute-trace} and Cauchy--Schwarz give
\[
  \begin{aligned}
    \langle v|\operatorname{Tr}_{\mathsf X_0}|J_m|\,|v\rangle
      &=\langle u|(NI+\widetilde J)|\widetilde J|\,|u\rangle\\
      &\le
       \bigl\|\,(NI+\widetilde J)|u\rangle\,\bigr\|
       \bigl\|\,|\widetilde J|\,|u\rangle\,\bigr\|\\
      &=\sqrt{N^2+m}\,\sqrt m.
  \end{aligned}
\]
This holds for every unit vector $|v\rangle$, proving the operator bound.
Tensoring with an identity gives the assertion about additional registers.
\end{proof}

\bibliographystyle{alpha}
\bibliography{custom,abbrev3,crypto}

\newcommand{\etalchar}[1]{$^{#1}$}
\begin{thebibliography}{BCKM21}

\bibitem[AK07]{AK07}
Scott Aaronson and Greg Kuperberg.
\newblock Quantum versus classical proofs and advice.
\newblock {\em Theory of Computing}, 3(1):129--157, 2007.

\bibitem[BBBV97]{BBBV97}
Charles~H. Bennett, Ethan Bernstein, Gilles Brassard, and Umesh Vazirani.
\newblock Strengths and weaknesses of quantum computing.
\newblock {\em SIAM Journal on Computing}, 26(5):1510–1523, October 1997.

\bibitem[BCKM21]{C:BCKM21b}
James Bartusek, Andrea Coladangelo, Dakshita Khurana, and Fermi Ma.
\newblock One-way functions imply secure computation in a quantum world.
\newblock In Tal Malkin and Chris Peikert, editors, {\em CRYPTO~2021, Part~I},
  volume 12825 of {\em {LNCS}}, pages 467--496, Virtual Event, August 2021.
  Springer, Cham.

\bibitem[BDF{\etalchar{+}}11]{AC:BDFLSZ11}
Dan Boneh, {\"O}zg{\"u}r Dagdelen, Marc Fischlin, Anja Lehmann, Christian
  Schaffner, and Mark Zhandry.
\newblock Random oracles in a quantum world.
\newblock In Dong~Hoon Lee and Xiaoyun Wang, editors, {\em ASIACRYPT~2011},
  volume 7073 of {\em {LNCS}}, pages 41--69. Springer, Berlin, Heidelberg,
  December 2011.

\bibitem[CDGS18]{EC:CDGS18}
Sandro Coretti, Yevgeniy Dodis, Siyao Guo, and John~P. Steinberger.
\newblock Random oracles and non-uniformity.
\newblock In Jesper~Buus Nielsen and Vincent Rijmen, editors, {\em
  EUROCRYPT~2018, Part~I}, volume 10820 of {\em {LNCS}}, pages 227--258.
  Springer, Cham, April~/~May 2018.

\bibitem[CGLQ20]{FOCS:CGLQ20}
Kai-Min Chung, Siyao Guo, Qipeng Liu, and Luowen Qian.
\newblock Tight quantum time-space tradeoffs for function inversion.
\newblock In {\em 61st FOCS}, pages 673--684. {IEEE} Computer Society Press,
  November 2020.

\bibitem[CK19]{TCC:CorKog19}
Henry {Corrigan-Gibbs} and Dmitry Kogan.
\newblock The function-inversion problem: Barriers and opportunities.
\newblock In Dennis Hofheinz and Alon Rosen, editors, {\em TCC~2019, Part~I},
  volume 11891 of {\em {LNCS}}, pages 393--421. Springer, Cham, December 2019.

\bibitem[CLQ20]{ITC:ChuLiaQia20}
Kai-Min Chung, Tai-Ning Liao, and Luowen Qian.
\newblock Lower bounds for function inversion with quantum advice.
\newblock In Yael~Tauman Kalai, Adam~D. Smith, and Daniel Wichs, editors, {\em
  ITC 2020}, volume 163 of {\em {LIPIcs}}, pages 8:1--8:15. Schloss Dagstuhl,
  June 2020.

\bibitem[DGK17]{EC:DodGuoKat17}
Yevgeniy Dodis, Siyao Guo, and Jonathan Katz.
\newblock Fixing cracks in the concrete: Random oracles with auxiliary input,
  revisited.
\newblock In Jean-S{\'{e}}bastien Coron and Jesper~Buus Nielsen, editors, {\em
  EUROCRYPT~2017, Part~II}, volume 10211 of {\em {LNCS}}, pages 473--495.
  Springer, Cham, April~/~May 2017.

\bibitem[DLM26]{DLM26}
Fangqi Dong, Alex Lombardi, and Fermi Ma.
\newblock Explicit separations for one-query unitary synthesis, 2026.

\bibitem[DLW24]{C:DonLiuWu24}
Fangqi Dong, Qipeng Liu, and Kewen Wu.
\newblock Tight characterizations for preprocessing against cryptographic
  salting.
\newblock In Leonid Reyzin and Douglas Stebila, editors, {\em CRYPTO~2024,
  Part~IV}, volume 14923 of {\em {LNCS}}, pages 377--411. Springer, Cham,
  August 2024.

\bibitem[DTT10]{C:DeTreTul10}
Anindya De, Luca Trevisan, and Madhur Tulsiani.
\newblock Time space tradeoffs for attacks against one-way functions and
  {PRGs}.
\newblock In Tal Rabin, editor, {\em CRYPTO~2010}, volume 6223 of {\em {LNCS}},
  pages 649--665. Springer, Berlin, Heidelberg, August 2010.

\bibitem[FN91]{STOC:FiaNao91}
Amos Fiat and Moni Naor.
\newblock Rigorous time/space tradeoffs for inverting functions.
\newblock In {\em 23rd ACM STOC}, pages 534--541. {ACM} Press, May 1991.

\bibitem[Gao15]{Gao15}
Jingliang Gao.
\newblock Quantum union bounds for sequential projective measurements.
\newblock {\em Physical Review A}, 92(5), November 2015.

\bibitem[GLLZ21]{TCC:GLLZ21}
Siyao Guo, Qian Li, Qipeng Liu, and Jiapeng Zhang.
\newblock Unifying presampling via concentration bounds.
\newblock In Kobbi Nissim and Brent Waters, editors, {\em TCC~2021, Part~I},
  volume 13042 of {\em {LNCS}}, pages 177--208. Springer, Cham, November 2021.

\bibitem[GLSV21]{EC:GLSV21}
Alex~B. Grilo, Huijia Lin, Fang Song, and Vinod Vaikuntanathan.
\newblock Oblivious transfer is in {MiniQCrypt}.
\newblock In Anne Canteaut and Fran\c{c}ois-Xavier Standaert, editors, {\em
  EUROCRYPT~2021, Part~II}, volume 12697 of {\em {LNCS}}, pages 531--561.
  Springer, Cham, October 2021.

\bibitem[Hel80]{Hellman80}
Martin~E. Hellman.
\newblock A cryptanalytic time--memory trade-off.
\newblock {\em IEEE Transactions on Information Theory}, 26(4):401--406, July
  1980.

\bibitem[HMY23]{EC:HhaMorYam23}
Minki Hhan, Tomoyuki Morimae, and Takashi Yamakawa.
\newblock From the hardness of detecting superpositions to cryptography:
  Quantum public key encryption and commitments.
\newblock In Carmit Hazay and Martijn Stam, editors, {\em EUROCRYPT~2023,
  Part~I}, volume 14004 of {\em {LNCS}}, pages 639--667. Springer, Cham, April
  2023.

\bibitem[HXY19]{AC:HhaXagYam19}
Minki Hhan, Keita Xagawa, and Takashi Yamakawa.
\newblock Quantum random oracle model with auxiliary input.
\newblock In Steven~D. Galbraith and Shiho Moriai, editors, {\em
  ASIACRYPT~2019, Part~I}, volume 11921 of {\em {LNCS}}, pages 584--614.
  Springer, Cham, December 2019.

\bibitem[JLS18]{C:JiLiuSon18}
Zhengfeng Ji, Yi-Kai Liu, and Fang Song.
\newblock Pseudorandom quantum states.
\newblock In Hovav Shacham and Alexandra Boldyreva, editors, {\em CRYPTO~2018,
  Part~III}, volume 10993 of {\em {LNCS}}, pages 126--152. Springer, Cham,
  August 2018.

\bibitem[Liu23]{EC:Liu23}
Qipeng Liu.
\newblock Non-uniformity and quantum advice in the quantum random oracle model.
\newblock In Carmit Hazay and Martijn Stam, editors, {\em EUROCRYPT~2023,
  Part~I}, volume 14004 of {\em {LNCS}}, pages 117--143. Springer, Cham, April
  2023.

\bibitem[LMW24]{STOC:LomMaWri24}
Alex Lombardi, Fermi Ma, and John Wright.
\newblock A one-query lower bound for unitary synthesis and breaking quantum
  cryptography.
\newblock In Bojan Mohar, Igor Shinkar, and Ryan {O'Donnell}, editors, {\em
  56th ACM STOC}, pages 979--990. {ACM} Press, June 2024.

\bibitem[MY22]{C:MorYam22}
Tomoyuki Morimae and Takashi Yamakawa.
\newblock Quantum commitments and signatures without one-way functions.
\newblock In Yevgeniy Dodis and Thomas Shrimpton, editors, {\em CRYPTO~2022,
  Part~I}, volume 13507 of {\em {LNCS}}, pages 269--295. Springer, Cham, August
  2022.

\bibitem[MY24]{TQC:MY24}
Tomoyuki Morimae and Takashi Yamakawa.
\newblock One-wayness in quantum cryptography.
\newblock In {\em 19th Conference on the Theory of Quantum Computation,
  Communication and Cryptography (TQC 2024)}, volume 310 of {\em Leibniz
  International Proceedings in Informatics (LIPIcs)}, pages 4:1--4:21. Schloss
  Dagstuhl -- Leibniz-Zentrum f{\"u}r Informatik, 2024.

\bibitem[Ros26]{Rosenthal22}
Gregory Rosenthal.
\newblock Query and depth upper bounds for quantum unitaries via grover search.
\newblock {\em Quantum}, 10:2144, June 2026.

\bibitem[Unr07]{C:Unruh07}
Dominique Unruh.
\newblock Random oracles and auxiliary input.
\newblock In Alfred Menezes, editor, {\em CRYPTO~2007}, volume 4622 of {\em
  {LNCS}}, pages 205--223. Springer, Berlin, Heidelberg, August 2007.

\bibitem[Yan22]{AC:Yan22}
Jun Yan.
\newblock General properties of quantum bit commitments (extended abstract).
\newblock In Shweta Agrawal and Dongdai Lin, editors, {\em ASIACRYPT~2022,
  Part~IV}, volume 13794 of {\em {LNCS}}, pages 628--657. Springer, Cham,
  December 2022.

\bibitem[Yao90]{STOC:Yao90}
Andrew Chi-Chih Yao.
\newblock Coherent functions and program checkers (extended abstract).
\newblock In {\em 22nd ACM STOC}, pages 84--94. {ACM} Press, May 1990.

\bibitem[YRC20]{YRC20}
Yuxiang Yang, Renato Renner, and Giulio Chiribella.
\newblock Optimal universal programming of unitary gates.
\newblock {\em Physical Review Letters}, 125(21):210501, 2020.

\bibitem[Yue22]{Yuen}
Henry Yuen.
\newblock Lecture 7 from {COMS E6998}: Frontiers of quantum complexity and
  cryptography, 2022.

\bibitem[Zha19]{C:Zhandry19}
Mark Zhandry.
\newblock How to record quantum queries, and applications to quantum
  indifferentiability.
\newblock In Alexandra Boldyreva and Daniele Micciancio, editors, {\em
  CRYPTO~2019, Part~II}, volume 11693 of {\em {LNCS}}, pages 239--268.
  Springer, Cham, August 2019.

\end{thebibliography}

\appendix
\crefalias{section}{appendix}

\section{A matching zero-query attack for 1OWS and 1PRS}
\label{sec:symmetric-subspace-attack}

We give a quantitative version of the folklore symmetric-subspace attack (see~\cite{STOC:LomMaWri24}), which matches our zero-query bounds for both 1OWS (\Cref{thm:1ows-ts-tradeoff}) and 1PRS (\Cref{thm:program-1prs}) up to logarithmic factors.
The idea is to store many copies of each of a fixed set of states and compare the challenge against these states sequentially.
Using more copies reduces false positives, while the quantum union bound controls the disturbance to the single challenge copy.

\begin{proposition}[Symmetric-subspace attack]
\label{prop:symmetric-subspace-attack}
Suppose $K \le N$, and let $1 \le L \le \min\{K,N/32\}$ be an integer. There are zero-query attacks using $32\cdot L^2\cdot\log N$ qubits of advice that achieve 1OWS key-recovery probability at least $3L/4K$
and 1PRS distinguishing advantage at least $L/2K$, respectively.
Both guarantees are averaged over the random oracle.
\end{proposition}

\begin{proof}
Fix distinct keys $k_1,\ldots,k_L$ independently of $R$, and let
$J=\{k_1,\ldots,k_L\}$.
Set $r=32L$, and use the advice state
\[
    \ket{\Phi_R}
    :=
    \bigotimes_{j=1}^{L}
    \ket{\psi_{R,k_j}}^{\otimes r}.
\]
This uses $rL\log_2 N=32L^2\log_2 N$ qubits.
The keys are fixed in advance, so their descriptions need not be
included in the advice.

\paragraph{The symmetric-subspace test.}
Let $\Pi_{\mathrm{sym}}^{(r+1)}$ be the projector onto the totally symmetric subspace of $(\mathbb C^N)^{\otimes(r+1)}$.
For any unit vectors $\ket{\varphi},\ket{\psi}$,
\begin{equation}
\label{eq:symmetric-test-acceptance}
    \left\|
        \Pi_{\mathrm{sym}}^{(r+1)}
        \bigl(\ket{\varphi}\otimes\ket{\psi}^{\otimes r}\bigr)
    \right\|^2
    =
    \frac{1+r|\langle\varphi|\psi\rangle|^2}{r+1}.
\end{equation}
Indeed, the component of $\ket{\varphi}$ parallel to $\ket{\psi}$ is accepted with probability one, while its orthogonal component has squared projection $1/(r+1)$.

On receiving the challenge, the algorithm tests it against the $r$ stored copies for $k_1$, then $k_2$, and so on, stopping at the first accepting outcome.
Each test is the two-outcome projective measurement
$\{\Pi_{\mathrm{sym}}^{(r+1)},
I-\Pi_{\mathrm{sym}}^{(r+1)}\}$ on the challenge and the corresponding stored block. All operations are independent of $R$, and no online oracle queries are needed.

For different keys, the expected squared overlap is $1/N$. The same holds between any fixed stored state and an independent Haar-random challenge. Thus the expected acceptance probability of an individual non-matching test, evaluated on the original input state, is
\[
    a := \frac{1+r/N}{r+1}
    \le \frac{1}{r+1}+\frac{1}{N},
    \qquad
    4\cdot L\cdot a \le \frac14.
\]
The last inequality follows from $r=32L$ and $L\le N/32$.

To account for the disturbance from earlier tests, we use the quantum union bound \cite{Gao15}: for projectors $Q_1,\ldots,Q_h$ and a unit vector $\ket{v}$,
\[
    \left\|Q_h\cdots Q_1\ket{v}\right\|^2
    \ge
    1-4\sum_{i=1}^{h}
    \left\|(I-Q_i)\ket{v}\right\|^2.
\]

\paragraph{For the 1OWS game.}
Output the key corresponding to the first accepting test, and output an arbitrary key if all tests reject. Suppose the challenge key is $k_j$. On the original input, the test for $k_j$ accepts with probability one, while every earlier test has expected acceptance probability $a$. Applying the quantum union bound to the desired sequence ``reject $k_1,\ldots,k_{j-1}$, then accept $k_j$'' gives
\[
    \mathbb E_R
    \Pr[\text{output }k_j \mid k=k_j]
    \ge 1-4(j-1)a
    \ge \frac34.
\]
Since the challenge key lies in $J$ with probability $L/K$, the overall key-recovery probability is at least $3L/4K$.

\paragraph{For the 1PRS game.}
Use the same sequence of tests, declaring the challenge real if any test accepts, and Haar-random otherwise.
Let $p_{\mathrm{in}}$ be the probability of declaring real conditioned on a real challenge with key in $J$, and let $p_{\mathrm H}$ be the probability of declaring real on a Haar-random challenge.
These probabilities include the average over $R$.
The preceding argument gives $p_{\mathrm{in}}\ge 3/4$.
Applying the quantum union bound to the sequence in which all
tests reject gives
\[
    p_{\mathrm H}\le 4La\le \frac14.
\]

It remains to account for real challenges whose keys are not
stored.
Conditioned on the oracle rows indexed by $J$, every such
challenge is independent of the advice and has average density
matrix $I/N$, exactly as in the Haar experiment.
Thus these challenges contribute the same acceptance probability
as the Haar baseline.
Writing $p_{\mathrm{real}}$ for the acceptance probability in
the real experiment, we obtain
\[
    p_{\mathrm{real}}
    =
    \frac{L}{K}p_{\mathrm{in}}
    +
    \left(1-\frac{L}{K}\right)p_{\mathrm H},
\]
and hence
\[
    p_{\mathrm{real}}-p_{\mathrm H}
    =
    \frac{L}{K}
    \bigl(p_{\mathrm{in}}-p_{\mathrm H}\bigr)
    \ge \frac{L}{2K}.
\]
This also lower bounds the distinguishing advantage averaged
over $R$, since the average absolute gap is at least the
absolute average gap.
\end{proof}

For $N\ge 32$ and an advice budget $S\ge 32\log_2 N$, take
\[
    L
    =
    \min\left\{
        K,\,
        \left\lfloor\frac{N}{32}\right\rfloor,\,
        \left\lfloor
            \sqrt{\frac{S}{32\log_2 N}}
        \right\rfloor
    \right\}.
\]
Since $K\le N$, both attacks achieve
\[
    \Omega\left(
        \min\left\{
            1,\,
            \frac{1}{K}\sqrt{\frac{S}{\log N}}
        \right\}
    \right).
\]
In particular, constant success probability or distinguishing
advantage is achievable with $O(K^2\log N)$ advice qubits and
no online queries.

\section{Proof of \Cref{lem:1ows-classical-controlled-unitary}}
\label{sec:proof-of-classical-controlled-unitary}

Now we prove \cref{lem:1ows-classical-controlled-unitary}. We first construct the
unitary for a single key, and then apply this construction after un-compressing
all the other rows.

\begin{lemma}[Single-key database estimate]
\label{lem:1ows-classical-single-row}
Let $\mathsf D_0$ be a single-key database register with basis
$\{\ket E:E\subseteq[N]\}$. For a normalized vector on $\mathsf D_0$
\[
\ket\beta=\sum_{E\subseteq[N]}\beta_E\ket E,
\qquad
\nu:=\sum_E|E|\,|\beta_E|^2,
\]
define $T_x\ket E:=\ket{E\oplus\{x\}}$.
There is a unitary $W_\beta$ acting only on $\mathsf D_0\mathsf B$ such that
\begin{align}
&\norm{W_\beta\left(\frac1{\sqrt N}\sum_x\ket x\otimes 
                 T_x\ket\beta_{\mathsf D_0}\ket0_{\mathsf B}\right)
       -\frac1{\sqrt N}\sum_x\ket x\otimes \ket\beta_{\mathsf D_0}\ket x_{\mathsf B}}
\le(1+\sqrt3)\sqrt{\frac\nu N}.
\label{eq:1ows-classical-single-row-bound}
\end{align}
Here $\nu$ is the expected size the database, and $W_\beta$ acts as identity on the $\ket{x}$ register.
\end{lemma}

\begin{proof}
Define $\mathcal H_{\mathsf D_0}$ as the space of register $\mathsf D_0$. Define two linear maps $M_+,M_-:\mathbb C^N\to\mathcal H_{\mathsf{D}_0}$ by
\[
M_+\ket x:=\sum_{E:x\notin E}\beta_E\ket{E\cup\{x\}},
\qquad
M_-\ket x:=\sum_{E:x\in E}\beta_E\ket{E\setminus\{x\}}.
\]
Thus $(M_++M_-)\ket x=T_x\ket\beta$. For any matrix $M$, write
$\norm{M}_{\mathrm{HS}}^2:=\operatorname{Tr}(M^\dagger M)$ for the sum of its squared
column norms. For a fixed $x$, distinct sets containing $x$ remain distinct after
deleting $x$, and thus
\begin{align}
\label{eq:1ows-classical-deletion-mass}
\norm{M_-}_{\mathrm{HS}}^2
=\sum_x\sum_{E:x\in E}|\beta_E|^2
=\sum_E|E|\,|\beta_E|^2=\nu.
\end{align}

\paragraph{Replacing by orthonormal vectors from the $M_+$ part.} Define the $N\times N$ Gram matrices
$G_+:=M_+^\dagger M_+$ and $G_-:=M_-^\dagger M_-$.
They satisfy the following:
\begin{align*}
    (G_-)_{xx}=\sum_{E:x\in E}|\beta_E|^2,
    \qquad
    (G_+)_{xx}=1-(G_-)_{xx},\\
    (G_+)_{xy}
=\sum_{E\subseteq[N]\setminus\{x,y\}}
       \overline{\beta_{E\cup\{y\}}}\,\beta_{E\cup\{x\}}
=(G_-)_{yx}.
\end{align*}
Therefore, for positive semidefinite $G_-^{\mathsf T}$ and $\operatorname{diag}(G_-)$, and trace norm $\norm{\cdot}_1$, 
\begin{align*}
    \norm{G_+-I_N}_1&=\norm{G_-^{\mathsf{T}}-I_N+(\operatorname{diag}(G_+)-\operatorname{diag}(G_-))}_1=\norm{G_-^{\mathsf{T}}-2\operatorname{diag}(G_-)}_1\\
    &\le \norm{G_-^{\mathsf{T}}}_1 + 2\norm{\operatorname{diag}(G_-)}_1\le 3\nu
\end{align*}

Choose an isometry $V:\mathbb C^N\to\mathcal H_{\mathrm{D}_0}$ satisfying $M_+=V\sqrt{G_+}$. That is, for an orthonormal eigenbasis $\{\ket{u_j}\}_{j=1}^N$ of $G_+$ with eigenvalues $\lambda_j\ge0$, set $V\ket{u_j}=M_+\ket{u_j}/\sqrt{\lambda_j}$ when $\lambda_j>0$.
These vectors are orthonormal because their inner products are given by the matrix $G_+$ in its eigenbasis. For the zero eigenvalues, extend them to an orthonormal family of $N$ vectors (this is possible since
$\dim\mathcal H_{\mathrm{D}_0}=2^N\ge N$).
Then $V^\dagger V=I_N$, and since $\lambda_j>0$,
\begin{align*}
\norm{M_+-V}_{\mathrm{HS}}^2
&=\sum_{j=1}^N(\sqrt{\lambda_j}-1)^2
\le\sum_{j=1}^N|\lambda_j-1|\\
&=\norm{G_+-I_N}_1\le3\nu.
\end{align*}
Together with \cref{eq:1ows-classical-deletion-mass}, this gives
\begin{align}
\label{eq:1ows-classical-orthonormal-approximation}
\norm{M_++M_--V}_{\mathrm{HS}}
\le\norm{M_-}_{\mathrm{HS}}+\norm{M_+-V}_{\mathrm{HS}}
\le(1+\sqrt3)\sqrt\nu.
\end{align}

\paragraph{Constructing $W_\beta$.}
The two families $\{V\ket x\ket0:x\in[N]\}$ and
$\{\ket\beta\ket x:x\in[N]\}$ are orthonormal in
$\mathcal H_{\mathrm{D}_0}\otimes\mathcal H_{\mathsf B}$, where $\mathcal H_{\mathsf B}$ is the space of register $\mathsf B$.
Complete them to orthonormal bases and map the first basis to the second.
This defines one unitary $W_\beta$ satisfying, simultaneously for all $x$,
\[
W_\beta(V\ket x\ket0)=\ket\beta\ket x.
\]
Thus $W_\beta$ moves the orthonormal encoding of $x$ into $\mathsf B$ and restores
$\ket\beta$ on the row. By orthogonality of the labels in $\mathsf X$ and unitarity,
the square of the norm in \cref{eq:1ows-classical-single-row-bound} is
\begin{align*}
&\frac1N\sum_x
\norm{W_\beta(T_x\ket\beta\ket0)-W_\beta(V\ket x\ket0)}^2\\
&\quad=\frac1N\sum_x\norm{T_x\ket\beta-V\ket x}^2
=\frac1N\norm{M_++M_--V}_{\mathrm{HS}}^2
\le\frac{(1+\sqrt3)^2\nu}{N}.
\end{align*}
\end{proof}

From the single-key database case in \cref{lem:1ows-classical-single-row}, now we are ready to prove \cref{lem:1ows-classical-controlled-unitary}

\begin{proof}[Proof of \cref{lem:1ows-classical-controlled-unitary}]
Fix $k$. As we would want $[W_k,T_{a,y}]=0$, we would need to go back to the uncompressed mode. Write $\mathsf D_k$ for the key $k$ part and $\mathsf D_{-k}$ for all other parts, and put $D_k:=\{x:(k,x)\in D\}$. On the parts outside of $k$, use the orthonormal basis
\[
\ket s_{-k}:=
\bigotimes_{\substack{a\in[K]\setminus\{k\}\\ y\in[N]}}
\frac{\ket0_{a,y}+{s(a,y)}\ket1_{a,y}}{\sqrt2}
=
\bigotimes_{\substack{a\in[K]\setminus\{k\}\\ y\in[N]}}
H\ket{(1-s(a,y))/2}.
\]
Here $s:([K]\backslash\{k\})\times[N]\to\{\pm 1\}$ ranges over all assignments (as for random oracle $R_{-k}$), and $\ket{0}_{a,y},\ket{1}_{a,y}$ mean absence or presence of a compressed database entry $(a,y)$. This is the uncompressed outside-$k$ basis written in compressed coordinates, while the key $k$ part remains compressed.
The two single-coordinate vectors are orthonormal, so $\sum_s\ket s\bra s_{-k}=I_{-k}$.

Define the unnormalized vectors on a single key (while fixing the oracle $R_{-k}$ on other keys):
\[
\ket{\beta_{k,s}}_k:=(\bra s_{-k}\otimes I_k)\ket\chi
=\sum_{E\subseteq[N]}\beta_{k,s,E}\ket E_k.
\]
Here $\ket E_k$ contains exactly the entries $E$ with key $k$, and $I_{-k},I_k$ are the identities on the indicated registers. 
For each nonzero $\ket{\beta_{k,s}}$, apply \cref{lem:1ows-classical-single-row} to $\ket{\beta_{k,s}}/\norm{\ket{\beta_{k,s}}}$, obtaining a unitary $W_{k,s}$ on $\mathsf D_k\mathsf B$. Squaring and using orthogonality of register $\ket{x}$,
\begin{equation}
\label{eq:1ows-classical-branch-error}
\sum_{x\in[N]}
\norm{W_{k,s}(T_x\ket{\beta_{k,s}}\ket0)
      -\ket{\beta_{k,s}}\ket x}^2
\le (1+\sqrt3)^2\sum_E |E|\,|\beta_{k,s,E}|^2.
\end{equation}
Here $T_x\ket E_k=\ket{E\oplus\{x\}}_k$, and the displayed registers inside the norm are $\mathsf D_k,\mathsf B$. When $\beta_{k,s}=0$, choose $W_{k,s}=I$, and the estimate remains valid.

Now set
\[
W_k:=\sum_s\ket s\bra s_{-k}\otimes W_{k,s},
\qquad
W:=\sum_{k\in[K]}W_k\otimes\ket k\bra k_{\mathsf K}.
\]
Both are unitary because their blocks are unitary. For $a\ne k$,
$T_{a,y}\ket s_{-k}={s(a,y)}\ket s_{-k}$, so $T_{a,y}$ is scalar on each
$s$-block and commutes with $W_k$. Also, $W$ leaves the algorithm's registers
untouched, hence $[W,U]=0$. Since $\Pi$ is diagonal in $k$ and each $\Pi_k$
acts only on the algorithm's registers, $[W,\Pi]=0$.

Finally, expand both initial states. The initial state on the database register can be written as 
\[
\ket\chi=
\left(\sum_s\ket s\bra s_{-k}\otimes I_k\right)\ket\chi
=\sum_s\ket s_{-k}\ket{\beta_{k,s}}_k.
\]
For each $E$, orthogonality gives
\[
\sum_s|\beta_{k,s,E}|^2
=\norm{(I_{-k}\otimes\bra E_k)\ket\chi}^2
=\sum_{D:D_k=E}|\alpha_D|^2.
\]
The labels $k,s,x$ are orthogonal, and the fixed algorithm state $\ket z$ has norm one. Therefore, by \cref{eq:1ows-classical-branch-error} and $\mu=\sum_D|\alpha_D|^2|D|$,
\begin{align*}
&\norm{W\ket{\xi_{\mathrm{real}}}-\ket{\xi_{\mathrm{ref}}}}^2\\
&\quad=\frac1{KN}\sum_{k,s,x}
\norm{W_{k,s}(T_x\ket{\beta_{k,s}}\ket0)
      -\ket{\beta_{k,s}}\ket x}^2\\
&\quad\le \frac{(1+\sqrt3)^2}{KN}
       \sum_{k,s,E}|E|\,|\beta_{k,s,E}|^2\\
&\quad=\frac{(1+\sqrt3)^2}{KN}
       \sum_D|\alpha_D|^2\sum_k|D_k|
=\frac{(1+\sqrt3)^2\mu}{KN}.
\end{align*}
\end{proof}

\begin{remark}[Why this does not extend to the quantum advice case?]
The construction uses the fact that $\ket\chi$ is on the database register alone, while the initial algorithm state $\ket z$ is fixed (and unentangled with $\ket{\chi}$). Thus for each $\ket{\beta_{k,s}}$ as a state on key $k$ alone, $W_k$ can be chosen entirely on inaccessible registers (the database register and $\mathsf B$).

For a general initial state entangled between quantum advice and the database, say $\ket{\psi}:=\sum_D\alpha_D\ket{\varphi_D}_{\mathsf A}\ket{D}_{\mathsf D}$,
the corresponding state on a single $k$ can also contain the advice register $\mathsf A$. Therefore, the same orthonormal-family construction only guarantees a unitary $V'$ on $\mathsf{AD_0B}$, not on $\mathsf{D_0 B}$ alone. This will result in a $W$ that also acts on register $\mathsf A$, and thus cannot commute with the operations on the algorithm's registers as $U$ and $\Pi$. Thus this proof does not extend to arbitrary quantum advice.
\end{remark}

\section{Proof of \cref{lem:prg-l-database-t-query}}
\label{sec:append-prg}

\begin{proof}[Proof of \cref{lem:prg-l-database-t-query}]
From \cref{eq:prg-goal-advantage}, we know 
\begin{align*}
    Y_R:=\mathbb{E}_x
    \bra{R(x),0^w}(U^\dagger O_R^\dagger)^T \Pi (O_R U)^T\ket{R(x),0^w} 
    - 
    \mathbb{E}_r
    \bra{r,0^w}(U^\dagger O_R^\dagger)^T \Pi (O_R U)^T\ket{r,0^w}.
\end{align*}
Throughout this section, we suppress the fixed workspace initialization $\ket{0^w}$. Operators $Y_R$ are on the advice register. The workspace is included in all evolving hybrid states, and suppressing its notation does not affect the hybrid estimates.

Now we add a resampling hybrid in between, by resampling position $x$ on the random $R$ (or, swapping position $x$ on $R$ with a random image).
Let $R^{x\to r}$ denote $R$ with its value at $x$ replaced by $r$. That is, now $R^{x\to r}(x)=r$, and for any $x'\neq x$, $R^{x\to r}(x')=R(x)$. Define 
\begin{align*}
\Gamma_R
&:=
\mathbb E_x \bra{R(x)}(U^\dagger O_R^\dagger)^T \Pi (O_R U)^T\ket{R(x)}
-
\mathbb E_{x,r} \bra{r}(U^\dagger O_{R^{x\to r}}^\dagger)^T \Pi (O_{R^{x\to r}} U)^T\ket{r},\\
B_R
&:=
\mathbb E_{x,r} \bra{r}(U^\dagger O_{R^{x\to r}}^\dagger)^T \Pi (O_{R^{x\to r}} U)^T\ket{r}
-
\mathbb E_r \bra{r}(U^\dagger O_R^\dagger)^T \Pi (O_R U)^T\ket{r}.
\end{align*}
Now $Y_R=\Gamma_R+B_R$. By linearity, in the compressed oracle framework,
$\widehat Y=\widehat\Gamma+\widehat B$. Therefore, we only need to bound $\bra{\psi}\widehat{\Gamma}\ket{\psi}$ and $\bra{\psi}\widehat{B}\ket{\psi}$ respectively.

\paragraph{For the $\widehat\Gamma$ part.} In the real mode, for each $x$, define the $x$ part of the operator for $\Gamma_R$ as 
\[
\Gamma_R(x):=\bra{R(x)}(U^\dagger O_R^\dagger)^T \Pi (O_R U)^T\ket{R(x)}
-
\mathbb E_{r} \bra{r}(U^\dagger O_{R^{x\to r}}^\dagger)^T \Pi (O_{R^{x\to r}} U)^T\ket{r},
\]
and thus $\Gamma_R=\frac{1}{N}\sum_x \Gamma_R(x)$, and $\Gamma_R(x)$ taking expectation over $R(x)$ will be equal to 0.

In the compressed oracle database mode, Let $Q_x$ be the projector onto databases where $x$ is contained, $(I-Q_x)$ as the part where $x$ is absent, and they act as the identity of the algorithm's register. We can correspondingly define $\widehat \Gamma_x$ as $\Gamma_R(x)$ in the compressed oracle framework, and 
\[
\left\|\widehat \Gamma_x\right\|\le 2, \qquad (I-Q_x)\widehat\Gamma_x(I-Q_x)=0.
\]
This is from the fact that, in the uncompressed representation, $(I-Q_x)$ projects the oracle value register at $x$ onto its uniform superposition, and thus it averages $\Gamma_{R}(x)$ over that coordinate $R(x)$.
Now for every $x\in[N]$,
\begin{align*}
\left|\bra\psi\widehat \Gamma_x\ket\psi\right|
&=
\left|
\bra\psi Q_x\widehat \Gamma_x\ket\psi
+
\bra\psi (I-Q_x)\widehat \Gamma_x Q_x\ket\psi
\right|
\le 2\cdot 2\left\|Q_x\ket\psi\right\|.
\end{align*}
Since $\sum_x Q_x$ will count the database size, on $\ket{\psi}$ supported on databases of size $\le\ell$, $\sum_x\|Q_x|\psi\rangle\|^2\le\ell$. By Cauchy-Schwarz,
\[
\left|\bra\psi\widehat\Gamma\ket\psi\right|
\le 
\frac{1}{N}\sum_x\left|\bra\psi\widehat\Gamma_x\ket\psi\right|
\le
\frac4N\sum_x\left\|Q_x\ket\psi\right\|
\le
4\sqrt{\frac\ell N}.
\]

\paragraph{For the $\widehat B$ part.}
Fix $r$. Define $\cO_x$ from the compressed oracle $\cO$ except that, at query input $x$, it applies the response corresponding to $r$ without changing
the database. Let $\ket{v_i}$ and $\ket{v_i^x}$ be the states just before the $(i+1)$-th query using $\cO$ and $\cO_x$, respectively, starting from the same state $\ket r\ket\psi$. Now our goal is to upper bound $\left|\mathbb E_x\bra{v_T^x}\Pi\ket{v_T^x}-\bra{v_T}\Pi\ket{v_T} \right|$, which can be bounded 
\begin{align*}
\left|
\mathbb E_x\bra{v_T^x}\Pi\ket{v_T^x}
-\bra{v_T}\Pi\ket{v_T}
\right|
&\le \left|2\Re\bra{v_T}\Pi\mathbb{E}_x(\ket{v_T^x}-\ket{v_T})+\mathbb{E}_x(\bra{v_T^x}-\bra{v_T})\Pi(\ket{v_T^x}-\ket{v_T})\right|
\\
&\le 2\left\|\mathbb{E}_x(\ket{v_T^x}-\ket{v_T})\right\|+\mathbb{E}_x\left\|\ket{v_T^x}-\ket{v_T}\right\|^2.
\end{align*}

Consider induction on queries, the difference can be viewed as
\begin{align*}
    \left\|\ket{v_{i+1}^x}-\ket{v_{i+1}}\right\|
    &=
    \left\|
    \cO_x (\ket{v_i^x}-\ket{v_i})
    +
    (\cO_x-\cO)\ket{v_i}
    \right\|.
\end{align*}
As $\cO_x-\cO$ only acts on query address $x$ and has norm at most $2$, $\sum_x \left\|(\cO_x-\cO)\ket{v_i}\right\|^2\le 4$, and the triangle inequality gives
\begin{align*}
\sqrt{\mathbb E_x
\|\ket{v_{i+1}^x}-\ket{v_{i+1}}\|^2}
&\le
\sqrt{\mathbb E_x
\|\ket{v_i^x}-\ket{v_i}\|^2}
+\frac2{\sqrt N}
\end{align*}
Since the initial states $\ket{v_i^0}=\ket{v_i}$, this gives
\begin{align}\label{ineq:prg-state-difference}
    \mathbb E_x\|\ket{v_i^x}-\ket{v_i}\|^2\le\frac{4i^2}{N}.
\end{align}

For the average state difference, we can rewrite the difference as
\begin{align*}
    \left\|\ket{v_{i+1}^x}-\ket{v_{i+1}}\right\|
    &=
    \left\|
    \cO(\ket{v_{i}^x}-\ket{v_{i}})
    +
    (\cO_x-\cO)(\ket{v_i^x}-\ket{v_i})
    +
    (\cO_x-\cO)\ket{v_i}
    \right\|,
\end{align*}
and therefore, by the triangle inequality,
\begin{align*}
    \left\|\mathbb{E}_x(\ket{v_{i+1}^x}-\ket{v_{i+1}})\right\|
    &=
    \left\|
    \mathbb{E}_x \cO(\ket{v_{i}^x}-\ket{v_{i}})
    +
    \mathbb{E}_x(\cO_x-\cO)(\ket{v_i^x}-\ket{v_i})
    +
    \mathbb{E}_x(\cO_x-\cO)\ket{v_i}
    \right\|
    \\
    &\le 
    \left\|
    \mathbb{E}_x (\ket{v_{i}^x}-\ket{v_{i}})
    \right\|
    +
    \left\|
    \mathbb{E}_x(\cO_x-\cO)(\ket{v_i^x}-\ket{v_i})
    \right\|
    +
    \left\|
    \mathbb{E}_x(\cO_x-\cO)\ket{v_i}
    \right\|
    \\
    &\le 
    \left\|
    \mathbb{E}_x (\ket{v_{i}^x}-\ket{v_{i}})
    \right\|
    +
    \sqrt{\frac{4}{N}\cdot \mathbb{E}_x \left\|\ket{v_i^x}-\ket{v_i}\right\|^2}
    +
    \left\|
    \mathbb{E}_x(\cO_x-\cO)
    \right\|
    \\
    &\le 
    \left\|
    \mathbb{E}_x (\ket{v_{i}^x}-\ket{v_{i}})
    \right\|
    +
    \frac{4i}{N}
    +
    \frac{2}{N},
\end{align*}
where the last line is from \cref{ineq:prg-state-difference}, and as $\{\cO_x-\cO\}_x$ are orthogonal for different $x$. This gives
\begin{align}\label{ineq:prg-average-state-difference}
    \left\|\mathbb E_x(\ket{v_T^x}-\ket{v_T})\right\|
    \le\frac{2T^2}{N}.
\end{align}

From \cref{ineq:prg-state-difference,ineq:prg-average-state-difference}, we can upper bound the $\widehat B$ part by
\begin{align*}
    \left|
    \mathbb E_x\bra{v_T^x}\Pi\ket{v_T^x}
    -\bra{v_T}\Pi\ket{v_T}
    \right|
    \le 2\left\|\mathbb{E}_x(\ket{v_T^x}-\ket{v_T})\right\|+\mathbb{E}_x\left\|\ket{v_T^x}-\ket{v_T}\right\|^2
    \le \frac{8T^2}{N}.
\end{align*}

Together with the bound for the $\widehat\Gamma$ part, this proves the lemma.
\end{proof}

\subsection{A weaker but simpler bound}
In this subsection, we also show a weaker bound of $O\left(\sqrt{(T^2+\ell)/N}\right)$ for \cref{lem:prg-l-database-t-query}, which can be derived directly from \cite{FOCS:CGLQ20}.

For normalized state $\ket{\psi}$ (with both the algorithm's register and the database register), by our definition of $\widehat Y$ in the compressed oracle database mode:
\begin{align*}
    &\left| \bra{\psi}\widehat Y\ket{\psi} \right|\\
    &= \left|\bra{\psi}\left(
    \frac1M\cdot\mathbb{E}_x \sum_{y_1,y_2}\bra{y_1}H (U^\dagger \cO^\dagger)^T \Pi (\cO\;U)^TH\ket{y_2}\otimes T_{x,(y_1\oplus y_2)}
    -
    \mathbb{E}_r
    \bra{r}(U^\dagger \cO^\dagger)^T \Pi (\cO\; U)^T\ket{r}\otimes I_D
    \right)\ket{\psi}\right|\\
    &=\left|
    \mathbb{E}_x
    \left\|
         \Pi(\cO\; U)^T \left(\sum_y \frac{1}{\sqrt{M}} \cdot H\ket{y}\otimes T_{x,y}\right) \ket{\psi}
    \right\|^2
    -
    \mathbb{E}_r
    \left\|
    \Pi(\cO\; U)^T\ket{r}\ket{\psi}
    \right\|^2
    \right|.
\end{align*}
This is exactly the distinguishing advantage of the following  distinguishing game:

\begin{enumerate}
    \item The algorithm starts with a normalized state $\ket{\psi}$ with support on databases of size $\le \ell$.
    \item The algorithm will receive a challenge input, from either $R(x)$ for a random $x\gets [N]$ or a completely random $r\gets [M]$.
    \item Then the algorithm is allowed to make $T$ queries to the oracle $R$ (or $O_R$), and it needs to tell whether the challenge is from $R$ or a completely random string from range.
\end{enumerate}
Therefore, now we only need to prove the maximum distinguishing advantage of this game, and this will give a proof for \cref{lem:prg-l-database-t-query}. This has been studied in \cite[Lemma~5.9]{FOCS:CGLQ20}.

From the above interpretation, now our goal is to upper bound the distinguishing advantage of the game starting with $\ket{\psi}$ supported on databases of size $\le\ell$.
\begin{align*}
    &\left|\Pr_{R,x}[A^{O_R}(\ket{\psi},R(x))=1]-\Pr_{R,y}[A^{O_R}(\ket{\psi},y)=1]\right|\\
    &\le \left|\sqrt{\Pr_{R,x}[A^{O_R}(\ket{\psi},R(x))=1]} - \sqrt{\Pr_{R,y}[A^{O_R}(\ket{\psi},y)=1]}\right| 
    \left|\sqrt{\Pr_{R,x}[A^{O_R}(\ket{\psi},R(x))=1]} + \sqrt{\Pr_{R,y}[A^{O_R}(\ket{\psi},y)=1]}\right|
    \\
    &\le \frac{T+\sqrt{\ell}}{\sqrt{N}}\cdot 2\sqrt{2}
    \le 
    4\cdot \sqrt{\frac{T^2+\ell}{N}}.
    \tag{\cite[Lemma~5.9]{FOCS:CGLQ20}}
\end{align*}

\end{document}